\documentclass[11pt,reqno]{amsart}
\usepackage{amsmath}
\usepackage{array}

\usepackage{mathabx}

\newtheorem{theorem}{Theorem}
\newtheorem{proposition}[theorem]{Proposition}
\newtheorem{lemma}[theorem]{Lemma}
\newtheorem{corollary}[theorem]{Corollary}
\newtheorem{coro}{Corollary}
\theoremstyle{definition}
\newtheorem{definition}[theorem]{Definition}

\theoremstyle{remark}
\newtheorem{remark}[theorem]{Remark}
\input{amssym.def}
\input{amssym}
\def\L{{\mathcal L}}

\def\B{{\mathcal B}}

\def\H{{\mathcal H}}
\def\P{{\mathcal P}}

\def\V{{\mathcal V}}
\def\W{{\mathcal W}}
\def\G{{\mathcal G}}  
\def\F{{\mathcal F}} 
\def\F{{\mathcal F}} 
\def\X{{\mathcal X}}

\def\J{{\mathcal J}}

\def\R{\Bbb R} 
\def\Z{\Bbb Z} 
\def\Na{\Bbb N}
\def\T{\Bbb T}  
\def\C{\Bbb C}

\def\He{\Bbb H}
\def\A{{\mathcal A}}

\def\la{\langle}
\def\be{\begin{equation}}
\def\ee{\end{equation}}
\def\bea{\begin{eqnarray}}
\def\eea{\end{eqnarray}}
\def\ra{\rangle}
\def\ds{\displaystyle}
\def\om{\omega}

\def\ep{\varepsilon}

\newcommand{\truc}{\widetilde}

\newcommand{\intm}{\int_{\R^m}}
\newcommand{\intdm}{\int_{\R^{2m}}}
\newcommand{\sumql}{\sum_{q\in\Z^l}}

\newcommand{\sumqql}{\sum_{q'\in\Z^l}}
\newcommand{\pom}{{p\cdot\omega}}
\newcommand{\pomv}{{\vert\pom\vert}}
\newcommand{\bid}{{\underline\omega}}
\newcommand{\bido}{{\bid}}
\newcommand{\chos}{{A}}
\newcommand{\rok}{{\rho,\bid,k}}
\newcommand{\machin}{{\mathcal M}}

\makeatletter

\@addtoreset{equation}{section}
\makeatother

\usepackage[hyperindex=true, 
					colorlinks=false]{hyperref}

\begin{document}
\baselineskip=18pt
\title{Quantum singular complete integrability}
\author{Thierry Paul}
\address{CNRS and Centre de Math\'ematiques Laurent Schwartz, 
\'Ecole Polytechnique, 91128 Palaiseau Cedex Fran\-ce}
\email{thierry.paul@polytechnique.edu}
\author{Laurent Stolovitch}
\address{CNRS and Laboratoire J.-A. Dieudonné\'e
Universit\'e de Nice - Sophia Antipolis
Parc Valrose
06108 Nice Cedex 02
France
}
\email{Laurent.STOLOVITCH@unice.fr}
\thanks{Research of L. Stolovitch was supported by ANR grant ``ANR-10-BLAN 0102'' for the project DynPDE
}

\date{}

\begin{abstract}
We consider some perturbations of a family 
 of pairwise commuting linear quantum Hamiltonians on the torus with possibly dense pure point spectra. We prove that the Rayleigh-Schr\"odinger perturbation series converge near each unperturbed eigenvalue under the form of a convergent quantum Birkhoff normal form. Moreover the family is jointly diagonalised by a common unitary operator explicitly constructed by a Newton type algorithm. This leads  to the fact that the spectra of the family remain pure point. The results are uniform in the Planck constant near $\hbar= 0$. The unperturbed frequencies satisfy a 
small divisors condition 
and we explicitly estimate how this condition can be released when the family tends to the unperturbed one. In the case where the number of operators is equal to the number of degrees of freedom - i.e. full integrability - our construction provides convergent normal forms for general perturbations of linear systems.
\end{abstract}
\maketitle

\small
\tableofcontents
\large

\section{Introduction}\label{intro}

Perturbation theory 
belongs to the history of quantum mechanics, and even to its pre-history, as it was used before the works of Heisenberg and Schr\"odinger in 1925/1926. The goal  at that time was to understand what should be the Bohr-Sommerfeld quantum conditions for systems nearly integrable \cite{MB}, by quantizing the perturbation series provided by celestial mechanics \cite{HP}. After (or rather during its establishment) the functional analysis point of view was settled for quantum mechanics, the ``modern" perturbation theory took place, mostly by using the Neumann expansion of the perturbed  resolvent, providing efficient and rigorous ways of establishing the validity of the Rayleigh-Schr\"odinger expansion and leading to great success of this method, in particular the convergence under a simple argument of size of the perturbation in the topology of operators  on Hilbert spaces \cite{TK}, and
Borel summability for (some) unbounded perturbations \cite{GG,BS}. On the other hand, by relying on the comparison between the size of the perturbation and the distance between consecutive unperturbed eigenvalues, the method has two inconveniences: it remains local in the spectrum in the (usual in dimension larger than one) case of spectra accumulating at infinity and is even  inefficient in the case of dense point unperturbed spectra which can be  the case  in the present article. 
\vskip 1cm
In the present article, we consider some commuting families of operators on $L^2(\T^d)$ close to a commuting family of unperturbed Hamiltonians whose spectra are pure point and might be dense for all values of $\hbar$. As already emphasized, standard (Neumann series expansion) perturbation theory does not apply in this context. Nevertheless, we prove that the pure point property is preserved and moreover, we show that the perturbed spectra are analytic functions of the unperturbed ones. All these results are obtained using a method inspired by classical local dynamics, namely the analysis of quantum Birkhoff forms. Let us first recall some known fact of (classical) Birkhoff normal forms.
\vskip 1cm
In the framework of (classical) local dynamics, R\"ussman proved in \cite{Ru} (see also \cite{bruno}) 
the remarkable result which says that, when the Birkhoff normal form (BNF), at any order, depends only on the unperturbed Hamiltonian, then it converges provided that the small divisors of the unperturbed Hamiltonian do not accumulate the origin too fast (we refer to \cite{Ar2} for an introduction to this subject). This leads to the integrability of the perturbed system.
On the other hand, Vey proved two theorems about the holomorphic normalization of families of $l-1$ (resp. $l$) of
commuting germs of holomorphic vector fields, volume preserving (resp. Hamiltonian) in a neighborhood of the origin of $\C^l$ (resp. $\C^{2l}$) (and vanishing at the origin) with diagonal and independent 1-jets \cite{JV1,JV2}.


These results were extended by one of us in \cite{LS1, LS2}, in the framework of general local dynamics of a families of of $1\leq m\leq l$ commuting germs of holomorphic vector fields near a fixed point. It is proved that under an assumption on the formal (Poincar\'e) normal form of the family and and under a generalized Brjuno type condition of the family of linear parts, there exists an holomorphic transformation of the family to a normal form. This fills up therefore the gap between R\"ussman-Brjuno and the complete integrability of Vey. In these directions, we should also mention works by H. Ito \cite{ito} and N.-T. Zung \cite{Zu} in the analytic case and H. Eliasson \cite{eliasson-vey} in the smooth case, and Kuksin-Perelman \cite{kuksin-vey} for a specific infinite dimensional version.
 
 In \cite{GP} one of us (the other) gave with S. Graffi a quantum version of the R\"ussmann theorem in the framework of perturbation theory of the quantization of linear vector fields on the torus $\T^l$. Moreover, in this setting, it is  possible to read on the original perturbation if the R\"ussman condition is satisfied and the results are uniform in the Planck constant belonging to $[0,1]$. The method seats in the framework of 
 Lie method perturbation theory initiated  in  classical mechanics in \cite{De,Ho} 
 and uses  the quantum setting established in \cite{BGP}. 
 \vskip .5cm
 The goal of the present paper 
 is 
 to provide a full spectral resolution for certain families of  commuting quantum Hamiltonians, not treatable by standard methods due to possible spectral accumulation, through the convergence of quantum normal Birkhoff forms and underlying unitary transformations. These families generalize the quantum version of R\"usmman theorem treated in \cite{GP}, to the quantum version of ``singular complete integrability" treated in \cite{LS1}. The methods use the quantum version of the Lie perturbative algorithm together with a newton type scheme in order to overcome the difficulty created by  small divisors.

\vskip 1cm 
Let $m\leq l\in\mathbb N^*$. 
 For 
 $\om=(\om_i)_{i=1\dots  m}$ with $\om_i=(\omega_i^j)_{j=1\dots l}\in\R^l$, let us  denote by $L_\om=(L_{\om_i})_{i=1\dots m}$,  the  
 operator valued vector
  of components 
$$L_{\om_i}=-i\hbar\om_i.\nabla_x=-i\hbar\sum_{j=1}^l\om_{i}^{j}\frac{\partial}{\partial x_j} ,\ i=1\dots m$$ on $L^2(\T^l)$.
 
We define the 
 operator valued vector
$H=(H_i)_{i=1\dots m}$ by\be\label{deih}
H=L_\omega+V,
\ee
where $V$ is a
bounded  operator valued vector
on $L^2(\T^l)$ whose action is defined after 
a
function  
$
\V:\ (x,\xi,\hbar)\in T^*\T^l\times[0,1]\mapsto\V(x,\xi,\hbar)\in\R^m
$ 
by the formula (Weyl quantization)
\be\label{wq1}
(Vf)(x)=\int_{\R^l\times\R^l}\V((x+y)/2,\xi,\hbar)e^{i\frac{\xi(x-y)}\hbar}f(y)\frac{dyd\xi}{(2\pi\hbar)^l},\ \
\ee
where in the integral $f(\cdot)$ and $\V((x+\cdot)/2,\xi,\hbar)$ are extended to $\R^l$ by periodicity (see Section \ref{matweyl} for details).
We make the following assumptions.
\vskip 1.5cm
\centerline{\textbf{Main assumptions}}
\vskip 0.5cm
\begin{itemize}
\item[(A1)]\hskip 0.5cm  The family of frequencies vectors $\om$ fulfills 
the {\bf  generalized Brjuno  condition }

\be\label{BC}
\sum_{l=1}^\infty\frac{\log \machin_{2^k}}{2^k}<+\infty
\mbox{ 
where }
\machin_M:=\min_{\substack{1\leq i\leq m}}\max_{0\neq |q|\leq M}|\la\om_i,q\ra|^{-1}.
\ee
\vskip 0.6cm
 \noindent We will sometimes impose to $\omega$ tu fulfill the strongest 
{\bf  collective Diophantine condition}: there exist $\gamma >0, \tau \geq l$ such that   
\be
\label{DC}
\forall q \in\Z^l, \; q\neq 0,\ \min_{1\leq i\leq m}|\la\om_i,q\ra|^{-1}\leq \gamma |q|^{\tau}. 
\ee
\vskip 0.6cm
%
{\bf Remark : usually, $\frac{1}{\machin_M} $ is denoted by $\omega_M$ in the literature} \cite{bruno, LS1}
\vskip 0.5cm
\item[(A2)]\hskip 0.5cm
$\V$ 
takes
 the form, for some $
\V':\ (\Xi,x,\hbar)\in\R^m\times\T^l\times[0,1]\mapsto\V'(\Xi,x,\hbar)\in\R^m 
 $, analytic in $(\Xi,x)$ and $k$th times differentiable in $\hbar$,
\be\label{Aiweyl}
\V(x,\xi,\hbar)=\V'(\om_1.\xi,\dots,\om_m.\xi,x,\hbar),
\ee
\item[(A3)]\hskip 0.5cm
The family $H$ satisfies
\be\label{com}
\ \ \ \ \ \ \ \ \ \ [H
_i,H
_j]=0,\ 
 \ \ \ 1\leq i, j\leq m,\ 0\leq\hbar\leq1.
\ee
\end{itemize}
Moreover we will suppose that the vectors $\omega_j,\ j=1\dots m$ are independent over $\R$ 
and we 
define
\be\label{normomega}
\bid
:=
\sum_{j=1}^{m}\vert \omega_j\vert
=\sum_{j=1}^{m}
\left(\sum_{i=1}^l(\omega_j^i)^2\right)^{1/2}
\ee
\vskip 1cm
Let us define for 
$\rho>0,\ k\in\{0\}\cup \mathbb N$ and $\V':\ (\Xi,x,\hbar)\in\R^m\times\T^l\times[0,1]\mapsto\V'(\Xi,x,\hbar)\in\R^m$
\[
\Vert \V'\Vert_{\rho,\bid,k}=\sum_{j=1}^{m}\sum_{r=0}^k
\Vert \partial_\hbar^r\widehat{\widetilde{\V'_j}}\Vert_{L^1_{\rho,\bid,r}(\R^m\times\Z^l)\otimes L^\infty([0,1])}\mbox{ and } \Vert\nabla\V'\Vert_{\rho,\bid,k}=\max_{i=1\dots l}\sum_{j=1}^{m}\sum_{r=0}^k\Vert\partial_\hbar^r\partial_{\Xi_j}\V'_i\Vert_{\rho,\bid,k},
\]
where $\widehat{\widetilde\cdot}$ denotes the Fourier transform on $\mathcal S(\R^m\times \T^l)$ and $L^1_{\rho,\bid,k}(\R^m\times\Z^l)$ is  the 
$L^1$ space equipped with the weighted norm  
$\sumql\limits\intm\limits\vert f(p,q)\vert (1+\vert \omega\cdot p\vert+\vert q\vert)^{\frac r2}e^{\rho(\bid\vert p\vert+\vert q\vert)}dp$ (See Section \ref{norms}).

{ Let us remark that $\Vert \V'\Vert_{\rho,\bid,k}<\infty$ implies that $\V'$ is analytic in a complex strip $\Im x<~\rho,\\ \Im \xi<~\rho\bid$ and $k$-times differentiable in $\hbar\in[0,1]$.}
\vskip 0.5cm
We will denote  $\overline{\V'}(\Xi):=\frac1{(2\pi)^l}\int_{\T^l}\V'(\Xi,x)dx$.
\vskip 0.5cm
Our assumptions are shown to be non empty in Remark \ref{nonvide}  and the relevance of assumption (A2) is discussed in Remark \ref{cond-nec}, both at the end of Section \ref{start} below.
\vskip 1cm
Our main result reads (see Theorems \ref{voila}, \ref{easygoing} and \ref{voiladio} for more precise  and explicit statements):

\begin{theorem}\label{first}
Let $k\in\mathbb N\cup\{0\}$ and $\rho>0$ be fixed. Let $H$ satisfy the Main Assumption above and $\Vert \V'\Vert_{\rho,\bid,k}$,$\Vert\nabla\overline{\V'}\Vert_{\rho,\bid,k}$ be small enough.

Then there exists a family of vector-valued functions $\B^\hbar_\infty (\cdot)$, $\partial^j_\hbar\B^\hbar_\infty (\cdot)$  
being
holomorphic in $\{\vert\Im z_i\vert<\frac\rho2,i=1\dots m\}$
uniformly 
with respect to $\hbar\in [0,1]$ and $0\leq j\leq k$, such that
the family $H$ is jointly unitary conjugated to 
$
\B^\hbar_\infty(L_\omega)
$ and therefore the spectrum of each $H_i$ is pure point and equals the set $\{(\B^\hbar_\infty)_i(\omega\cdot n), \ n\in\Z^l\}$ where $\omega\cdot n=(<\omega_i,n>)_{i=1\dots m}$.
\end{theorem}

Note that the use of Brjuno condition necessitates the intermediary result Theorem \ref{voila} involving an extra condition on $\omega$ removed by a scaling argument in Theorem \ref{easygoing}, as explained in Section \ref{kamiter}.

%
Our results being uniform in $\hbar$ we get as a partial bi-product of the preceding result the following   global version of \cite{LS1}:
 \begin{theorem}\label{corfirst}
 Let $\rho>0$ be fixed. Let $\mathcal H
 $ be a family of 
 $m\leq l$ Poisson commuting classical 
 Hamiltonians $(\H_i)_{i=1\dots m}$ on $T^*\T^l$ of the form $\mathcal H=\H^0+\V$, $\H^0(x,\xi)=\omega.\xi$, $\omega$ and $\V$ satisfying assumption $(A1)$ and $\V$ on the form $\V(x,\xi)=\V'(\om_1.\xi,\dots,\om_m.\xi,x)$. Let finally  $\Vert\V'\Vert_{\rho,\bid}$,$\Vert\nabla\overline{\V'}\Vert_{\rho,\bid,0}$ be small enough (here we consider $\V'$ as a  function constant in $\hbar$).
 
 Then $\mathcal H$ is (globally) symplectomorphically and holomorphically conjugated to $\mathcal B^0_\infty(\H^0)$.
 \end{theorem}
 
 Once again let us mention that our results are much more explicit, precise and complete (in particular concerning radii of convergence and unitary/symplectic conjugations) as expressed in Theorems \ref{voila}, \ref{easygoing} and \ref{voiladio} and Corollary \ref{coresygoing}.

Moreover it appears in the proofs that the statement in Theorem \ref{first}, as well as in Theorems \ref{voila}, \ref{easygoing} and \ref{voiladio} and Corollary \ref{coresygoing}, is valid for fixed value of the Planck constant $\hbar$ under the Main Assumption  lowed down by restricting \eqref{com}  to $\hbar$ fixed. More precisely under the Main Assumption with (A3) restricted to, e.g., $\hbar=1$, the Theorem \ref{first} is still valid by putting in the statement $k=0$ and $\hbar=1$. 
Let us mention also that, as in the original formulations in \cite{Ru}-\cite{LS1}, one easily sees that condition (A2) can be replaced by the fact that the quantum  Birkhoff normal form (see section \ref{start} below for the precise definition) at each order is a function of 
$(L_1,\dots,L_m)$ only.
\vskip 0.5cm
Let us emphasize the two extreme cases, that is $m=l$ and $m=1$.
\begin{coro}[Quantum Vey theorem]\label{extreme-l}
Assume that the $\omega_j\in\mathbb R^l$, $j=1,\dots,l$, are independent over $\mathbb R$. Assume that the $H_i=L_{\om_i}+V_i$, $i=1,\ldots, l$ are pairwise commuting. Let the perturbation $V_i$ be the quantization of any small enough analytic function $\V_i$. Then the family $H$ is jointly unitary conjugated to 
$\B^\hbar_\infty(L_\omega)$ as defined in theorem \ref{first}.
\end{coro}
We emphasize that this last result do not require neither a small divisors condition nor a condition on the perturbation, see Section \ref{secextrem}. This correspond to full quantum integrability. Quantum integrability is a huge subject - see the seminal articles \cite{CdV1,CdV2} to quote only two. The difference that provides our construction is the fact that our results gives convergent result even at $\hbar=1$ is the case of perturbations of linear systems. 

\begin{coro}[consolidated Graffi-Paul theorem]\label{extreme-1}
Assume that  $\omega\in\mathbb R^l$ satisfies Brjuno condition ($m=1$). Assume that  $H=L_{\om}+V$, where the perturbation $V$ is small enough and $\V(\xi,x)=\V'(\om.\xi,x)$. Then $H$ is unitary conjugated to 
$\B^\hbar_\infty(L_\omega)$ as defined in theorem \ref{first}.
\end{coro}
The main difference between this last result and the main result of \cite{GP} is the small divisors condition used (a Siegel type condition with constraints). 


Le us finally mention a by-product of our resul, a kind of inverse result, obtained thanks to the fact that we carefully took care of the precise estimations and constants all a long the proofs. This result is motivated by the remark that, though a small divisors condition is necessary to obtain the perturbed integrability (and  Brjuno condition is sufficient),  such a condition should disappear when the perturbation vanishes, as the Hamiltonian $H^0$ is always integrable, whatever the frequencies $\omega$ are. Our last result quantifies this remark. 

Let us define, for $\omega$ satisfying \eqref{DC} and $\alpha<2\log{2}$,  
\[
B_\alpha(\gamma,\tau)=2\log{\left[2^\tau\gamma(\frac\tau{e\alpha})^\tau\right]
}
\]
(note that
$B_\alpha(\gamma,\tau)\to\infty$ as $\gamma$ and/or $\tau\to\infty$).

\noindent The next Theorem shows that, in the Diophantine case, the small divisors condition can be released as  $B_\alpha(\gamma,\tau)$ diverging logarithmically 
as the perturbation vanishes.
\begin{theorem}\label{cortout}
  Let $k\in\mathbb N\cup\{0\}$ and $\rho>0$ be fixed. Let $\omega$ and $\V$ satisfy (A1) (Diophantine case), (A2) and (A3), and let $0<\bid_-\leq\bid\leq\bid_+<\infty$ and $\Vert\V'\Vert_{\rho,\bid_+,k}$, $\Vert\nabla\overline{\V'}\Vert_{\rho,\bid_+,k}$ be small enough (depending only on $k$).
  

Then there exist a constant $C_{\bid_-}$ such that the conclusions of
 Theorems \ref{first}  hold
as soon as, for some $\alpha<\rho/2,\ \alpha< 2\log{2}$, 
\[
B_\alpha(\gamma,\tau)<\frac 13 \log{\left(\frac1 {\Vert V\Vert_{\rho,\bid_+,k}}\right)}+C_{\bid_-}.
\]
 \end{theorem}
 See Corollary \ref{corder} for details and the Remark after on the case of the Brjuno condition. Let us remark that an equivalent result for Theorem  \ref{corfirst} is straightforwardly obtainable.
\vskip 1cm
Let us finish this section by mentioning three comments and remarks concerning our results.

First of all, as mentioned earlier, no hypothesis on the minimal distance between two consecutive unperturbed eigenvalues is required in our article. More, the spectra of our unperturbed operators $L_{\omega_i}$ might be dense for all value of $\hbar$  (actually in the Diophantine case for 
$m=1,\ l>1$ they are) so 
the Neumann series expansion is not possible. For $m>1$ the non degeneracy  of the unperturbed eigenvalues is not even insured by the arithmetical property of $\omega$ because it relies on the minimum over $i\leq m$ of the inverse of the small denominator of the vector $\omega_i$. In fact, for a resonant $\omega_j$ the operator $H_j$ will have an eigenvalue with infinite degeneracy, so the projection of the perturbation $V_j$ on the corresponding and infinite dimensional eigenspace, which leads to the first order perturbation correction to the unperturbed eigenvalue, might have continuous spectrum.  Nevertheless  our results show that the perturbed spectra are  analytic functions of the spectra  of the $L_{\omega_i}$'s.

Secondly, because of the fact that non degeneracy of some of the unperturbed  spectra is not even guaranteed by our assumptions, the standard argument on existence of a common eigenbasis of commuting operators with simple spectra cannot be involved here. This existence is a bi-product of our results. 

Finally let us mention that, as it was the case in \cite{GP}, 
 though our hypothesis on the perturbations are 
 restrictive,
our results, compared with the usual construction of quasi-modes 
\cite{Ra,CdV,PU,Po1,Po2}, have the property of being global in the spectra (full diagonalization),  and exact (no smoothing or $O(\hbar^\infty)$ remainder), together of course with sharing the property of being uniform in the Planck constant.

Let us point out that this paper has been written in order to be self-contained
\section*{Notations}
Function valued vectors in $\R^n$ will be denoted in general in calligraphic style, and operator valued vectors by capital letters, e.g. $V=(V_l)_{l=1\dots m}$ or $\mathcal V=(\mathcal V_l)_{l=1\dots m}$.

For $i,j\in\Z^n$ we will denote by $\cdot_{ij}$ or $\cdot_{,ij}$ when $\cdot$ has already an index, the matrix element of an (vector) operator in the basis 
$\{e_j,\  e_j(x)=e^{ij.x}/(2\pi)^{\frac l2}, \theta\in\T^l\}$, namely
\[
V_{ij}=(V_{l,ij})_{l=1\dots m}=((e_i,V_{l}e_j)_{L^2(\T^n)})_{l=1\dots m},
\]
and by $\overline V$ the diagonal part of $V$:
\[
\overline V_{ij}=V_{ii}\delta_{ij},
\]
together with
\[
\overline \V=(2\pi)^{- l}\int_{\T^l}\V dx.
\]

We will denote by $\vert \cdot\vert$ the Euclidean norm on  $\R^m$ (or $\mathbb C^m$), $\vert Z\vert^2=\sum\limits_{i=1}^m\vert Z_i\vert^2$, and by $\Vert \cdot\Vert_{L^2(\T^l)\to L^2(\T^l)}$ the operator norm on the Hilbert space $ L^2(\T^l)$.

Finally for 
$\omega=(\omega_i\in\R^l)_{i=1\dots m}$ and $\xi\in\R^l,\ p\in\R^m,\ q\in\Z^l$ we will denote
\be\label{notpmega}
\omega\cdot\xi=(<\omega_i,\xi>_{\R^l})_{i=1\dots m}\in\R^m,
\ee
\be\label{pmega}
p.\omega=\left(\sum_{i=1}^mp_i\omega_i^j\right)_{j=1\dots l}\in\R^l
\ee
and
\be\label{pomegaq}
p.\omega.q=\sum_{i=1}^m\sum_{j=1}^lp_i\omega_i^jq_j
=\la\pom,q\ra_{\R^l}.
\ee

\section{Strategy of the proofs}\label{start}

The general idea in proving Theorem \ref{first} will be to construct a Newton-type iteration procedure consisting in constructing a family of unitary operators $U_r$ such that (norms will be defined later)
\be\label{strat1}
U_r^{-1}(\mathcal B^\hbar_r(L_\om)+V_r)U_r=\mathcal B^\hbar_{r+1}(L_\om)+V_{r+1},
\ee
with 
$\Vert V_{r+1}\Vert_{r+1}\leq D_{r+1}\Vert V_{r}\Vert_r^2$ and $\mathcal B^\hbar_0(L_\om)=L_\om,\ V_0=V$.

$U_r$ will be chosen of the form
\be\label{strat2}
U_r=e^{i\frac{W_r}\hbar}, \ W_r \mbox{ self-adjoint}.
\ee
It is easy to realize that $\eqref{strat2}$ 
implies $\eqref{strat1}$ if $W_r$ satisfies the (approximate) cohomological equation 
\be\label{strat3}
\frac 1{i\hbar} [\mathcal B^\hbar_r(L_\om),W_r]+V_r=\mathcal D_{r+1}(L_\om)+O(\Vert V_r\Vert_r^2),
\ee
or equivalently
\be\label{strat4}
\frac 1{i\hbar} [\mathcal B^\hbar_r(L_\om),W_r]+V^{co}_r=\mathcal D_{r+1}(L_\om)+O(\Vert V_r\Vert_r^2),
\ee
for any $V^{co}_r$ such that $\Vert V^{co}_r-V_r\Vert_r=O(\Vert V_r\Vert_r^2)$.

We will solve for each $r$ the equation \eqref{strat4} where $V^{co}_r$ will be obtained by a suitable ``cut-off" 
in order to have to solve \eqref{strat4} with only small denominators of finite order (see Brjuno condition \eqref{BC}).

In fact we will see in Section \ref{form} that we can find a (scalar) solution of the (vector) equation \eqref{strat4} satisfying 
\be\label{strat5}
\frac 1{i\hbar} [\mathcal B^\hbar_r(L_\om),W_r]+V^{co}_r=\mathcal B^\hbar_{r+1}(L_\om)+R_r,
\ee
where $\Vert R_r\Vert_{k+1}=O(\Vert V_r\Vert_r^2)$. To do this we will remark that since the components of $\mathcal B^\hbar_r(L_\om)+V_r$ commute with each other (since the ones of $L_\om+ V$ do) we have that
\be\label{strat6}
[(\mathcal B^\hbar_r(L_\om))_l,(V_r)_{l'}]-[(\mathcal B^\hbar_r(L_\om))_{l'},(V_r)_{l}]=[(V_r)_{l'},(V_r)_l]=O( V_r^2)
\ee
which is an almost compatibility condition (see Section \ref{form} for details).

Summarizing, the solution $W_r$ of \eqref{strat4} will provide a unitary operator $U_r$ such that \eqref{strat1} will hold  with 
$\mathcal B^\hbar_{r+1}=\mathcal B^\hbar_r+\mathcal D_{r+1}$ and $V_{r+1}$ 
being the sum of three terms:
\begin{itemize}
\item $V_{r+1}^1=U_r^{-1}(\mathcal B^\hbar_r(L_\om)+V_r)U_r-(\mathcal B^\hbar_r(L_\om)+V_r)-\frac1{i\hbar}[\mathcal B^\hbar_r(L_\om),W_r]$
\item $V_{r+1}^2=V_r-V^{co}_r$
\item $V_{r+1}^3=R_r$
\end{itemize}

The choice of the family of norms $\Vert\cdot\Vert_r$ will be made in order to have that
\[
\Vert V_{r+1}\Vert_{r+1}=\Vert V_{r+1}^1+V_{r+1}^2+V_{r+1}^3\Vert_{r+1}\leq D_{r+1}\Vert V_r\Vert_r^2
\]
with $D_r$ satisfying
\[\prod_{r=1}^RD_r^{2^{R-r}}\leq C^{2^R}.
\]
Hence, we have 
$$
\Vert V_{R+1}\Vert_{R+1}\leq \left(C\Vert V_{0}\Vert_{0}\right)^{2^R}, 
$$
so that $\Vert V_{R+1}\Vert_{R+1}\to 0$ as $R\to\infty$ if 
$\Vert V_0= V\Vert_0< C^{-1}$ and $\Vert\cdot\Vert_\infty$ exists.
\begin{remark}\label{a2a3}[Propagation of assumptions (A2)-(A3)]
It is  clear (and it will be explicit in the body of the proofs of the main Theorem) that Condition (A2) will be satisfied by the solution of equations \eqref{strat3},\eqref{strat4} as soon as $V^{}_r$ and $V^{co}_r$ do. This last condition can be easily seen to be propagated from the decomposition $V_{r+1}=V^1_{r+1}+V^2_{r+1}+V^3_{r+1}$ given before by considering that $U_r=e^{i\frac{W_r}\hbar}$ by \eqref{strat2} and $W_r$ satisfying (A2). (A3) is obviously propagated by \eqref{strat1}.
\end{remark}
\begin{remark}\label{nonvide}[Non emptiness of the hypothesis] Consider a family of operators of the form $L_\omega+\B^\hbar(L_\omega)$ for $\B^\hbar: \R^m\to\R^m$  with $\Vert \B^\hbar\Vert_{\rho,\bid,k}<+\infty$. Then for each bounded self-adjoint operator $W$ whose Weyl symbol $\W$ satisfies (A2) and $\Vert W\Vert_{\rho',\bid,k}<+\infty$ for some $\rho'>\rho$, consider the family $e^{i\frac W\hbar}(L_\omega+\B^\hbar(L_\omega))e^{-i\frac W\hbar}:=L_\omega+V:=(H_i)_{i=1\dots m}$. Obviously the family $(H_i)_{i=1\dots m}$ satisfies (A3). By the same argument as the one in Remark \ref{a2a3} one sees easily that the Weyl symbol $\V$ of $V$ satisfies (A2) for some $\V'$. 
Finally  estimates \eqref{emboites} and \eqref{emboitesL} in Proposition \ref{stimeMo} below  show that the expansion 
$e^{i\frac W\hbar}(L_\omega+\B^\hbar(L_\omega))e^{-i\frac W\hbar}=L_\omega+\B^\hbar(L_\omega)+[L_\omega+\B^\hbar(L_\omega),\frac{iW}\hbar]+\frac 12 [[L_\omega+\B^\hbar(L_\omega),\frac{iW}\hbar],\frac{iW}\hbar]+\dots$ is actually convergent.  This  implies  that $\Vert \V\Vert_{\rho,\bid,k}$ is bounded. 
Therefore the family $L_\omega+V$ satisfies all the assumptions of Section \ref{intro}.
\end{remark}
\begin{remark}\label{cond-nec}[Relevance of assumption (A2)] Let us recall some classical facts from dynamical systems. Let $H_0=\sum\limits_{i=1}^n \lambda_i(x_i^2+y_i^2)$ be a quadratic Hamiltonian on $\Bbb R^{2n}$. Any analytic higher order perturbation $H=H_0+higher\ order\ terms$ is formally conjugate to a formal Birkhoff normal form $\hat H(x_1^2+y_1^2,\ldots, x_n^2+y_n^2)$. R\"ussman-Brjuno's theorem asserts that, if $(*)\quad\hat H=\hat F(H_0)$ (i.e. $\hat H$ is a function of that peculiar linear combination $\sum\limits_{i=1}^n \lambda_i(x_i^2+y_i^2)$ and contains no other terms), for some formal power series $\hat F$ of one variable and if a "small divisors" condition is satisfied, then the transformation to the Birkhoff normal form is analytic in a neighborhood of the origin.
Condition $(*)$ is known as Brjuno's condition A (cf. \cite{bruno}). It is a  sharp condition for the analycity of the transformation to Birkhoff normal form in the following sense : if a normal form NF doesn't satisfy it, then it is possible to perturb $H$ in such way that the analytic perturbation $\tilde H$ still has $NF$ as normal form and the transformation from $\tilde H$ to $NF$ is a divergent power series. 
In our quantum version, we only focus on the sufficiency of the analogue condition. The linear combination $\sum_j\omega_j\xi_j$ in our article plays the r\^ole of "quantum analogue" of $\sum_i \lambda_i(x_i^2+y_i^2)$ 
\end{remark}

\section{The cohomological equation: the formal construction}\label{form}
In this section we want to show how it is possible 
to construct the solution of the equation
\be\label{coho1}
\frac 1{i\hbar} [\mathcal B^\hbar(L_\om),W]+V=\mathcal D(L_\om)+O( V^2),
\ee
where we denote by
 $L_\om,\ \om=(\om_i\in\R^l)_{i=1\dots  m}$, the 
 operator valued vector 
 of components (with a slight abuse of notation) $L_{\om_i}=-i\hbar\om_i.\nabla_x ,\ i=1\dots m$ on $L^2(\T^l)$ and $V$ is a ``cut-off"ed.
\[
V_{ij}=0 \mbox{ for } \vert i-j\vert>M.
\]

%
%
%
We will present the strategy only in the case of the Brjuno condition, the Diophantine case being very close.

Let us recall also  that equation \eqref{coho1} is in fact a system of $m$ equations and that it might seem surprising at the first glance that the same $W$ solves \eqref{coho1} for all $\ell=1\dots m$. 

\medskip
\subsection{First order}\label{cohofirstorder} 
At the first order the cohomological equation is 
\be\label{heqell1}
\frac{[L_{\om_\ell},W]}{i\hbar}+V_\ell=\mathcal D_\ell(L_\om),\ \ \ l=1\dots m
\ee
solved on the eigenbasis of any $L_{\om_\ell}$ by 
$\mathcal D_\ell(L_\om)=diag (V_\ell)$ and 
\be\label{truc}
W_{ij}=-\frac{(V_\ell-\mathcal D_\ell)_{ij}}{i\om_\ell\cdot(i-j)}.
\ee
Indeed, since $L_{\omega_l}$ is selfadjoint, we have
\begin{eqnarray*}
<e_j, [L_{\omega_l},W]e_i>&=& <e_j,L_{\omega_l}We_i- WL_{\omega_l}e_i>= <L_{\omega_l}e_j,We_i> - <e_j,WL_{\omega_l}e_i>\\
&=& i\omega_l.(j-i) <e_j,We_i>\
\end{eqnarray*}
In \eqref{truc} we will picked up, for every $ij$ such that $|i-j|\leq M$, an index $\ell={\ell_{i-j}}$ which minimize 
 the quantity 
 \be\label{cohobc}
 |\la\om_{\ell_q},q\ra|^{-1}:=
 \min_{1\leq i\leq m}|\la\om_i,q\ra|^{-1}\leq \machin_M.
 \ee
  
  We define $W$ by
\be\label{truc2}
W_{ij}=-\frac{(V_{{\ell_{i-j}}})_{ij}}{i\om_{{\ell_{i-j}}}\cdot(i-j)},\ \ \  i-j\neq 0
\ee
Since $[H_\ell,H_{\ell'}]=0$, then we have that $[L_{\ell'}V_{\ell}]+[V_{\ell'},L_{\ell}]=-[V_{\ell},V_{\ell'}]$. 
Therefore, 
evaluating the operators on $e_j$ and taking the scalar product with $e_i$, leads to 
\be\label{eupeur}
\om_{\ell'}\cdot(i-j)(V_{\ell})_{ij}=\om_{\ell}\cdot(i-j)(V_{\ell'})_{ij}-([V_{\ell},V_{\ell'}])_{ij}
\ee
that is 
\[\
\frac{(V_{\ell})_{ij}}{\om_{\ell}\cdot(i-j)}=\frac{(V_{\ell'})_{ij}}{\om_{\ell'}\cdot(i-j)}-\frac{([V_{\ell},V_{\ell'}])_{ij}}{\om_{\ell}\cdot(i-j)\om_{\ell'}\cdot(i-j)}
\]
(note that when $\om_{\ell'}\cdot(i-j)=0$ on has 
$(V_{\ell'})_{ij}=
\frac{-([V_{{\ell_{i-j}}},V_{\ell'}])_{ij}}{\om_{{\ell_{i-j}}}\cdot(i-j)}$).

Let us remark that, though $[V_{\ell},V_{\ell'}]$ is quadratic in $V$, it has the same cut-off property as $V$, namely 
$([V_{\ell},V_{\ell'}])_{ij}=0$ if $\vert i-j\vert>M$ as seen clearly by \eqref{eupeur}.

This means that  $W$ defined by  \eqref{truc2} satisfies
\[
\frac{[L_{\om},W]}{i\hbar}+V=\mathcal D(L_\om)+\widehat V,
\]
where 
\be\label{coho4}
(\widehat V_\ell)_{ij}=\frac{([V_\ell,V_{{\ell_{i-j}}}])_{ij}}{i\hbar\omega_{{\ell_{i-j}}}\cdot(i-j)}.
\ee
\vskip 1cm
%
%
Note that this construction is different from the one used in \cite{LS1}.
\vskip 1cm
%
%
%
%
%
%
%
%
%
%


\subsection{Higher orders}\label{cut}


The cohomological equation at  order $r$ will follow the same way, at the exception that $L_\omega$ has to be replaced by 
$\B^\hbar_r(L_\omega)$.

The corresponding cohomological equation is therefore of the form
\be\label{ho1}
\frac{[\B^\hbar_r(L_\omega),W_r]}{i\hbar}+ V_r=O((V_r)^2),
\ee
equivalent to 
\be\label{ho2}
\frac{\B^\hbar_r(\hbar\omega\cdot i)-\B^\hbar_r(\hbar\omega \cdot j)}{i\hbar}(W_r)_{ij}+(V_r)_{ij}=O((V_r)^2).
\ee
\begin{lemma}\label{matruc}
For $\B^\hbar_r$ close enough to the identity there exists a $m\times m$ matrix  $A^r(i,j)$ such that
\be\label{ho3}
\frac{\B^\hbar_r(\hbar\omega\cdot i)-\B^\hbar_r(\hbar\omega\cdot j)}{i\hbar}
=\left(I+ A^r(i,j)\right)\omega.(i-j),
\ee
where $I$ is the  $m\times m$ identity matrix and $\omega.(i-j)=(\omega_l.(i-j))_{l=1\dots m}$.
Moreover    
\be\label{db}
\Vert A^r(i,j)\Vert_{\C^m\to\C^m}\le\Vert \nabla(\B^\hbar_r-
\B^\hbar_0)\Vert_{(\C^m\to\C^m)\otimes L^\infty(\R^m)}
\leq\max_{j=1\dots m}\sum_{i=1}^{m}\Vert \nabla_j(\B^\hbar_r-\B^\hbar_0)_i\Vert_{ L^\infty(\R^m)}.
\ee
\end{lemma}
\begin{proof}
We have
\bea
\frac{\B^\hbar_r(\hbar\omega\cdot i)-\B^\hbar_r(\hbar\omega\cdot j)}{i\hbar}&=&\omega.(i-j)+\int_0^1\partial_t\left[(\B^\hbar_r-\B^\hbar_0)(t\hbar\omega.i+(1-t)\hbar\omega.j)\right]\frac{dt}\hbar\nonumber
\\
&=&\omega.(i-j)+\int_0^1\left[\nabla(\B^\hbar_r-\B^\hbar_0)(t\hbar\omega.i+(1-t)\hbar\omega.j)\right]\cdot [\omega.(i-j)]dt\nonumber
\eea
so $A^r(i,j)=\int_0^1\nabla(\B^\hbar_r-\B^\hbar_0)(t\hbar\omega.i+(1-t)\hbar\omega.j)dt$ and the first part of \eqref{db} follows. The second part is a standard estimate of the operator norm.
\end{proof}
Plugging \eqref{ho3} in \eqref{ho2} we get that $W$ must solve 
\be\label{ho4}
\omega.(i-j)W_{ij}=\left(I+ A^r(i,j)\right)^{-1}\left[-(V_r)_{ij}+O((V_r)^2)\right],
\ee
and we are reduced to the first order case with
$V_r\to \widetilde V^r$ where 
\be\label{coho5}
\widetilde V^r_{ij}:=\left(I+ A^r(i,j)\right)^{-1}(V_r)_{ij}.
\ee
\subsection{Toward estimating}\label{twes}

We will first have to estimate $\widetilde V^r$: this will be done out of its matrix coefficients given by \eqref{coho5} 
by the method developed in Section \ref{matweyl}. 
We will estimate 
 $\left(I+ A^r(i,j)\right)^{-1} \widetilde V^r$ in section \ref{brunofundamental} by using the formula $\left(I+ A^r(i,j)\right)^{-1}=\sum\limits_{k=0}^\infty (-A^r(i,j))^k$ and a bound of the norm of $(-A^r(i,j))^k\widetilde V^r$ of the form $|C|^k$ times the norm of $\widetilde V^r$ leading to a bound of 
 $\left(I+ A^r(i,j)\right)^{-1} \widetilde V^r$ of the form 
 $\frac1{1-|C|}$ times the norm of $\widetilde V^r$, by summation of the geometric series $\sum\limits_{k=0}^\infty C^k$, possible at the condition that $|C|<1$.
 

We will then have to estimate $W$ defined through
\be\label{truc2tilde}
W_{ij}=-\frac{(\widetilde V^r_{{\ell_{i-j}}})_{ij}}{i\om_{{\ell_{i-j}}}\cdot(i-j)},\ \ \  i-j\neq 0
\ee
with again $(\widetilde V^r_{{\ell_{i-j}}})_{ij}=0$ for $\vert i-j\vert>M$. We get 
\[
\vert W_{ij}\vert\leq \machin_M\vert(\widetilde V_{{\ell_{i-j}}})_{ij}\vert,
\]
and we will get an estimate of $W$,\  $\Vert W\Vert\leq \machin_M\Vert\widetilde V^r\Vert$, for a norm $\Vert\cdot\Vert$ to be specified later.

Finally we will have to estimate

\be\label{coho4tilde}
(\widehat V^r_l)_{ij}=\frac{([\widetilde V^r_l,\widetilde V^r_{{\ell_{i-j}}}])_{ij}}{i\hbar\omega_{{\ell_{i-j}}}\cdot(i-j)}.
\ee

We will get immediately 
$\Vert\widehat V^r_l\Vert\leq \machin_M\Vert P\Vert,\ P_{ij}=\frac{([\widetilde V^r_l,\widetilde V^r_{{\ell_{i-j}}}])_{ij}}{i\hbar}$ and the estimate of the commutator will be done by the method developed in Section \ref{sectionweyl}. 

In the  two next sections we will define the norms and the Weyl quantization procedure used in order to precise the results of this section,

\section{Norms}\label{norms}


  Let $m,l$ be positive integers. For $\F\in C^\infty(\R^m\times\T^l\times [0,1]; \C)$  
 we will use the following normalization for the Fourier transform. 
 \vskip 6pt\noindent
 \begin{definition}[Fourier transforms]\label{deffour}
 Let $p\in\R^m$ and $q\in\Z^l$ 
 \bea
 \widehat{\F}(p,x,\hbar)&=&\frac1{(2\pi)^{m}}\int_{\R^m}\F(\xi,x,\hbar)e^{-i\la p,\xi\ra}\,d\xi\label{fourtrans}\\
\truc{\F}(\xi,q;\hbar)&=&\frac1{(2\pi)^{l}}
\int_{\T^l}\F(\xi,x;\hbar)e^{-i\la q,x\ra}\,dx \label{deftild}
\\
\widehat{\widetilde\F}(p,q,\hbar)&=&\frac1{(2\pi)^{m+l}}\int_{\R^m\times\T^l}
\F(\xi,x,\hbar)e^{-i\la p,\xi\ra-i\la q,x\ra}d\xi dx\label{check}\\
&=&\frac1{(2\pi)^{m}}\int_{\R^m}
\widetilde\F(\xi,q,\hbar)e^{-i\la p,\xi\ra}d\xi \label{check2}\\
&=&\frac1{(2\pi)^{l}}\int_{\T^l}
\widehat\F(p,x,\hbar)e^{-i\la q,x\ra} dx\label{check3}\\
&&\nonumber\\
\mbox{Note that}\ \ \ \ \ \ &\ &\nonumber\\
\label{FE0}
{\F}(\xi,x,\hbar)&=&
\int_{\R^m}\widehat\F(p,x,\hbar)e^{i\la p,\xi\ra}\,dp\\
\label{FE1}
&=&\sum_{q\in\Z^l}\truc{\F}(\xi,q;\hbar)e^{i\la q,x\ra}\\
\label{FE2} 
&=&\sum_{q\in\Z^l}\int_{\R^m}\widehat{\widetilde{\F}}(p,q,\hbar)e^{i\la p,\xi\ra+i\la q,x\ra}\,dp
\eea
\end{definition}
 \vskip 10pt\noindent
Set now for $ k\in\Na\cup\{0\} $ 
and $p\cdot\omega=(\sum\limits_{j=1\dots m}p_j.\omega_j^i)_{i=1\dots l}$:
\be\label{muk}
\mu_{k}(p,q):=(1+|p\cdot\omega |^{2}+|q|^{2})^{\frac k 2}
\ee
(note that
$
 \mu_r(p-p',q-q')\leq 2^{\frac k 2} \mu_r(p,q)\mu_r(
 p',q')
$
because $|x-x^\prime|^2\leq 2(|x|^2+|x^\prime|^2)$ and that $|p\cdot\omega|\to\infty$ as $|p|\to\infty$ because the vectors $(\omega^i)_{1=1\dots l}$ are independent over $\mathbb R$). 
\begin{definition}[Norms I]
{\it   For $\rho> 0$,
$\F\in C^\infty(\R^m\times\T^l\times [0,1]; \C)$ we 
 introduce the weighted norms }
 \vskip 3pt\noindent
\bea 
\label{norma1}
\Vert\F\Vert^\dagger_{\rho}=\Vert\F\Vert^\dagger_{\rho,\bid}&:
=&\max_{\hbar\in [0,1]}\intm\sumql|\widehat{\widetilde{\F}}(p,q,\hbar)|\,e^{\rho (\bid|p|+|q|)}\,dp.
\\
\label{norma1k}
\Vert\F\Vert^\dagger_{\rho,\bid,k}=\Vert\F\Vert^\dagger_{\rho,\bid,k}&:=&\max_{\hbar\in [0,1]}\sum_{j=0}^k\intm\sumql\mu_{k-j}(p,q)\partial^j_\hbar|\widehat{\widetilde\F}(p,q,\hbar)|\,e^{\rho (\bid|p|+|q|)}\,dp.
\eea
Note that $\bid$ is given by \eqref{normomega} and $\Vert\cdot\Vert^\dagger_{\rho;0}=\Vert\cdot\Vert^\dagger_{\sigma}$.

\end{definition}

\begin{definition}[Norms II]\label{norm2} {\it 
Let ${\mathcal O}_\omega$ be the set of functions $\F:\R^l\times\T^l\times[0,1]\to\C$ such that $\F(\xi,x;\hbar)=\F'(\omega\cdot\xi,x,\hbar)$ for some $\F':\ \R^m\times\T^l\times [0,1]\to \C$. 
Define, for $\F\in {\mathcal O}_\omega$:
\bea\label{sigom}
\Vert \F\Vert_{\rho,\bid,k}:=
\Vert \F'\Vert_{\rho,\bid,k}^\dagger.
\eea
We will also need the following definition for $\F\in{\mathcal O}_\omega$:
\be\label{sigomh}
\Vert \F\Vert_{\rho,\bid,k}^\hbar:=
\sum_{j=0}^k\intm\sumql\mu_{k-j}(p,q)\partial^j_\hbar|\widehat{\widetilde\F'}(p,q,\hbar)|\,e^{\rho (\bid|p|+|q|)}\,dp.
\ee
Let us note that, obviously, $\Vert\cdot\Vert_{\rho,\bid,k}^\hbar\leq\Vert\cdot\Vert_{\rho,\bid,k}$.
\vskip 4pt\noindent
}
\end{definition}
We will need an extension of the previous definition to the vector case. Consider now  $\F\in C^\infty(\R^m\times\T^l\times [0,1]; \C^m)$  
 and $\G\in C^\infty(\R^m\times [0,1]; \C^m)$. The definition of the Fourier transform is defined as usual, component by component. 
\begin{definition}\label{norm3}[Norms III] {\it
Let  $\F=(\F_i)_{i=1\dots m}\in C^\infty(\R^m\times\T^l\times [0,1]; \C^m)$.  
 We define
 
 \begin{itemize}
\item[(1)]
\be
\label{normamk}
\Vert\F\Vert^\dagger_{\rho,\bid,k}=
\sum_{i=1}^{m}
\Vert\F_i\Vert^\dagger_{\rho,\bid,k}
\ee
\item[(2)]

Let 
\be\label{omegam}
{\mathcal O}_\omega^m=\left\{{\F}=(\F_i)_{i=1\dots m}:\R^m\times\T^l\times[0,1]\to\C^m/\ 
\F_i\in\mathcal O_\omega, i=1\dots m\right\}
\ee
Let $\F\in {\mathcal O}_\omega^m$. We define:
\be\label{sigom}
\Vert \F\Vert_{\rho,\bid,k}
=
\sum_{i=1}^{m}
\Vert \F_i\Vert_{\rho,\bid,k}
\ee
Let
\be\label{omegamat}
{\mathcal O}_\omega^{m\times m}=\left\{{\F}=(\F_{ij})_{i,j=1\dots m}:\R^m\times\T^l\times[0,1]\to\C^m/\ 
\F_{ij}\in\mathcal O_\omega, i,j=1\dots m\right\}
\ee
Let $\F\in {\mathcal O}_\omega^{m\times m}$ . We define:
\be\label{sigomat}
\Vert \F\Vert_{\rho,\bid,k}
=\sup_{i=1\dots m}
\sum_{j=1\dots m}\Vert \F_{ij}\Vert_{\rho,\bid,k}.
\ee

\item
[(3)] Finally we denote 
$F$ the Weyl quantization of $\F$ recalled in Section \ref{sectionweyl} and
\bea
\label{opesym}
\Vert F\Vert_{\rho,\bid,k}&=&\Vert\F\Vert_{\rho,\bid,k} \\
&&\nonumber\\
\label{normsymb'}
\J^\dagger_k(\rho,\bid)&=&\{\F \,|\,\Vert \F\Vert^\dagger_{\rho,\bid,k}<\infty\},
\\
\label{normop'}
J^\dagger_k(\rho,\bid)&=&\{F\,|\,\F\in\J^\dagger_k(\rho,\bid)\},
\\
\label{normsymb0}
\J_k(\rho,\bid)&=&\{\F\in {\mathcal O}_\omega\,|\,\Vert \F\Vert_{\rho,\bid,k}<\infty\},
\\
\label{normop0}
J_k^{}(\rho,\bid)&=&\{F\,|\,\F\in\J_k^{}(\rho,\bid)\}.
\\\label{normsymbh}
\J_k^\hbar(\rho,\bid)&=&\{\F\in {\mathcal O}_\omega\,|\,\Vert \F\Vert_{\rho,\bid,k}^\hbar<\infty\},
\\
\label{normoph}
J_k^{\hbar}(\rho,\bid)&=&\{F\,|\,\F\in\J_k^{\hbar}(\rho,\bid)\}.
\\
\label{normsymb}
\J_k^m(\rho,\bid)&=&\{\F\in {\mathcal O}_\omega^m\,|\,\Vert \F\Vert_{\rho,\bid,k}<\infty\},
\\
\label{normop}
J_k^{m}(\rho,\bid)&=&\{F\,|\,\F\in\J_k^{m}(\rho,\bid)\}.
\\
\label{normsymbmat}
\J_k^{m\times m}(\rho,\bid)&=&\{\F\in {\mathcal O}_\omega^{m\times m}\,|\,\Vert \F\Vert_{\rho,\bid,k}<\infty\},
\\
\label{normopmat}
J_k^{m\times m}(\rho,\bid)&=&\{F\,|\,\F\in\J_k^{m\times m}(\rho,\bid)\}
\eea 
\end{itemize}}
and $\J^@(\rho,\bid)=\J^@_{k=0}(\rho,\bid), \ J^@(\rho,\bid)=J^@_{k=0}(\rho,\bid)$ 
 $\forall @\in\{\dagger,m,m\times m\}$.
 
 \noindent{\bf When there will be no confusion we will  forget about the subscript $_\bid$ in the label of the norms and also denote by $\J^@_k(\rho)=\J^@_k(\rho,\bid).$}
 \end{definition}

 \section{Weyl quantization and first estimates}\label{sectionweyl}
 
We express the definitions and results of this section in case of scalar operators and symbols. The extension to the vector case is trivial component by component. The reader only interested by explicit expression can skip the beginning of the next paragraph and go directly to Definition \ref{weylII}.
\subsection{Weyl quantization, matrix elements and first estimates}\label{matweyl}
In this section we recall briefly the definition of the Weyl quantization of $T^*\T^l$. The reader is referred to \cite{GP} for more details (see also  e.g. \cite{Fo}).

Let us recall that the Heisenberg  group over $\ds T^*\T^l\times\R$,   denoted by $\He_l(\R^l\times\Z^l\times\R)$,  
is (the subgroup of the standard Heisenberg group $\He_l(\R^l\times\R^l\times\R)$)  topologically equivalent to $\R^l\times\Z^l\times\R$ with group law
$
(u,t)\cdot (v,s)= (u+v, t+s+\frac12\Omega(u,v))
$. Here $u:=(p,q), p\in\R^l, q\in\Z^l$,  $ t\in\R$ and  $\Omega(u,v)$ is the canonical $2-$form on $\R^l\times\Z^l$: $\Omega(u,v):=\la u_1,v_2\ra-\la v_1,u_2\ra
$.

The unitary representations of $\He_l(\R^l\times\Z^l\times\R)$  in $L^2(\T^l)$ are defined for any $\hbar\neq 0$ as follows
\be
\label{UR}
(U_\hbar(p,q,t)f)(x):=e^{i\hbar t +i\la q,x\ra+\hbar\la p.q\ra/2}f(x+\hbar p)
\ee

Consider now a family of smooth phase-space functions indexed by $\hbar$,   $\A(\xi,x,\hbar):\R^l\times\T^l\times [0,1]\to\C$,  written under its Fourier representation
\vskip 4pt\noindent
\be
\label{FFR}
\A(\xi,x,\hbar)=\int_{\R^l}\sum_{q\in\Z^l}\widehat{\widetilde\A}(p,q;\hbar)e^{i(\la p.\xi\ra +\la q,x\ra)}\,dp
\ee
\vskip 6pt\noindent
\begin{definition}[Weyl quantization I]
\label{Qdef}
{\it By analogy with the usual Weyl quantization on $T^*\R^l$\cite{Fo}, the (Weyl) quantization of  $\A$  is the operator $A(\hbar)$ defined as}
\bea
\label{Wop}
&&
A(\hbar):=(2\pi)^l\int_{\R^l}\sum_{q\in\Z^l}\widehat{\widetilde\A}(p,q;\hbar)U_\hbar(p,q,0)\,dp
\eea
(note that the factor $(2\pi)^l$ in \eqref{Wop} is due to the (convenient for us) normalization of the Fourier transform in Definition \ref{deffour}).
\end{definition}
It is a straightforward computation to show that, considering $f\in L^2(\T^l)$ and $\V((x+\cdot)/2)$ as  periodic functions on $\R^l$, we get the equivalent definition
\begin{definition}[Weyl quantization II]
\label{QdefII}
\be\label{weylII}
(A(\hbar)f)(x):=\int_{\R^l_\xi\times \R^l_y}
\A((x+y)/2,\xi,\hbar)e^{i\frac{\xi(x-y)}\hbar} f(y)\frac{d\xi dy}{(2\pi\hbar)^l}
\ee
\end{definition}
\begin{remark}\label{tropchou}
The expression \eqref{QdefII} is exactly the same as the definition of Weyl quantization on $T^*\R^l$ except 
the fact that $f$ is periodic. Note that $A(\hbar)f$ is periodic thanks to the fact that $\A(x,\xi,\hbar)$ is periodic: 

\noindent 
$
\int
\A((x+2\pi+y)/2,\xi)e^{i\frac{\xi(x+2\pi-y)}\hbar} f(y)\frac{d\xi dy}{\hbar^l}=
\int
\A((x+2\pi+y+2\pi)/2,\xi)e^{i\frac{\xi(x-y)}\hbar} f(y+2\pi)\frac{d\xi dy}{\hbar^l}=
\int
\A((x+y)/2+2\pi,\xi)e^{i\frac{\xi(x-y)}\hbar} f(y)\frac{d\xi dy}{\hbar^l}
=
\int
\A((x+y)/2,\xi)e^{i\frac{\xi(x-y)}\hbar} f(y)\frac{d\xi dy}{\hbar^l}=(A(\hbar))f(x).
$
\end{remark}
The first results concerning this definition are contained in the following Proposition.

\noindent
\begin{proposition} 
\label{corA}
Let $A(\hbar)
$ be defined by the expression \eqref{weylII}. 
 Then:
\begin{enumerate}
\item
$\forall\rho> 0, \forall\,k\geq 0$ we have:
\be\label{stimz}
\Vert A(\hbar)\Vert_{\B(L^2(\T^l))}\leq\Vert\A\Vert
_{\rho,k}
\ee
and, if $\A(\xi, x,\hbar)=\A'(\omega\cdot\xi,x;\hbar)$
\be\label{stimg}
\Vert A(\hbar)\Vert_{\B(L^2(\T^l))}\leq\Vert\A'\Vert_{\rho,k}.
\ee
\item
 Let, for $n\in\Z^l$, $e_n(x)=\frac{e^{inx}}{(2\pi)^l}$. Then for all $m,n$ in $\Z^l$,
\be\label{stimem}
\la e_m,A(\hbar) e_n\ra_{L^2(\T^l)}=\truc{A}((m+n)\hbar/2,m-n,\hbar)
\ee
\item Reciprocally,  let $A(\hbar)$ be an operator whose matrix elements satisfy \eqref{stimem} for some $\A$ belonging to $\J^@, @\in~\{\dagger,m,m\times m\}$. Then $A(\hbar)$ is the Weyl quantization of $\A$.

\end{enumerate}
\end{proposition}
\begin{proof}
\eqref{stimem} is obtained by a simple computation. It also implies that $$\Vert A(\hbar) e_m\Vert_{L^2(\R^l)}^2=\sum\limits_{q\in\Z^l}\vert\truc\A (\hbar(m+q)/2,m-q,\hbar)\vert^2\leq \sup\limits_{\xi\in\R^l}\sum\limits_{q\in\Z^l}\vert\truc\A (\xi,q,\hbar)\vert^2.$$
So that $$\Vert A(\hbar) \sum\limits_{\Z^l}c_me_m\Vert^2_{L^2(\R^l)}\leq \sum\limits_{\Z^l}\vert c_m\vert^2\sup\limits_{\xi\in\R^l}\sum\limits_{q\in\Z^l}\vert\truc\A (\xi,q,\hbar)\vert^2\leq(\sum\limits_{\Z^l}\vert c_m\vert^2)
\left(\sum\limits_{q\in\Z^l}\sup\limits_{\xi\in\R^l}\vert\truc\A (\xi,q,\hbar)\vert\right)^2.$$

\noindent  And therefore, since by 
\eqref{FE0}-\eqref{FE1}-\eqref{FE2} $\widetilde{\A}(\xi,q,\hbar)=\int_{\R^l} \widehat{\widetilde\A}(p,q,\hbar) e^{i<\xi, p>}dp$ 
so that $\vert \widetilde\A(\xi,q,\hbar)\vert$ $
\leq\int_{\R^l}\vert \widehat{\widetilde\A}(p,q,\hbar)\vert dp$,
\be\label{preuvecasse}
\Vert A(\hbar)\Vert_{\B(L^2(\T^l))}\leq
\sum\limits_{q\in\Z^l}\sup\limits_{\xi\in\R^l}\vert\truc\A (\xi,q,\hbar)\vert\leq \int\limits_{\R^l}\sum\limits_{q\in\Z^l}\vert\widehat{\widetilde \A}(p,q,\hbar)\vert dp\leq\Vert\A\Vert_{\rho,k}, \forall\rho>0,k\geq0.
\ee 

In the case $\A(\xi, x,\hbar)=\A'(\omega\cdot\xi,x;\hbar)$ we get, $\forall\rho>0,k\geq0$:
 
\noindent
$$\Vert A(\hbar)\Vert_{\B(L^2(\T^l))}\leq
\sum\limits_{q\in\Z^l}\sup\limits_{\xi\in\R^l}\vert\truc\A (\xi,q,\hbar)\vert=
\sum\limits_{q\in\Z^l}\sup\limits_{Y\in\R^m}\vert\truc\A' (Y,q,\hbar)\vert\leq 
\int\limits_{\R^m}\sum\limits_{q\in\Z^l}\vert\widehat \A'(p,q,\hbar)\vert dp\leq\Vert\A'\Vert_{\rho,k}.$$
(3) is obvious.\end{proof}


\subsection{Fundamental estimates}\label{firsesti}
This section contains the fundamental estimates which will be the blocks of the  estimates needed in the proofs of our main results. These primary estimates are contained in the following Proposition. We shall omit to write the subscript $\underline{\omega}$ in the norms.
\begin{proposition} 
\label{stimeMo}
We have: 
 \begin{enumerate}
 \item[\bf{(1)}] 
 For $F,G\in J_k^1(\rho)$, $FG\in J_k^1(\rho)$ and fulfills the estimate
  \vskip 3pt\noindent
 \be
\label{2conv}
 \|FG\|_{\rho,k}
\leq (k+1)8^k \|F\|_{\rho,k} \cdot \|G\|_{\rho,k}
\ee
 \vskip 3pt\noindent
 \item[\bf{(2)}] 
There exists a positive constant $C'$ such that for $F\in J_k^m(\rho)$ and for $G\in J_k^1(\rho)$,
 we have, $\forall \delta_1> 0,\delta\geq 0, \rho>\delta+\delta_1$,
  \vskip 5pt\noindent
 \be
 \label{normaM2}
\left\Vert \frac{[F,G]}{i\hbar}\right\Vert_{\rho-\delta-\delta_1,k}
\leq \frac{2(k+1)8^k}{e^2\delta_1(\delta+\delta_1)}\|F\|_{\rho,k}
\|G \|_{\rho-\delta,k},\ 
\ee
\be\label{emboites}
\frac 1{d!}\Vert{\underbrace{[G,\dots [G}_{d\ times}, F]\cdots]}/{(i\hbar)^d}\Vert_{\rho-\delta,k}\leq \frac1{2\pi}\left(\frac{2(1+k)8^k}{\delta^2}\right)^d\Vert F\Vert_{\rho,k}\Vert G\Vert_{\rho,k}^d,
\ee
and  
\be\label{emboitesL}
\left\Vert\frac{[L_\omega,G]}{i\hbar}\right\Vert_{\rho-\delta, k}\leq\frac\bid{e\delta}\Vert G\Vert_{\rho,k}
\ee
 \item[\bf{(3)}]
 For $\F,\G\in\J_k^1(\rho)$,
 $\F \G \in \J_k^1(\rho)$ and
\be
\label{simple}
 \|\F\G\|_{\rho,k}
\leq (k+1)4^k  \|\F\|_{\rho,k} \cdot \|\G\|_{\rho,k}.
\ee

\item[\bf{(4)}]
 Let $V=(V_l)_{l=1\dots m}\in J_k^m(\rho)$ and let $W$ be defined by $\la e_m,W e_n\ra=\frac{{\la e_m,V_{l_{m-n}} e_n\ra}}{{\omega_{\ell_{m-n}}\cdot(m-n)}}$  where, 
$\forall m,n\in\Z$, 
the index $l_{m-n}$ is such that $|\la\om_{\ell_{m-n}},m-n\ra|^{-1}:=
 \min\limits_{1\leq i\leq m}|\la\om_i,m-n\ra|^{-1}$ . Then
\be\label{ptidiv}
\Vert W\Vert_{\rho-d,k}\leq \gamma \frac{\tau^\tau}{(ed)^\tau}
\Vert V\Vert_{\rho,k}
\ee
in the Diophantine case and (obviously) when $|\la e_m,V_{l_{m-n}} e_n\ra|=0$ for $|m-n|>M$,
\be\label{brunoptidiv}
\Vert W\Vert_{\rho,k}\leq 
\machin_M
\Vert V\Vert_{\rho,k}
\ee
in the case of the Brjuno condition ($\machin_M$ defined by \eqref{BC}).

\item[\bf{(5)}]\label{new} Let finally $V=(V_l)_{l=1\dots m}\in J_k^m(\rho)$ and let $P$ be defined by
$(P_l)_{ij}=\frac{([ V_l,V_{\ell_{i-j}}])_{ij}}{i\hbar}$ for any choice of $(i,j)\to \ell_{i-j}$. Then $P=(P_l)_{l=1\dots m}\in J_k^m(\rho-\delta),\forall \delta_1\geq 0,\delta>0, \rho>\delta+\delta_1$ and
 \be
 \label{normaM2new}
\| P\|_{\rho-\delta-\delta_1,k}
\leq \frac{2(k+1)8^k}{e^2\delta_1(\delta+\delta_1)}\|V\|_{\rho,k}
\|V \|_{\rho-\delta,k}
\ee
\item[\bf{(6)}] Moreover let $\F: \xi\in\R^m\mapsto\F(\xi)\in\R^m$ be in $\J_k^m(\rho)$. Let us define $\nabla\F$ the matrix $((\nabla\F)_{ij})_{i,j=1\dots m}$ with 
\be\label{mfoism1}
(\nabla\F)_{ij}:=\partial_{\xi_i}\F_j.
\ee Then, for all $\delta>0$, $\nabla\F\in\J^{m\times m}_k(\rho-\delta)$ and 
\be\label{mfoism}
\Vert\nabla\F\Vert_{\rho-\delta,k}\leq \frac1{e\delta}\Vert\F\Vert_{\rho,k}.\ee

\end{enumerate}
 \vskip 3pt\noindent
\end{proposition}

Let us remark that, as the proof will show, Proposition \ref{stimeMo} remains valid when the norm $\Vert\cdot\Vert_{\rho,k}$ is replaced by the norm $\Vert\cdot\Vert^\hbar_{\rho,k}$
\begin{proof}
Items \textbf{(1)} and \textbf{(2)} are simple extension  to the multidimensional case of the corresponding results for $m=1$ proven in \cite{GP}. For sake of completeness we give here an alternative proof in the case $m=1$. The proof will use 
the three elementary inequalities, 
\bea
\mu_k(p+p',q+q')&\leq & 2^\frac k2\mu_k(p,q)\mu_k(p',q')\label{three1}\\
\vert(p.\omega.q'-p'\omega.q)/2\vert^k&\leq&\mu_k(p,q)\mu_k(p',q')\label{three2}\\
\left\vert\partial^k_\hbar\frac{\sin{x\hbar}}\hbar\right\vert&\leq&\vert x\vert^{k+1}\label{three3}\\
\vert p\cdot\omega\cdot q\vert&\leq&\bid\max_{j-1\dots m}\vert p_j\vert\vert q\vert\leq\bid\vert p\vert\vert q\vert\label{three4}
\eea
where we have used the notation 
\eqref{pomegaq} and the definition \eqref{normomega}.

(in order to prove \eqref{three1}, \eqref{three2}, \eqref{three3} and \eqref{three4} just use  $\vert X+X'\vert^2\leq2(\vert X\vert^2+\vert X'\vert^2)\mbox{ for all }X,X'\in\R^{2l}
$, $\vert(p.\omega.q'-p'\omega.q)/2\vert^2\leq(\pomv^2+\vert q\vert^2)(\vert \pom'\vert^2+\vert q'\vert^2)$,
$\frac{\sin{x\hbar}}\hbar=\int\limits_0^x\cos{(s\hbar)}ds$  and $\vert p\cdot\omega\cdot q\vert
\leq \sum\limits_{j=1}^m\vert p_j\vert\vert\sum\limits_{i=1}^l\omega_j^i q_i\vert
=  \sum\limits_{j=1}^m\vert p_j\vert\vert\omega_j\cdot q\vert
\leq \sum\limits_{j=1}^m\vert p_j\vert\vert\omega_j\vert\vert q\vert$ by Cauchy-Schwarz, respectively).

We start with \eqref{2conv}. Since $F,G\in J^1(\rho)$ we know that there exist two functions $\F',\G'$ such that the symbols of $F,G$ are $\F(\xi,x)=\F'(\omega.\xi,x),\ \G(\xi,x)=\G'(\omega.\xi,x)$. By \eqref{stimem} we have that
\begin{eqnarray}
(FG)_{mn}&=&\sum_{q'\in\Z^l}F_{mq'}G_{q'n}\nonumber\\
&=&\sum_{q'\in\Z^l}\widetilde\F\left(\frac{m+q'}2\hbar,m-q'\right)\widetilde\G\left(\frac{q'+n}2\hbar,q'-n\right)\nonumber\\
&=&\sum_{q'\in\Z^l}\widetilde\F\left(\frac{m+n+q'}2\hbar,m-n-q'\right)\widetilde\G\left(\frac{q'+2n}2\hbar,q'\right)\nonumber\\
&=&\sum_{q'\in\Z^l}\widetilde\F'\left(\omega\cdot\frac{m+n+q'}2\hbar,m-n-q'\right)\widetilde\G'\left(\omega\cdot\frac{q'+2n}2\hbar,q'\right)\nonumber\\
&=&
\sum_{q'\in\Z^l}\widetilde\F'\left(\omega\cdot\frac{m+n+q'}2\hbar,m-n-q'\right)\widetilde\G'\left(\omega\cdot\frac{q'+m+n-(m-n)}2\hbar,q'\right).\label{en plus}
\end{eqnarray}
Calling $\P$ the symbol of $FG$ we have that, by \eqref{stimem} again, $(FG)_{mn}=\widetilde\P(\xi,q)$ with $\xi=\frac{m+n}2\hbar$ and $q=m-n$. Therefore
\be
\widetilde\P(\xi,q)=
\sum_{q'\in\Z^l}\widetilde\F'\left(\omega.\xi+\omega\cdot\frac{q'}2\hbar,q-q'\right)\widetilde\G'\left(\omega.\xi+\omega\cdot\frac{q'-q}2\hbar,q'\right),
\ee
so we see that $\P(\xi,\cdot)$ depends only on $\omega.\xi$: $\P(\xi,x)=\P'(\omega.\xi,x)$. Moreover, since by \eqref{check} $\widehat{\widetilde\P'}(p,\cdot)=\frac1{(2\pi)^m}\intm\widetilde{\P'}(\Xi,\cdot)e^{-i<\Xi,p>}d\Xi$ we get easily by simple changes of integration variables and the fact that the Fourier transform of a product is a convolution,
\be\label{preum}
{\widehat{\widetilde\P'}} (p,q)=
\intm\sumqql\left(\widehat{\widetilde\F'}(p-p',q-q')
e^{i\frac\hbar2(p-p').\omega.q'}\right)\left(
\widehat{\widetilde\G'}(p',q')
e^{i\frac\hbar2p'.\omega.(q'-q))}\right)dp'.
\ee   
Therefore $\Vert FG\Vert_{\rho,k}$ is equal to the maximum over ${\hbar\in [0,1]}$ of 
\be\label{524}
\sum_{\gamma=0}^k\intdm\sum_{(q,q')\in\Bbb Z^{2l}}
\mu_{k-\gamma}(p,q)\vert
\partial^\gamma_\hbar\left[\widehat{\widetilde\F'}(p-p',q-q')
e^{i\frac\hbar2((p-p').\omega.q'-p'.\omega.(q-q'))}
\widehat{\widetilde\G'}(p',q')\right]\vert 
e^{\rho(\bid|p|
+|q|
)}dpdp'.\ \ \ \ \ \ \ \ \ \ \ \ \ \ \ \ \ \ \ \ \ \ \ \
\ee
Writing, by \eqref{three2}, that
\bea\label{deus}
&&\vert\partial^\gamma_\hbar\left[\widehat{\widetilde\F'}(p-p',q-q')
e^{i\frac\hbar2((p-p').\omega.q'-p'.\omega.(q-q'))}
\widehat{\widetilde\G'}(p',q')\right]\vert\nonumber\\
&\leq&\sum_{\mu=0}^\gamma\binom{\gamma}{\mu}\sum_{\nu=0}^{\gamma-\mu}\binom{\gamma-\mu}{\nu}
\vert\partial^{\gamma-\mu-\nu}_\hbar\widehat{\widetilde\F'}(p-p',q-q')\vert
\vert\partial^{\nu}_\hbar e^{i\frac\hbar2((p-p').\omega.q'-p'.\omega.(q-q'))}\vert
\vert\partial^{\mu}_\hbar\widehat{\widetilde\G'}(p',q')\vert\nonumber\\
&\leq&\sum_{\mu=0}^\gamma\binom{\gamma}{\mu}\sum_{\nu=0}^{\gamma-\mu}\binom{\gamma-\mu}{\nu}
\vert\partial^{\gamma-\mu-\nu}_\hbar\widehat{\widetilde\F'}(p-p',q-q')\vert
\vert((p-p').\omega.q'-p'.\omega.(q-q'))/2\vert^\nu
\vert\partial^{\mu}_\hbar\widehat{\widetilde\G'}(p',q')\vert\nonumber\\
&=:&\mathbb P(\F',\G')\label{525}\\
&\leq&\sum_{\mu=0}^\gamma\binom{\gamma}{\mu}\sum_{\nu=0}^{\gamma-\mu}\binom{\gamma-\mu}{\nu}
\mu_\nu(p-p',q-q')\mu_\nu(p',q')
\vert\partial^{\gamma-\mu-\nu}_\hbar\widehat{\widetilde\F'}(p-p',q-q')\vert
\vert\partial^{\mu}_\hbar\widehat{\widetilde\G'}(p',q')\vert\nonumber\\
&\ &(\mbox{changing } \mu\to\gamma',\ \nu\to\nu':=\gamma-\gamma'-\nu)\nonumber\\
&\leq&
\sum_{\gamma'=0}^\gamma\binom{\gamma}{\gamma'}\sum_{\nu'=0}^{\gamma-\gamma'}\binom{\gamma-\gamma'}{\nu}
\mu_{\gamma-\gamma'-\nu'}(p-p',q-q')\mu_{\gamma-\gamma'-\nu'}(p',q')
\vert\partial^{\nu'}_\hbar\widehat{\widetilde\F'}(p-p',q-q')\vert
\vert\partial^{\gamma'}_\hbar\widehat{\widetilde\G'}(p',q')\vert\nonumber\\
& &
(\mbox{since }\binom{m}{n}\leq2^m, \gamma\leq k,\gamma-\gamma'\leq k)\nonumber\\
&\leq&
\sum_{\gamma'=0}^k2^k\sum_{\nu'=0}^{k}2^k
\mu_{\gamma-\gamma'-\nu'}(p-p',q-q')\mu_{\gamma-\gamma'-\nu'}(p',q')
\vert\partial^{\nu'}_\hbar\widehat{\widetilde\F'}(p-p',q-q')\vert
\vert\partial^{\gamma'}_\hbar\widehat{\widetilde\G'}(p',q')\vert\nonumber\\
\eea
using \eqref{three1} under the form
\[
\mu_k(p,q)\leq  2^\frac k2\mu_k(p-p',q-q')\mu_k(p',q')
\]
together with the fact that $\mu_k(p,q)$ is  increasing in $k$ and $\mu_k\mu_{k'}=\mu_{k+k'}$.

We find that 
\bea
&&\mu_{k-\gamma}(p,q)
\mathbb P(\F',\G')\nonumber\\
&\leq&
4^k2^{\frac k2}
\sum_{\gamma'=0}^k\sum_{\nu'=0}^{k}
\mu_{k-\gamma+\gamma-\gamma'-\nu'}(p-p',q-q')\mu_{k-\gamma+\gamma-\gamma'-\nu'}(p',q')
\vert\partial^{\nu'}_\hbar\widehat{\widetilde\F'}(p-p',q-q')\vert
\vert\partial^{\gamma'}_\hbar\widehat{\widetilde\G'}(p',q')\vert\nonumber\\
&&(\mbox{replacing }2^{\frac k2}\mbox{ by }2^k \mbox{ to avoid heavy notations and since }k-\gamma'-\nu'\leq k-\gamma',\ k-\nu') \nonumber\\
&\leq&8^k
\sum_{\gamma'=0}^k\sum_{\nu'=0}^{k}
\mu_{k-\nu'}(p-p',q-q')\mu_{k-\gamma'}(p',q')
\vert\partial^{\nu'}_\hbar\widehat{\widetilde\F'}(p-p',q-q')\vert
\vert\partial^{\gamma'}_\hbar\widehat{\widetilde\G'}(p',q')\vert.\label{524bis}
\eea
Note that $\gamma$ disappeared  from \eqref{524bis} so the $\sum\limits_{\gamma=0}^k$ in \eqref{524} gives a fractor $(1+k)$. We get that
$(1+k)^{-1} 8^{-k}\Vert FG\Vert_{\rho,k}$ is majored by the maximum over $\hbar\in[0,1]$ (note the change $\nu'\to\gamma$)
\bea\label{tresss}
\sum_{q,q'\in\Z^l}\intdm\sum_{\gamma,\gamma'=0}^k\mu_{k-\gamma}(p,q)\vert\partial^\gamma_\hbar\widehat{\widetilde\F'}(p,q)\vert
\mu_{k-\gamma'}(p',q')\vert\partial^{\gamma'}_\hbar\widehat{\widetilde\G'}(p',q')\vert e^{\rho(\bid|p|+\bid|p'|+|q|+|q'|)} dpdp'
\eea
which is equal to 
$$
\Vert\F'\Vert_{\rho,k}\Vert\G'\Vert_{\rho,k}.
$$
The proof of \eqref{normaM2} follows the same lines, except that it is easy to see that,  in 
\eqref{preum}, $e^{i\frac\hbar2((p-p').\omega.q'-p'.\omega.(q-q'))}$ has to be replaced by
$2\sin{\left(\frac\hbar2((p-p').\omega.q'-p'.\omega.(q-q'))\right)}$, since \eqref{en plus} becomes
\begin{eqnarray}
\left(\frac{[F,G]}{i\hbar}\right)_{mn}&=&
\sum_{q'\in\Z^l}\frac{F_{mq'}G_{q'n}-G_{mq'}F_{q'n}}{i\hbar}\nonumber\\
&=&\frac1{i\hbar}
\sum_{q'\in\Z^l}\left[\widetilde\F\left(\frac{m+n+q'}2\hbar,m-n-q'\right)\widetilde\G\left(\frac{m+n+q'-(m-n)}2\hbar,q'\right)\right.\nonumber\\
&&-\left.\widetilde\G\left(\frac{m+n+q'}2\hbar,m-n-q'\right)\widetilde\G\left(\frac{m+n+q'-(m-n)}2\hbar,q'\right)\right]
\end{eqnarray}

 It generates in \eqref{524} the  replacement of $\vert((p-p').\omega.q'-p'.\omega.(q-q'))/2\vert^\nu$by the term
\[
2\vert((p-p').\omega.q'-p'.\omega.(q-q'))/2\vert^{\nu+1}\leq
\mu_\nu(p-p',q-q')\mu_\nu(p',q')
(\vert p\cdot\omega\cdot q'-p'\cdot\omega\cdot q\vert)
\]
thanks to \eqref{three2},
and we get by  a discussion verbatim the same  than the one contained in equations \eqref{525}-\eqref{tresss} that

\be\label{555}
\Vert[F,G]/i\hbar\Vert_{\rho-\delta-\delta_1,k} \leq(1+k) 8^k
\sum_{q,q'\in\Z^l}\intdm
\sum_{\gamma,\gamma'=0}^k\mu_{k-\gamma}(p,q)
\mu_{k-\gamma'}(p',q')\mathbb Q dpdp',
\ee
where thanks to   \eqref{three4},
\bea\nonumber
\mathbb Q&=&\vert\partial^\gamma_\hbar\widehat{\widetilde\F'}(p,q)\vert
(\vert \pom\cdot q'\vert+\vert p'\cdot\omega q\vert)\vert\partial^{\gamma'}_\hbar\widehat{\widetilde\G'}(p',q')\vert e^{(\rho-\delta-\delta_1)(\bid|p|+|q|+\bid|p'|+|q'|)} \\
&\leq&\left[\vert\partial^\gamma_\hbar\widehat{\widetilde\F'}(p,q)\vert
\vert\partial^{\gamma'}_\hbar\widehat{\widetilde\G'}(p',q')\vert e^{\rho(\bid|p|+|q|)+(\rho-\delta)(\bid|p'|+|q'|)}\right]\times\nonumber\\
&&\times
(\bid\vert p\vert\vert q'\vert+\bid\vert p'\vert\vert q\vert)
e^{-(\delta+\delta_1)(\bid|p|+|q|)-\delta_1(\bid|p'|+|q'|))}\nonumber\\
&\leq&
 \frac 2{e^2\delta_1(\delta+\delta_1)}
 \vert\partial^\gamma_\hbar\widehat{\widetilde\F'}(p,q)\vert
\vert\partial^{\gamma'}_\hbar\widehat{\widetilde\G'}(p',q')\vert e^{\rho(\bid|p|+|q|)+(\rho-\delta)(\bid|p'|+|q'|)}\label{jgd}
\eea
because ($e^{-x}\leq 1, x\geq 0$ and)
\be\label{truccc}
\sup_{x\in\mathbb R^+}{xe^{-\alpha x}}=\frac1{e\alpha}.
\ee
  \eqref{normaM2} follows immediatly from \eqref{555}.

The proof of \eqref{emboitesL} follows also the same line and is obtained thanks to the remark \eqref{truccc}: indeed since 
\[
\left(\frac{[L_\omega,W]}{i\hbar}\right)_{mn}=-i\omega\cdot (m-n)W_{mn},
\]
we see, again by \eqref{stimem}, that the symbol $\mathcal Q(\xi,x)$ of $[L_\omega,W]/i\hbar$ is  given trough the formula
\[
\widetilde{\mathcal Q}(\xi,q)=([L_\omega,W]/i\hbar)_{mn} \mbox{ for } \xi=\frac{m+n}2\hbar\mbox{ and }q=m-n. 
\]
Therefore $\widetilde{\mathcal Q}(\xi,q)=(-i\omega\cdot q)\widetilde{\W}(\xi,q)$, so $\mathcal Q(\xi,x)=\mathcal Q'(\omega.\xi,x)$ with
\[
\widetilde{\mathcal Q'}(\omega.\xi,q)=-i\omega.q\ \widetilde\W'(\omega.\xi,q).
\]
We get immediatly
\[
\widehat{\widetilde{\mathcal Q'}}(p,q)e^{(\rho-\delta)(\bid|p|+|q|)}\leq \frac\bid{e\delta}
\widehat{\widetilde\W'}(p,q)e^{\rho(\bid|p|+|q|)}
\]
and \eqref{emboitesL} follows.

\eqref{emboites} is easily obtained by iteration of \eqref{normaM2} and the Stirling formula: 
consider the finite sequence of numbers $\delta_s=\frac{d-s}{d}\delta$. We have $\delta_0=\delta$, $\delta_d=0$ and $\delta_{s-1}-\delta_s=\frac{\delta}{d}$. Let us define $G_0:= F$ and $G_{s+1}:=\frac{1}{i\hbar}[G,G_s]$, for $0\leq s \leq d-1$.
According to \eqref{normaM2}, we have 
\bea\nonumber
\Vert G_{s} \Vert_{\rho-\delta_{d-s}, k} \leq \frac{c_k}{e^2\delta_{d-s}(\frac{\delta}{d})}\Vert G \Vert_{\rho,  k}\Vert G_{s-1} \Vert_{\rho-\delta_{d-s+1},  k},
\eea
where 
$$
c_k:=2(k+1)8^k.
$$
Hence, by induction, we obtain
%
\begin{eqnarray}
\frac{1}{d!}\Vert G_{d} \Vert_{\rho-\delta_{0},  k} &\leq &\frac{c_k^{d-1}}{d! e^{2(d-1)}\delta_{0}\cdots\delta_{d-2}(\frac{\delta}{d})^{d-1}}\Vert G \Vert_{\rho,  k}^{d-1}\Vert G_{1} \Vert_{\rho-\delta_{d-1},  k}\label{encore}\\
&\leq & \frac{c_k^{d}}{d! e^{2d}\delta_{0}\cdots\delta_{d-1}\delta_{d-1}(\frac{\delta}{d})^{d-1}}\Vert G \Vert_{\rho,  k}^{d}\Vert F \Vert_{\rho,  k}\nonumber\\
&\leq & \frac{c_k^{d}}{d! e^{2d}d!(\frac{\delta}{d})^{d} \delta_{d-1}(\frac{\delta}{d})^{d-1}}\Vert G \Vert_{\rho,  k}^{d}\Vert F \Vert_{\rho,  k}\nonumber\\
&\leq & \frac1{2\pi}\left(\frac{c_kd^2}{e^2\delta^2}\right)^d\frac1{ (d-1)!d!}\Vert G \Vert_{\rho,  k}^{d}\Vert F \Vert_{\rho,  k}\nonumber\\
&=&\frac1{2\pi}\left(\frac{c_k}{\delta^2}\right)^d
\left(\frac{\sqrt{2\pi d}d^de^{-d}}{ d!}\right)^2\Vert G \Vert_{\rho,  k}^{d}\Vert F \Vert_{\rho,  k}\nonumber\\
&\leq&\frac1{2\pi}\left(\frac{c_k}{\delta^2}\right)^d
\Vert G \Vert_{\rho}^{d}\Vert F \Vert_{\rho,  k}\nonumber
\end{eqnarray}
since $\frac{d!}{\sqrt{2\pi d}e^{-d}d^d}\geq 1$. This weel know inequality can be seen from Binet's second expression for the $\log \Gamma (z)$ \cite{WW}[p. 251]~:
$$
\log \left(\frac{n!}{\left(\frac{n}{e}\right)^n\sqrt{2\pi n}}\right)=2\int_0^\infty \frac{\arctan (t/n)}{e^{2\pi t}-1}dt \geq 0.
$$

Finally \eqref{simple} is obtained by noticing that $\Vert\F\G\Vert_{\rho,k}$ has the same expression as $\Vert FG\Vert_{\rho,k}$ after removing the term $e^{i\frac\hbar2(p.\omega.q'-p'.\omega.q)}$ in \eqref{preum}.

To prove \textbf{(4)} it is enough to notice that by \eqref{stimem} the symbol of $W$ satisfies $\widetilde W(\xi,q,\hbar)=\frac{\widetilde V_{\ell_q}(\xi,q,\hbar)}{{\omega_{\ell_{q}}\cdot q}}$, so that $\widehat{\widetilde{ W}}(p,q,\hbar)
=\frac{\widehat{\widetilde{V}}1_{\ell_q}(p,q,\hbar)}{{\omega_{\ell_{q}}\cdot q}}$ and therefore, for all $r\in\mathbb N $,
\[
\vert \partial_\hbar^r\widehat{\widetilde {W}}(p,q,\hbar)\vert\leq \gamma\vert q\vert^\tau \sup_{l=1\dots m}\vert\partial_\hbar^r \widehat{\widetilde {V}}_{l}(p,q,\hbar)\vert
\leq \gamma\vert q\vert^\tau \sum_{l=1}^{m}\vert\partial_\hbar^r \widehat{\widetilde {V}}_{l}(p,q,\hbar)\vert
\]
out of which we deduce \eqref{ptidiv} by standard  arguments ($x^\tau e^{-\delta x}\leq (\frac\tau{e\delta})^\tau,\ x>0$) in the Diophantine case, and
\[
\vert \partial_\hbar^r\widehat{\widetilde {W}}(p,q,\hbar)\vert\leq \machin_M \sup_{l=1\dots m}\vert\partial_\hbar^r \widehat{\widetilde {V}}_{l}(p,q,\hbar)\vert 
\leq \machin_M \sum_{l=1}^{m}\vert\partial_\hbar^r \widehat{\widetilde {V}}_{l}(p,q,\hbar)\vert 
\]
from which \eqref{brunoptidiv} follows. 

\noindent To prove \eqref{mfoism} we just notice that
$\widehat{\widetilde{\partial_{\xi_i}\F_j}}(p,q,\hbar)=p_i\widehat{\widetilde{\F_j}}(p,q,\hbar)$. So $$\vert\widehat{\widetilde{\partial_{\xi_i}\F_j}}(p,q,\hbar)\vert\leq \vert p_i\widehat{\widetilde{\F_j}}(p,q,\hbar)\vert\leq \vert \widehat{\widetilde{\F_j}}(p,q,\hbar)\vert\vert p\vert.$$ Therefore
$\vert\partial^r_\hbar\widehat{\widetilde{\partial_{\xi_i}\F_j}}(p,q,\hbar)\vert e^{(\rho-\delta)\vert p\vert}\leq \frac1{e\delta}
\vert\partial^r_\hbar\widehat{\widetilde{\partial_{\xi_i}\F_j}}(p,q,\hbar)\vert e^{\rho\vert p\vert}$ and \eqref{mfoism} follows.

\textbf{(5)} is an easy extension of 
\eqref{normaM2}. Indeed we find immediately, by \eqref{stimem} and the fact that $l_{i-j}$ depends only on $i-j$,  that the Fourier transform of the symbol of $P_l$ is 
$\widehat{\widetilde{\mathcal P}}(p,q,\hbar)=\widehat{\widetilde{\mathcal X}}_{l_q}(p,q,\hbar)
$ 
where $\widehat{\widetilde{\mathcal X}}_{l_q}$ is the Fourier transform of the symbol of the operator $X_{l_q}=\frac{[V_l,V_{l_q}]}{i\hbar}$. Therefore $\vert\partial^r_\hbar\widehat{\widetilde{\mathcal  X}}_{l_q}(p,q,\hbar)\vert\leq\max\limits_{l=1\dots m}\vert\partial^r_\hbar\widehat{\widetilde{\mathcal  X}}_{l}(p,q,\hbar)\vert, \forall r\geq0, q\in\Z^l$. So $\Vert P_l\Vert_{\rho-\delta,  k}\leq \max\limits_{l'=1\dots m}\Vert [V_l,V_{l'}]/i\hbar\Vert_{\rho-\delta,  k}$ and 
$$
\Vert P\Vert_{\rho-\delta,  k}\leq \max\limits_{l,l'=1\dots m}\Vert [V_l,V_{l'}]/i\hbar\Vert_{\rho-\delta,  k}\leq \sum\limits_{l,l'=1}^{m}\Vert [V_l,V_{l'}]/i\hbar\Vert_{\rho-\delta,  k}
$$ 
and we conclude by using \eqref{normaM2}.
\end{proof}

\section{Fundamental iterative estimates: Brjuno condition case}\label{brunofundamental}
\textbf{In all this section the norm subscripts $\bid$ and $k$ are omitted.}

Let us recall from Sections \ref{start} and \ref{form} that we want to find $W_r$ such that
\be\label{think}
e^{i\frac{W_r}\hbar}(H_r+V_r)e^{-i\frac{W_r}\hbar}=H_{r+1}+V_{r+1}\ee
where $H_{r+1}=H_r+h_{r+1}$ and $H_r=\B_r(L_\om),\ h_{r+1}=\overline{V_r}=\mathcal D_r(L_\om)$ and, for $0<\delta<\rho<\infty$,
\be
 \Vert h_{r+1}\Vert_{\rho }=\Vert \overline{V_r}\Vert_{\rho }\leq\Vert V_{r}\Vert_{\rho },\ \ \ \ \Vert V_{r+1}\Vert_{\rho -\delta }\leq D_r\Vert V_{r}\Vert_{\rho }^2,
 \ee
 \vskip 1cm
\noindent and that we  look at $W_r$ solving:
 \be\label{629}
 \frac 1{i\hbar}[H_r,W_r]+V^{co,r}=\overline{V^{co,r}}+\widehat V^r
 \ee
 with
 \[
 V^{co,r}=V_r-V^{M_r}.
 \]
 \noindent$V^{M_r}$ is given by 
 
 \be\label{cuto}
 V^{M_r}_{ij}=(V_r)_{ij} \mbox{ if }\vert i-j\vert>M_r,\ \ \  V^{M_r}_{ij}=0\mbox{ otherwise.}
 \ee
 \vskip 0.5cm
  (note that $\overline{V^{co,r}}=\overline{V^r}$). 
  
  \noindent$\widehat V^r=(\widehat V^r_l)_{l=1\dots m}$ is given by
 \be\label{kjh}
 (\widehat V^r_l)_{ij}=\frac{([\widetilde V^r_l, \widetilde V^r_{l(i-j)}])_{ij}}{i\hbar\omega_{l(i-j)}\cdot(i-j)},
 \ \ \ \ \ \widetilde V^r_{ij}:=\left(I+
  A^r(i,j)\right)^{-1}V^{co}_{ij},\ \ \ \ \ \ V^{co,r}=V_r-V^{M_r},
 \ee
where $A^r(i,j)$ is the matrix given by Lemma \ref{matruc}, that is: 
$$
\frac{\B_r(\hbar\omega\cdot i)-\B_r(\hbar\omega\cdot j)}{i\hbar}
=\left(I+ A^r(i,j)\right)\omega.(i-j).
$$
Let
\be\label{GL}
Z_k={2(k+1)8^{k}}.
\ee

\vskip 1cm

Let us denote $\text{ad}_{W}$ the operator $H\mapsto [W,H]$.
The l.h.s. of \eqref{think} is then:
\[
H_r+V^r+\frac1{i\hbar}[H_r,W_r]+\sum_{j=1}^\infty\frac1{(-i\hbar)^j j!}\text{ad}^j_{W_r}(V_r)+\sum_{j=2}^\infty\frac1{(-i\hbar)^j j!}\text{ad}^j_{W_r}(H_r)
\]
that is
\[
H_{r}+\overline{V^{co,r}}+\widehat V^r+V^r-V^{co,r}+
\sum_{j=1}^\infty\frac1{(-i\hbar)^j j!}\text{ad}^j_{W_r}(V_r)+\sum_{j=2}^\infty\frac1{(-i\hbar)^j j!}\text{ad}^j_{W_r}(H_r).
\]
or
\be\label{touut}
H_{r}+h_{r+1}+(V_r-V^{co,r})+\widehat V_r+R_1+R_2
\ee
Let us set
\be\label{vr+1}
V_{r+1}:= (V_r-V^{co,r})+\widehat V^r+R_1+R_2.
\ee
We want to estimate $V_{r+1}$.
We first prove the following proposition.
\begin{proposition}\label{fund}
Let $W$ be in $J_k(\rho)$ and $0<\delta<\rho$. Then

\be\label{d1}
\Vert [H_r,W]/i\hbar\Vert_{\rho-\delta}\leq
\frac 1{e\delta}(\bid+Z_k\Vert \nabla(\B^\hbar_r-\B^\hbar_0)\Vert_\rho)\Vert W\Vert_\rho.
\ee
and for $d\geq 2$, 
\be\label{emboitesH}
\frac 1{d!}\Vert[H_r,\underbrace{W],\dots}_{d\ times}]/(i\hbar)^d\Vert_{\rho-\delta}\leq
\frac{\delta\bid}{2\pi Z_k}(1+Z_k\Vert \nabla(\B^\hbar_r-\B^\hbar_0)\Vert_\rho)\left(\frac{Z_k}{\delta^2_r}\right)^d\Vert W\Vert_{\rho}^d
\ee
Let now  $W_r$ be the (scalar) solution of 
\eqref{629}. Then, we have 
\be\label{estiw}
\Vert W_r\Vert_{\rho}\leq 
\frac{\machin_{M_r}}{1-
Z_k\Vert D(\B^\hbar_r-\B^\hbar_0)\Vert_{\rho}}
\Vert V_r\Vert_{\rho},
\ee
and therefore for $d\geq 2$, 
\be\label{estiwd}
\frac 1{d!}\Vert[H_r,\underbrace{W_r],\dots}_{d\ times}]/(i\hbar)^d\Vert_{\rho-\delta}\leq
\bid\frac{
1+Z_k\Vert D(\B^\hbar_r-\B^\hbar_0)\Vert_{\rho}}{2\pi Z_k/\delta}
\left(
\frac{Z_k\machin_{M_r}/\delta_r^2}{1-
Z_k\Vert D(\B^\hbar_r-\B^\hbar_0)\Vert_{\rho}}
\Vert V_r\Vert_{\rho}
\right)^d
\ee
($\machin_M$ is defined in \eqref{BC} and $\Vert D\B\Vert_\rho$ is meant for $\max\limits_{i=1\dots m}
\sum\limits_{j=1\dots m}
\Vert\nabla_i
\B_j\Vert_\rho$).
\end{proposition}
\begin{proof}
We first 
prove \eqref{d1}.
Note that the proof is somehow  close to the proof of Proposition \ref{stimeMo}, items (1) and (2). 

Since $\B^0(L_\omega)=L_\omega$, \eqref{emboitesL} reads
\be\label{primo}
\Vert [H_0,W_r]/i\hbar\Vert_{\rho-\delta}
\leq\frac \bid{e\delta}
\Vert W_r\Vert_{\rho}.
\ee
Note that    $([H_r-H_0,W_r]/\hbar)_{ij}=
\frac{\mathcal G^r(\omega.i\hbar)-\mathcal G^r(\omega.j\hbar)}\hbar W_{ij}$ where $\G^r(Y)=\B^\hbar_r(Y)-Y,\ Y\in\R^m$ (note that $\G^r$
has an explicit dependence in $\hbar$ that we omit to avoid heaviness of notations). 
Indeed, since each $L_{\omega_i}$ is self-adjoint on $L^2(\T^l)$, 
$\B^\hbar_r(L_{\omega})$ can be defined by the spectral theorem. Hence, we have
\begin{eqnarray*}
[\B^\hbar_r(L_{\omega}),W]_{ij}= (e_i,[\B^\hbar_r(L_{\omega}),W]e_j)&=&(e_i,\B^\hbar_r(L_{\omega})We_j-W\B^\hbar_r(L_{\omega})e_i)\\
&=&(\B^\hbar_r(\omega.i\hbar)-\B^\hbar_r(\omega.j\hbar))(e_i,We_j).
\end{eqnarray*}
Using \eqref{simple} we get that
\be\label{hacher11}
\Vert [H_r-H_0,W_r]/i\hbar\Vert_{\rho -\delta}
\leq 
\Vert X_r\Vert_{\rho-\delta}.
\ee
where $X_r$ is defined through $(X_r)_{ij}=\frac{\mathcal G^r(\omega.i\hbar)-\mathcal G^r(\omega.j\hbar)}\hbar (W_r)_{ij}$. 

In order to estimate the norm of $X^r$ we need to  express its symbol $\X_r$. This is done thanks to formula \eqref{stimem} and the fact that we know the matrix elements of $X_r$.

\noindent Expressing $(X_r)_{ij}$ as a function of $((i+j)\hbar/2,i-j)$ through $i,j=\frac{i+j}2\pm\frac{i-j}2$ and using \eqref{stimem} we get that 
\[
\widetilde{\X_r}(\xi,q,\hbar)=
\frac{\mathcal G^r(\omega.\xi+\omega.q\hbar/2)-\mathcal G^r(\omega.\xi-\omega.q\hbar/2)}\hbar \widetilde{\W'_r}(\omega.\xi,q,\hbar):=\widetilde{\X'_r}(\omega.\xi,q,\hbar),
\]
so that, using (remember that we denote  $p.\omega.q=\sum\limits_{j=1\dots m} 
\sum\limits_{i=1\dots l}p_j\omega_j^iq_i$)
\[
\intm (\mathcal G^r(\Xi+\omega.q\hbar/2)-\mathcal G^r(\Xi-\omega.q\hbar/2))e^{-i<\Xi,p>}dp=
2\sin{[p.\omega.q\hbar/2]}\intm \mathcal G^r(\Xi)e^{-i<\Xi,p>}dp,
\]
\[
\widehat{\widetilde{\X'_r}}(p,q,\hbar)=
\int_{\R^m}\widehat\G_i^r(p-p')\frac{\sin{[(p-p').\omega.q\hbar/2]}}{\hbar}\widehat{\widetilde{\W'}}_r(p',q,\hbar)dp.
\]
Therefore
$\Vert X_r\Vert_{\rho-\delta}$ is equal to 

\bea\label{mnbv}
&&\sum_{i=1}^{m}\sumql\intdm  dp dp' \vert 2\sum_{\gamma=1}^k\mu_{k-\gamma}(p,q)\partial^\gamma_\hbar\left[\widehat\G_i^r(p-p')\frac{\sin{[(p-p').\omega.q\hbar/2]}}{\hbar}\widehat{\widetilde{\W}}_r(p',q,\hbar)\right]\vert e^{(\rho-\delta)(\bid\vert p\vert+\vert q\vert)}\nonumber\\
&\leq&\sum_{i=1}^{m}
\sumql\int
\sum_{\gamma=1}^k\mu_{k-\gamma}(p,q)
\sum_{\mu=1}^\gamma
\sum_{\nu=1}^{\gamma-\mu}\binom{\gamma}{\mu}\binom{\gamma-\mu}{\nu}\times\nonumber\\
&&\times 2\vert\partial^{\gamma-\mu-\nu}_\hbar\widehat\G_i^r(p-p')\vert\vert
\partial^{\nu}_\hbar\frac{\sin{((p-p').\omega.q\hbar/2)}}{\hbar}
\vert\vert
\partial^{\mu}_\hbar \widehat{\widetilde{\W}}_r(p',q,\hbar)\vert e^{(\rho-\delta)(\bid\vert p\vert
+\vert q\vert)}dpdp'\label{coupee}\\
\eea
Using now the  inequalities  \eqref{three3} and \eqref{three4}, 
we get, 

\be\nonumber
\vert\sum_{\gamma=1}^k\mu_{k-\gamma}
(p,q)
\sum_{\mu=1}^\gamma
\sum_{\nu=1}^{\gamma-\mu}\binom{\gamma}{\mu}\binom{\gamma-\mu}{\nu}
\partial^{\gamma-\mu-\nu}_\hbar\widehat\G_i^r(p-p')
\partial^{\nu}_\hbar\frac{\sin{((p-p').\omega.q\hbar)}}{\hbar}
\partial^{\mu}_\hbar \widehat{\widetilde{\W}}_r(p',q,\hbar)\vert
\ee
\be
\leq\bid\max\limits_{j=1\dots m}\vert p_j-p'_j\vert\vert q\vert
\sum_{\gamma=1}^k\mu_{k-\gamma}
(p,q)
\sum_{\mu=1}^\gamma
\sum_{\nu=1}^{\gamma-\mu}\binom{\gamma}{\mu}\binom{\gamma-\mu}{\nu}
\vert\partial^{\gamma-\mu-\nu}_\hbar\widehat\G_i^r(p-p')\vert
\vert (p-p').\omega.q\vert^\nu
\vert\partial^{\mu}_\hbar \widehat{\widetilde{\W}}_r(p',q,\hbar)\vert.\ \ \ \ \ \ \ \ \ \ \ \nonumber
\ee
Therefore we notice (after changing $q\leftrightarrow q'$)  that $\Vert X_r\Vert_{\rho-\delta}$ is majored by the maximum over $\hbar\in[0,1]$ of
\be
\sum_{\gamma=0}^k\intdm\sum_{(q,q')\in\Bbb Z^{2l}}
\mu_{k-\gamma}(p,q)
 \bid\max\limits_{j=1\dots m}\vert p_j-p'_j\vert
 \mathbb P(\underline{\G^r_i},\W_r)
e^{(\rho-\delta)(\bid|p|
+|q|
)}dpdp'
\ee
where $\mathbb P$ is defined in \eqref{525} and
\[
\underline{\G^r_i}(\Xi,x)=\G^r_i(\Xi)\mbox{ so that }
\widehat{\widetilde{\underline\G^r_i}}(p,q)=\widehat{\G^r_i}(p)\delta_{q=0}.
\]
Therefore we can verbatim use the argument contained between  formulas \eqref{525}-\eqref{524bis} and  we arrive, in analogy with \eqref{tresss}, to the fact that
$(1+k)^{-1} 8^{-k}\Vert X_r\Vert_{\rho-\delta}$ is majored by the maximum over $\hbar\in[0,1]$ of
\bea
\sum_{q,q'\in\Z^l}\intdm\bid\max\limits_{j=1\dots m}\vert p_j\vert\vert q'\vert\sum_{\gamma,\gamma'=0}^k\mu_{k-\gamma}(p,q)\vert\partial^\gamma_\hbar\widehat{\widetilde{\underline{\G^r_i}}}(p,q)\vert
\mu_{k-\gamma'}(p',q')\vert\partial^{\gamma'}_\hbar\widehat{\widetilde\W_r}(p',q')\vert e^{(\rho-\delta)(\bid|p|+\bid|p'|+|q|+|q'|)} dpdp'
\nonumber\\
=
\sum_{q'\in\Z^l}\intdm
\bid\max\limits_{j=1\dots m}\vert p_j\vert\vert q'\vert
\sum_{\gamma,\gamma'=0}^k\mu_{k-\gamma}(p,0)\vert\partial^\gamma_\hbar\widehat{\G^r_i}(p)\vert
\mu_{k-\gamma'}(p',q')\vert\partial^{\gamma'}_\hbar\widehat{\widetilde\W_r}(p',q')\vert e^{(\rho-\delta)(\bid|p|+\bid|p'|
+|q'|)} dpdp'
\nonumber
\eea

Since $\vert\widehat\G_i^r(p)\vert\vert p_j\vert=
\vert\widehat\G_i^r(p) p_j\vert=
\vert\widehat{\nabla_j\G_i^r}(p)\vert
$, 
we get that (use $\rho-\delta\leq\rho$ and again $\vert q'\vert e^{-\delta\vert q'\vert}\leq\frac1{e\delta}$)
\be
\Vert X_r\Vert_{\rho-\delta}
\leq  \frac{(1+k)8^{k}\bid}{e\delta}\Vert\nabla G^r\Vert_\rho\Vert W_r\Vert_\rho\leq\frac{Z_k\bid}{e\delta}\Vert \nabla(\B^\hbar_r-\B^\hbar_0)\Vert_\rho\Vert W_r\Vert_\rho.\label{follo}
\ee
 Here $\Vert \nabla(\B^\hbar_r-\B^\hbar_0)\Vert_\rho$ is understood in the sense of \eqref{mfoism1}-\eqref{sigomat}. 
%
%

\eqref{d1} follows form \eqref{primo} and \eqref{hacher11}-\eqref{follo}.

We will prove \eqref{emboitesH} by the same argument as in the proof of  Proposition \ref{stimeMo}. Take \eqref{encore} with $G_1:=\frac{1}{i\hbar}[W_r,H_r]$, $G_{s+1}=\frac{1}{i\hbar}[W_r,G_s]$  and $\gamma_s=\frac{d-s}{d}\delta$ for $1\leq s\leq d-1$, $\gamma_0=\delta$, $\gamma_{d-1}=\frac{\delta}{d}$.

 We get
\begin{eqnarray}
\frac{1}{d!}\Vert G_{d} \Vert_{\rho-\gamma_{0}} &\leq &\frac{Z_k^{d-1}}{d! e^{2(d-1)}\gamma_{0}\cdots\gamma_{d-2}(\frac{\delta}{d})^{d-1}} \Vert W_r \Vert_{\rho}^{d-1}\Vert G_{1}\Vert_{\rho-\gamma_{d-1}}\nonumber\\
&\leq&
\frac{Z_k^{d-1}}{d!d!e^{2(d-1)}(\frac\delta d)^{2d-2}}
\Vert W_r \Vert_{\rho}^{d-1}\Vert G_{1} \Vert_{\rho-\delta/d}\nonumber\\
&\leq&
\frac{1+Z_k\Vert \nabla(\B^\hbar_r-\B^\hbar_0)\Vert_\rho}{d!d!e^{2d-1}(\frac{\delta}{d})^{2d-1}}Z_k^{d-1}\Vert W_r\Vert_{\rho}^d\nonumber\\
&\leq&
\frac{\delta}{2\pi d^2 e^{-1}Z_k}(1+Z_k\Vert \nabla(\B^\hbar_r-\B^\hbar_0)\Vert_\rho)\left(\frac{d^de^{-d}\sqrt{2\pi d}}{d!}\right)^2\left(\frac{Z_k}{\delta^2}\right)^d\Vert W_r\Vert_{\rho}^d\nonumber
\eea
and we get \eqref{emboitesH} by $\frac{d^de^{-d}\sqrt{2\pi d}}{d!}\leq 1$ and $d^2 e^{-1}\geq 1$ if $d\geq 2$, and setting $\rho=\rho_r,\ \gamma_0=\gamma_r$.

In order to prove \eqref{estiw} we 
first estimate $\Vert \widetilde V^r\Vert_{\rho_r}$ defined by \eqref{kjh} where $A^r(i,j)$ is given by \eqref{ho3}.
\begin{lemma}\label{deuxio}
Let $V'$ be defined by $V'_{ij}=A^r(i,j)V^{co}_{ij}$. Then
\be\label{deuxio1}
\Vert V'\Vert_\rho\leq Z_k\Vert D(\B^\hbar_r-\B^\hbar_0)\Vert_\rho\Vert V^{co}\Vert_\rho.
\ee
\end{lemma}
\begin{proof}
The proof will actually be close to the one of 
\eqref{d1}. $A^r(i,j)\omega.(i-j)=\frac{\G^r(\omega.j\hbar)-\G(\omega-i\hbar)}\hbar$ so
\[
A^r(i,j)=\int_0^1\nabla\G^r((1-t)\omega.j\hbar+t\omega.i\hbar)dt.
\]
Therefore $$V'=\int\limits_0^1 \sum\limits_{n=1}^m
\nabla_n\G^r((1-t)\omega.j\hbar+t\omega.i\hbar)V^{co}_ndt$$ so
$$\Vert V'\Vert\leq \sup\limits_{0\leq t\leq 1}
\sum\limits_{n=1}^m\Vert
\nabla_n\G^r((1-t)\omega.j\hbar+t\omega.i\hbar)V^{co}_n\Vert.$$

Let $X^r_n$ be  defined through
 \[
(X^r_n)_{ij}=
\nabla_n\G^r((1-t)\omega.j\hbar+t\omega.i\hbar)(V^{co}_n)_{ij}
=\nabla_n\G^r(\omega\cdot\frac{i+j}2\hbar-(1-2t)(i-j)\frac\hbar2)(V^{co}_n)_{ij}
\]
 By the argument as before, using \eqref{stimem}, we get that the symbol of $X^r$ satisfies $\widetilde{\mathcal X^r_n}(\xi,q)=$
 \be\nonumber
 \nabla_n\G^r(\omega\cdot\xi-(1-2t)q\frac\hbar2)\widetilde{\V_n^{co}}(\xi,q)=
\widetilde{(\mathcal X^{r}_n)'}(\omega\cdot\xi,q) :=
\nabla_n\G^r(\omega\cdot\xi-(1-2t)q\frac\hbar2)\widetilde{(\V_n^{co})'}(\omega\cdot\xi,q)
 \ee
 since $\V_n^{co}$ has the same structure as $\V$ so there exists $(\V_n^{co})'$ such that $\V_n^{co}(\xi,x)=(\V_n^{co})'(\omega\cdot\xi,x)$.
 
 Taking now the Fourier transform of $\widetilde{(\mathcal X^{r}_n)'}(\Xi,q)$ with respect to $\Xi$ one gets by translation-convolution
\[
\widehat{\widetilde{(\mathcal X^r_n)'}}(p, q,\hbar)=
\intm \widehat{{\nabla_n\G^r}}(p-p')e^{i(p-p').\omega.q(1-2t)\hbar/2}\widehat{\widetilde{\mathcal V^{co}_n}}(p',q,\hbar)dp'.
\]
So, as before,
\bea
\vert\partial^\gamma_\hbar\widehat{\widetilde{(\mathcal X^r_n)'}}(p, q,\hbar)\vert&\leq&\int
\sum_{\mu=1}^\gamma\sum_{\nu=1}^{\gamma-\mu}
\vert \partial_\hbar^{\gamma-\mu-\nu}
\widehat{\nabla_n\G^r}(p-p')
\partial_\hbar^{\nu}
e^{i(p-p').\omega.q(1-2t)\hbar/2}
\partial_\hbar^{\mu}
\widehat{\widetilde{\mathcal V^{co}_n}}(p',q,\hbar)\vert\binom{\gamma}{\mu}\binom{\gamma-\mu}{\nu} dp'\nonumber\\
&\leq&
\intm
\sum_{\mu=1}^\gamma\sum_{\nu=1}^{\gamma-\mu}
\vert \partial_\hbar^{\gamma-\mu-\nu}
\widehat{\nabla_n\G^r}(p-p')
(\vert p-p'\vert\vert q\vert/2)^\nu
\partial_\hbar^{\mu}
\widehat{\widetilde{\mathcal V^{co}_n}}(p',q,\hbar)\vert\binom{\gamma}{\mu}\binom{\gamma-\mu}{\nu} dp'\nonumber
\eea
Following the same lines than in the proof of 
\eqref{d1}
 we get that (remember that, by definition, $\Vert X^r\Vert_\rho:=\sum\limits_{n=1}^m\Vert X^r_n\Vert_\rho,\ \Vert X^r_n\Vert_\rho=\Vert (\mathcal X^r_n)'\Vert_\rho$ by Definitions \ref{norm3} and \ref{norm2})
\[
\Vert X^r\Vert_\rho\leq Z_k\sum_{i=1}^m\sum_{n=1}^{m}\Vert \nabla_n\G^r_i\Vert_\rho\Vert V^{co}_n\Vert_\rho\leq Z_k\max_n\Vert \nabla_n\G^r\Vert_\rho\sum_{n=1}^m\Vert V^{co}_n\Vert_\rho\leq Z_k\Vert \nabla\G^r\Vert_\rho\Vert V^{co}\Vert_\rho
\]
and the Lemma is proved.\end{proof}
\begin{corollary}\label{tertio}
Let $V''$ defined by $V''_{ij}=(1+A^r(i,j))^{-1}V^{co}_{ij}$. Then 
\be\label{tertio1}
\Vert V''\Vert_\rho\leq \frac 1 {1-Z_k\Vert\nabla(\B^\hbar_r-\B^\hbar_0)\Vert_\rho}\Vert V^{co}\Vert_\rho.
\ee
\end{corollary}
\eqref{estiw} is now a consequence of \eqref{brunoptidiv} and the fact that $\Vert V^{co}\Vert_\rho\leq\Vert V\Vert_\rho$.

\eqref{estiwd} is obtained by putting \eqref{estiw} in \eqref{emboitesH}. 
The proposition is proved.
\end{proof}
\vskip 1cm
\vskip 1cm
We need finally the following obvious Lemma:
\begin{lemma}\label{cutof}
Define
\be\label{co}
\mathcal V^M(x,\xi):=
\sum_{|q|\geq M} 
\widetilde{\mathcal V}_q
(\xi) e^{iqx}.
\ee
Then
\be
\Vert \mathcal V^M\Vert_{\rho-\delta}\leq e^{-\delta M}\Vert \mathcal V\Vert_{\rho}
\ee
\end{lemma}
\begin{corollary}\label{corcitof}
Let $V^M$ be  defined by
\bea
V^M_{ij}&=& V_{ij}\ \ \  \mbox { when } \vert i-j\vert\geq M\nonumber\\
&=& 0\ \ \ \ \  \mbox { when } \vert i-j\vert< M.\nonumber
\eea Then 
\be\label{corcutofff}
\Vert V^M\Vert_{\rho-\delta}\leq e^{-\delta M}\Vert V\Vert_\rho.
\ee
\end{corollary}
\begin{proof}
Just notice that the symbol of $V^M$, $\V^M$, satisfies, by \eqref{stimem}, $\widetilde \V^M(\xi,q,\hbar)=0$ when $\vert q\vert\leq M$ and apply Lemma \ref{cutof}.
\end{proof}
\vskip 1cm

Let us define, for a decreasing positive sequence $(\rho_r)_{r=0\dots\infty},\ \rho_{r+1}=\rho_r-\delta_r$ to be specified later, 
\be\label{Gr}
G_r=\Vert D(\B^\hbar_r-B^\hbar_0)\Vert_{\rho_r}=
\max\limits_{i=1\dots m}
\sum\limits_{j=1\dots m}
\Vert\nabla_i
(\B^\hbar_r-B^\hbar_0)_j\Vert_{\rho_r}.
\ee
We are now in position to derive the following fundamental estimates of the five terms in \eqref{touut}: 
\be\label{moyenne}
\Vert h_{r+1}\Vert_{\rho_r-\delta_r}\leq \Vert h_{r+1}\Vert_{\rho_r}=\Vert V^{co,r}\Vert_{\rho_r}=\Vert V_{r}\Vert_{\rho_r}
\ee
\be\label{estico}
\Vert V_r-V^{co,r}\Vert_{\rho_r-\delta_r}=\Vert V^{M_r}\Vert_{\rho_r-\delta_r}\leq 
\left(\frac{e^{-M_r\delta_r}}{\Vert V^r\Vert_{\rho_r}}\right)\Vert V_r\Vert_{\rho_r}^2
\ee
\be\label{estihat}
\Vert \widehat V^r\Vert_{\rho_r-\delta_r}\leq \frac{Z_k\frac{\machin_{M_r}}{\delta_r^2}}{(1-{Z_k G_r})^2}\Vert V_r\Vert_{\rho_r}^2
\ee
\be\label{r1b}
\Vert R_1\Vert_{\rho_r-\delta_r}\leq\frac{Z_k\frac{\machin_{M_r}}{\delta_r^2(1-{Z_k G_r})}}{1-Z_k\frac{\machin_{M_r}}{\delta_r^2(1-{Z_k G_r})}\Vert V^r\Vert_{\rho_r}}\Vert V_r\Vert_{\rho_r}^2
\ee
\be\label{r2b}
\Vert R_2\Vert_{\rho_r-\delta_r}\leq
\frac{Z_k\frac{\machin_{M_r}^2{\bid(1+Z_kG_r)}}{\delta_r^3(1-{Z_k G_r})^2}}{1-Z_k\frac{\machin_{M_r}}{\delta_r^2(1-{Z_k G_r})}\Vert V^r\Vert_{\rho_r}}\Vert V_r\Vert_{\rho_r}^2
\ee


Indeed, 
\noindent $(\ref{moyenne})$ is obvious and $(\ref{estico})$ is nothing but Corollary \ref{corcitof}. 

\noindent $(\ref{estihat})$ is derived by using Proposition \ref{stimeMo}, item (5) equation \eqref{normaM2new}, Lemma \ref{tertio} and equation \eqref{kjh}. Note that, as pointed out before, $\widehat V^r$ is cut-offed as $V^{co,r}$ thanks to \eqref{eupeur}.

\noindent $(\ref{r1b})$ is obtained through the definition $R_1=\sum_{j=1}^\infty\frac1{(-i\hbar)^j j!}\text{ad}^j_{W_r}(V_r)$, the fact that, by \eqref{emboites}, $\Vert\frac1{(-i\hbar)^j j!}\text{ad}^j_{W_r}(V_r)\Vert_{\rho_r-\delta_r}\leq (Z_k/\delta_r^2)^j\Vert V_r\Vert_{\rho_r}\Vert W_r\Vert_{\rho_r}^j$ and  \eqref{estiw}, so that
\[
\Vert R_1\Vert_{\rho_r-\delta_r}\leq \sum_{j=1}^\infty
(Z_k/\delta_r^2)^j\Vert V_r\Vert_{\rho_r}
\left(\frac{\machin_{M_r}}{1-Z_kG_r}\Vert V_r\Vert_{\rho_r}\right)^j.
\]

\noindent $(\ref{r2b})$ is proven by the definition
$R_2=\sum\limits_{j=2}^\infty\frac1{(-i\hbar)^j j!}\text{ad}^j_{W_r}(H_r)$ and  the fact that, by \eqref{estiwd} we have that
$\Vert\frac1{(-i\hbar)^j j!}\text{ad}^j_{W_r}(H_r)\Vert_{\rho_r-\delta_r}\leq \bid\frac{1+Z_kG_r}{Z_k/\delta_r}\left(\frac{Z_k\machin_{M_r}/\delta_r^2}{1-Z_kG_r}\Vert V_r\Vert_{\rho_r}\right)^j$.
\vskip 1cm

\vskip 1cm
Collecting all the preceding estimates together with the definition $(\ref{vr+1})$~:
$$
V_{r+1}:= (V_r-V^{co,r})+\widehat V^r+R_1+R_2,
$$
we obtain:
\begin{proposition}\label{oufouf}
For $r=0,1,\dots$ ,
 we have 
\begin{equation}\label{induct-D}
\Vert V_{r+1}\Vert_{\rho_r-\delta_r}\leq F_r\Vert V_{r}\Vert_{\rho_r}^2+e^{-\delta_rM_r}\Vert V_{r}\Vert_{\rho_r}
\end{equation} with
\be\label{pouf}
F_r= 
\frac{\machin_{M_r}{Z_k}}{\delta_r^2(1-Z_kG_r)^2}\left(1+
\frac {(1-Z_kG_r)+\frac{\machin_{M_r}}{\delta_r}\bid(1+Z_kG_r)}{1-\frac{\machin_{M_r}}{1-Z_kG_r}\frac{Z_k}{\delta_r^2}
\Vert V_r\Vert_{\rho_r}
}
\right).
\ee
\end{proposition}
\vskip 1cm



\section{Fundamental iterative estimates: Diophantine condition case}
\textbf{In all this section also the norm subscripts $\bid$ and $k$ are omitted.}

Let $0<\delta<\rho$.
Let us recall 
that we want to find $W_r$ such that
\be\label{dthink}
e^{i\frac{W_r}\hbar}(H_r+V_r)e^{-i\frac{W_r}\hbar}=H_{r+1}+V_{r+1}\ee
where $H_{r+1}=H_r+h_{r+1}$ and $H_r=\B^\hbar_r(L_\om),\ h_{r+1}=\overline{V_r}=\mathcal D_r(L_\om)$ and
\be
 \Vert h_{r+1}\Vert_{\rho}=\Vert \overline{V_r}\Vert_{\rho}\leq\Vert V_{r}\Vert_{\rho},\ \ \ \ \Vert V_{r+1}\Vert_{\rho-\delta}\leq D_r\Vert V_{r}\Vert_{\rho}^2.
 \ee
In the case where $\omega$ satisfies the Diophantine condition \eqref{DC} we look at $W_r$ solving:
 \be\label{desticohodio}
 \frac 1{i\hbar}[H_r,W_r]+V_r=\overline{V_{r}}+\widehat V^r
 \ee
 with
  $\widehat V^r=(\widehat V^r_l)_{l=1\dots m}$ 
   given by
 \be\label{dkjh}
 (\widehat V^r_l)_{ij}=\frac{([\widetilde V^r_l, \widetilde V^r_{l(i-j)}])_{ij}}{i\hbar\omega_{l(i-j)}\cdot(i-j)},
 \ \ \ \ \ \widetilde V^r_{ij}:=\left(I+
  A^r_\ep(i,j)\right)^{-1}(V_{r})_{ij}.
 \ee
Here $A^r(i,j)$ is the matrix given by Lemma \ref{matruc}, that is:  
\be\label{ho3}
\frac{\B^\hbar_r(\hbar\omega\cdot i)-\B_r(\hbar\omega\cdot j)}{i\hbar}
=\left(I+ A^r(i,j)\right)\omega.(i-j),
\ee
\vskip 1cm
The l.h.s. of \eqref{dthink} is:
$$
H_r+V_r+\frac1{i\hbar}[H_r,W_r]+\sum_{j=1}^\infty\frac1{(-i\hbar)^j j!}\text{ad}^j_{W_r}(V_r)+\sum_{j=2}^\infty\frac1{(-i\hbar)^j j!}\text{ad}^j_{W_r}(H_r)
$$
that is
\[
H_{r}+\overline{V^{r}}+\widehat V^r
+\sum_{j=1}^\infty\frac1{(-i\hbar)^j j!}\text{ad}^j_{W_r}(V_r)+\sum_{j=2}^\infty\frac1{(-i\hbar)^j j!}\text{ad}^j_{W_r}(H_r).
\]
or
\be\label{dtouut}
H_{r}+h_{r+1}+\widehat V^r+R_1+R_2
\ee
\begin{proposition}\label{fundio}
Let $W_r$ the (scalar) solution of 
\eqref{desticohodio}.
Then, for $d\geq 2,\ 0<\delta<\rho<\infty$,
\be\label{emboitesHdio}
\frac 1{d!}\Vert[H_r,\underbrace{W_r],\dots}_{d\ times}]/(i\hbar)^d\Vert_{\rho-\delta}\leq
\frac{\delta\bid}{2\pi Z_k}(1+Z_k\Vert \nabla(\B^\hbar_r-\B^\hbar_0)\Vert_\rho)\left(\frac{Z_k}{\delta^2}\right)^d\Vert W_r\Vert_{\rho-\delta}^d
\ee
\be\label{estiwdio}
\Vert W_r\Vert_{\rho-\delta}\leq 
\frac{2^\tau\gamma(\frac\tau{e\delta})^\tau}{1-
Z_k\Vert D(\B^\hbar_r-\B^\hbar_0)\Vert_{\rho}}
\Vert V_r\Vert_{\rho},
\ee
so
\be\label{estiwddio}
\frac 1{d!}\Vert[H_r,\underbrace{W_r],\dots}_{d\ times}]/(i\hbar)^d\Vert_{\rho-\delta}\leq\bid
\frac{
1+Z_k\Vert D(\B^\hbar_r-\B^\hbar_0)\Vert_{\rho}}{2\pi Z_k/\delta}
\left(
\frac{2^\tau\gamma(\frac\tau{e\delta})^\tau}{1-
Z_k\Vert D(\B^\hbar_r-\B^\hbar_0)\Vert_{\rho}}
\Vert V_r\Vert_{\rho}
\right)^d
\ee   
(let us recall that $\machin_M$ is defined in \eqref{BC} and $\Vert D\B\Vert_\rho$ is meant for $\max\limits_{j=1\dots m}
\sum\limits_{i=1\dots m}\Vert
\nabla_i
\B_j
\Vert_\rho$).
\end{proposition}
\begin{proof}
The proof of Proposition \ref{fundio} is the same than the one of Proposition \ref{fund} done in details in Section \ref{brunofundamental}. The only minor difference is the discussion of the small denominators and is  adaptable without pain. We omit the details here.
\end{proof}
Using notation $(\ref{Gr})$, Proposition \ref{stimeMo}, 
last item \eqref{normaM2}, and Proposition \ref{fund} we can derive the following fundamental estimates of the 
four
terms in \eqref{dtouut}
\[
\Vert h_{r+1}\Vert_{\rho_r-\delta_r}\leq \Vert h_{r+1}\Vert_{\rho_r}=\Vert V_{r}\Vert_{\rho_r}
\]
\[
\Vert \widehat V^r\Vert_{\rho_r-\delta_r}\leq \frac{Z_k\frac{2^{2+\tau}\gamma(\frac\tau{e\delta_r})^\tau}{\delta_r^2}}{(1-{Z_kG_r})^2}\Vert V_r\Vert_{\rho_r}^2
\]
\[
\Vert R_1\Vert_{\rho_r-\delta_r}\leq\frac{Z_k\frac{\gamma(\frac\tau{e\delta_r})^\tau}{\delta_r^2(1-{Z_kG_r})}}{1-Z_k\frac{\gamma(\frac\tau{e\delta_r})^\tau}{\delta_r^2(1-{Z_kG_r})}\Vert V^r\Vert_{\rho_r}}\Vert V_r\Vert_{\rho_r}^2
\]
\[
\Vert R_2\Vert_{\rho_r-\delta_r}\leq
\frac{Z_k\frac{(\gamma(\frac\tau{e\delta_r})^\tau)^2{\bid(1+Z_kG_r)}}{\delta_r^3(1-{Z_kG_r})^2}}{1-Z_k\frac{\gamma(\frac\tau{e\delta_r})^\tau}{\delta_r^2(1-{Z_kG_r})}\Vert V^r\Vert_{\rho_r}}\Vert V_r\Vert_{\rho_r}^2
\]

Collecting all the preceding results 
we get:
\begin{proposition}\label{oufoufdio}For $r=0,1,\dots$,
$\Vert V_{r+1}\Vert_{\rho_r-\delta_r}\leq F'_r\Vert V_{r}\Vert_{\rho_r}^2$ with
\be\label{poufd}
F'_r= 
\frac{\gamma(\frac\tau{e\delta_r})^\tau{Z_k}}{\delta_r^2(1-Z_kG_r)^2}\left(2^{2+\tau}+
\frac {(1-Z_kG_r)+\frac\gamma{\delta_r}(\frac\tau{e\delta_r})^\tau \bid(1+Z_kG_r)}{1-\frac{\gamma(\frac\tau{e\delta_r})^\tau}{1-Z_kG_r}\frac{Z_k}{\delta_r^2}
\Vert V_r\Vert_{\rho_r}
}
\right)
\ee
\end{proposition}
\vskip 1cm

\section{Strategy of the KAM iteration}\label{kamiter}

In the case of the Diophantine condition, the strategy consists in  finding a sequences  $\delta_r$
such that, with $F'_r$  given by \eqref{poufd},
\be\label{plouf2}
\sum_{r=1}^\infty\delta_r=\delta<\infty\ \ \mbox{ and }\ \ \prod_{i=1}^{r}D_{i}^{2^{r-i}}
\leq R^{2^r},\ R>0.
\ee
Indeed when \eqref{plouf2} is satisfied and thanks to Proposition \ref{oufoufdio}, the series  $\sum\limits_{r=1}^\infty V_r$, and therefore $\sum\limits_{r=1}^\infty h_r=\sum\limits_{r=1}^\infty \overline{V_r}$ are easily shown to be convergent in $J_k(\rho-\delta,\bid)$ for $\rho>\delta$,  at the condition that
\[
\Vert V\Vert_\rok<R.
\]
This last sum is the quantum Birkhoff normal form $\mathcal B_{\infty}^{\hbar}$ of the perturbation.
Estimates on the solution of the cohomological equations provide also the existence of a limit unitary operator conjugating the original Hamiltonian to its normal form.

The case of the Brjuno condition follows the same way, except that one has also to find a sequence of numbers $M_r$ so that \eqref{pouf} holds. The main difference  comes from the extra linear and non quadratic term in Proposition \ref{oufouf}. This difficulty is overcome by  deriving out of $\Vert V\Vert_\rok$ a sequence of quantities with a  quadratic growth as in \eqref{poufd}. This leads to an extra condition for the convergence of the iteration, condition involving only the arithmetical properties of $\omega$ and  which can be removed by a scaling argument. These ideas will be implemented in the following section.
%
%
%
%

 \section{Proof of the convergence of the KAM iteration}
 \textbf{In  this section the  norm subscripts $\bid$ and $k$  might be  committed in the body of the proofs. They are nevertheless  reestablished in the main statements.}
 \vskip 0.5cm
 This section  is organized as follows: we first prove the convergence of the KAM iteration 
in the Brjuno case with a restriction on $\omega$ (Theorem \ref{voila}), restriction released in Theorem \ref{easygoing}  thanks to the scaling argument already mentioned. This proves and precises  {\bf Theorem \ref{first}}. We then prove the corresponding classical  version (Corollary \ref{coresygoing}, global Hamiltonian  version of the singular integrability of \cite{LS1}) leading to {\bf Theorem~ \ref{corfirst}} precised. We end the section by more refined results under Diophantine condition on $\omega$, Theorem \ref{voiladio}, leading to the criterion contained in {\bf Theorem \ref{cortout}}.

 \subsection{Convergence of the KAM iteration I: constraints on $\omega$}
 \begin{proposition}\label{pouf2}
 Let us fix $0<C<\eta<1,\ \rho>0$ and let us choose 
 \be\label{choix}
 \rho_0=\rho,\ \delta_r=\alpha 2^{-r},\ 0<\alpha\leq \log{2}, \ \ \mbox{ and }\ \ \ M_r=2^r.
 \ee
 For $E\geq E_0$ defined below by \eqref{eee} let us suppose:
 \be\label{out2}
%
 \sum_{r=0}^\infty \left[\frac{\vert\log{\machin_{M_r}}\vert}{2^{r-1}}-3\frac{\log\delta_r}{2^{r}}+\frac{\log{Z_kE}}{2^r}\right]=C_k<\infty
 \ee
 and, for $1\leq r\leq l$,
 \be\label{out1}
  Z_kG_r <\eta-C/r,
 \ee
 \be\label{out3}
 \frac{\delta_r^3}{\delta_{r+1}^3}e^{(\delta_{r+1}M_{r+1}-2\delta_rM_r)}>2
 \ee
 and
 \be\label{out31}
 \frac{\machin_{M_r}Z_k}{\delta_r^2}\Vert V_r\Vert_{\rho_r}<\frac 12 (1-\eta+C/r),
 \ee
 together with
 \be\label{out10}
 Z_kG_0=0,
 \ee
 \be\label{out30}
 \frac{\delta_0^3}{\delta_{1}^3}e^{(\delta_{1}M_{1}-2\delta_0M_0)}>2
 \ee
 and
 \be\label{out310}
 \frac{\machin_{M_0}Z_k}{\delta_0^2}\Vert V\Vert_{\rho_0}<\frac 12.
 \ee
 Then, for $r\geq 0$
 \be\label{out4}
\Vert V_{r+1}\Vert_{\rho_{r+1}}\leq (D_k)^{2^{r+1}}.
\ee 
where \be\label{out22}
D_k:=e^{C_k}(\Vert V\Vert_\rho+\frac{e^{-
\alpha}\alpha^3}{2\machin_1^2Z_kE}).
\ee
 
 \end{proposition}
Note that 
$G_0=0$ and that, taking \eqref{out1} for $r=1$ we get:
\be\label{etaC}
1>\eta>Z_k\Vert\nabla \overline {\V'}\Vert_\rho \mbox{ and }C<\eta-
Z_k\Vert\nabla \overline {\V'}\Vert_\rho.
\ee
Therefore we will impose the condition 
\be\label{condiv}
\Vert\nabla \overline {\V'}\Vert_\rho
<\frac
{\eta-C}{Z_k}
\ee
 \begin{proof}

 We first prove the two following Lemmas.
 
 \begin{lemma}\label{EEE}
 Under the hypothesis 
 \eqref{induct-D}, \eqref{out1}  and \eqref{out31}, and $\eta<1$,  we have that, if 
 \be\label{eee}
 E\geq 
 \frac{3\alpha+(1+\eta)\bid{\machin_1(\omega)}}{(1-\eta)^2{\machin_1(\omega)}}=:E_0
 \ee 
 then
 \begin{equation*}
 \Vert V_{r+1}\Vert_{\rho_{r+1}}\leq d_r\Vert V_{r}\Vert_{\rho_r}^2+e^{-\delta_rM_r}\Vert V_{r}\Vert_{\rho_r}\ \ \ \ \mbox{ with } d_r=\frac{\machin_{M_r}^2Z_kE}{\delta_r^3}.
 \end{equation*}
 \end{lemma}
 The proof is immediate 
 by noticing that, under 
 proposition \ref{oufouf},\  \eqref{out1} and \eqref{out31}, \eqref{pouf} gives
 that, for $r=1,\dots$,
  $$F_r\leq 
 \frac{\machin_{M_r}{Z_k}}{\delta_r^2(1-\eta+C/r)^2}\left(3+2\frac{\machin_{M_r}}{\delta_r}\bid(1+\eta-C/r{})\right)
$$ so, for $r=0,1,\dots$
 $$F_r\leq \frac{\machin_{M_r}{Z_k}}{\delta_r^2(1-\eta)^2}\left(3+2\frac{\machin_{M_r}}{\delta_r}\bid(1+\eta{})\right).
$$ 
The case $r\geq 1$ is obtained out of the preceding inequality, and the case $r=0$ comes form the fact that 
$Z_kG_0=0\leq \eta$. 

Therefore  $E$ must be $\geq 
 \frac{3+\bid(1+\eta)\frac{\machin_{M_r}}{\delta_r}(\omega)}{(1-\eta)^2\frac{\machin_{M_r}}{\delta_r}(\omega)}
\leq
 \frac{3+\bid(1+\eta)\machin_1(\omega)/\alpha}{(1-\eta)^2\machin_1(\omega)/\alpha}=\frac{3\alpha+\bid(1+\eta)\machin_1(\omega)}{(1-\eta)^2\machin_1(\omega)}$
since $
 \machin_{M_r}$is increasing with $M_r$ and the Lemma is proved.
 \begin{lemma}\label{pouf3}
 Let $\widetilde V_r=\Vert V_r\Vert_{\rho_r}+\frac{e^{-\delta_rM_r}}{2d_r}$ where $d_r=\frac{\machin_{M_r}^2Z_kE}{\delta_r^3}$, $V_r$ satisfy \eqref{pouf}  and $V_0:=V,\rho_o=\rho$. Then
 \be\label{out6}
 \widetilde V_{r+1}\leq d_r\widetilde V_{r}^2.
 \ee
 \end{lemma}
 The proof reduces to completing the square in Proposition \ref{oufouf} and noticing that $\frac{e^{-2\delta_rM_r}}{4d_r}-\frac{e^{-\delta_{r+1}M_{r+1}}}{2d_{r+1}}>0$ by \eqref{out3}, since $\machin_{M_{r+1}}\geq \machin_{M_r}$.
 The Lemma has for consequence the fact the
 \be\label{out7}
 \widetilde V_{r+1}
\leq\prod_{s=0}^r d_s^{2^{r-s}}\widetilde V_0^{2^r}
 \leq (e^{C_k}\widetilde V_0)^{2^r}.
 \ee
 This concludes the proof  of Proposition \ref{pouf2} since $\Vert V_r\Vert_{\rho_r}\leq\widetilde V_r$. 
 \end{proof}
 \begin{proposition}\label{pouf4}
 Let $\Vert V_r\Vert_{\rho_r}\leq (D_k)^{2^r}$ with $D_k< e^{-P}$ and $D_k< M$, $M$ and $P$ defined below by   \eqref{constantm} and \eqref{pppp}. Then \eqref{out1}, \eqref{out3} and \eqref{out31} hold.
 \end{proposition}
 Note that $\sum\limits_{r=0}^\infty\delta_r=2\alpha$.
\begin{proof}
{\bf \eqref{out3}}: 

it is trivial to show that \eqref{out3} is satisfied when $\alpha\leq 2\log{2}$.

\eqref{out31}: 

\eqref{out31}-\eqref{out310} are equivalent to 
\be
\frac 1 {2^r}\log{\machin_{M_r}}-\frac{ \log{\delta_r^2}}{2^r}+\frac{ \log{Z_k}}{2^r}+\log{D_k}< \frac 1 {2^r}\log{\frac 12(1-\eta+\frac Cr)}
\ee
and
\be
\log{\machin_{1}}-\log{\delta_0^2}+\log{Z_k}+\log{D_k}<\log{\frac12}
\ee
which is implied by
\bea\label{pppp0}
\log{D_k}
&<& -\sum_{r=0}^\infty\frac {\vert\log{\machin_{M_r}}\vert}{2^r}
+\inf_{r\geq0}\frac{\log{\delta_r^2}-\log{Z_k}}{2^r}
-\Delta\\
&<&-\sum_{r=0}^\infty\frac {\vert\log{\machin_{M_r}}\vert}{2^r}
-\log{Z_k}+2\log{\alpha}-\frac2e-\Delta
\eea
which is implied by
\be\label{pppp}
\log{D_k}<
-\sum_{r=0}^\infty\frac {\vert\log{\machin_{M_r}}\vert}{2^r}
-\log{Z_k}-\frac2e-\Delta:=-P
\ee
where 
\be\label{delt}
\Delta=-\inf\limits_{r\geq 1}\left\{\frac 1 {2^r}\log{\frac 12(1-\eta+\frac Cr)},\log{\frac 12}\right\}<\infty \mbox{ for } \eta<1.
\ee

Note that $\Delta>0$, $P>0$.

\eqref{out1}:

remember that 
$\B_{r+1}=\B_r+\overline V_{r+1}$,
and $\Vert \B_{r+1}\Vert_{\rho_{r+1}}\leq \Vert \B_r\Vert_{\rho_{r+1}}+\Vert\overline \V'_{r+1}\Vert_{\rho_{r+1}}\leq \Vert \B_r\Vert_{\rho_r}+\Vert\overline \V'_{r+1}\Vert_{\rho_r}$.
So $\Vert D\B_{r+1}\Vert_{\rho_{r+1}}\leq\Vert D\B_r\Vert_{\rho_{r+1}}+\Vert D\overline \V'_{r+1}\Vert_{\rho_{r+1}}$.

Moreover one has, by \eqref{mfoism},  $\Vert D\overline \V'_{r+1}\Vert_{\rho_{r+1}-\frac{\delta_r}2}
\leq\frac{\Vert \overline \V'{r+1}\Vert_{\rho_{r+1}}}{\frac{\delta_r}2e}
\leq 
 2\frac{\Vert V_{r+1}\Vert_{\rho_{r+1}}}{\delta_re}    \leq2\frac{(D_k)^{2^{r+1}}}{\delta_re}$ out of which  we conclude that
 
 $Z_kG_r<\eta-C/r\Longrightarrow Z_kG_{r+1}< \eta-C/(r+1), \forall r\geq1$, if 
\be\label{samedi}
2Z_k\frac{D_k^{2^{r+1}}}{\alpha2^{-r}e}
< 
\frac Cr-\frac C{r+1}=\frac C{r(r+1)}
\ee 
which is implied by 
 $D_k< M$ for
 \be
 \label{constantm}
 M=\inf_{r\geq 1}
 \left(\frac{\alpha2^{-r}eC}{2Z_kr(r+1)}\right)^{2^{-(r+1)}}=\left(\frac{\alpha eC}{8Z_k}\right)^{1/4}<1
 \ee
 since $\alpha\leq\log{2}$, $Z_k\geq 8$ and $C\leq 1$.
\end{proof}
\vskip 1cm
 Proposition \ref{pouf4} together with Proposition \ref{pouf2} shows clearly that 
 \be\label{impli}
 \eqref{condiv} \mbox{ and }\left[D_k< e^{-P}\mbox{ and }D_k < M\right]\Longrightarrow \Vert V_r\Vert_{\rho_r}\leq (D_k)^{2^r}
 \ee
 where $D_k$ is given by \eqref{out22} i.e. $D_k:=e^{C_k}\Vert V\Vert_\rho+e^{C_k}\frac{e^{-\delta_0M_{0}}}{2d_0}$. Note that since $M<1$ so is $D_k\leq M$ leading ot the superquadratic convergence of the sequence $(V_r)_{r=0,\dots}$. In order for $D_k$ to satisfy the two conditions of the bracket in the l.h.s. of \eqref{impli} the two terms in $D_k$ will have to both satisfy the two conditions. This remark will be the key of  the main theorem below.

%
 Let us denote by $\omega^i_j,\ j=1\dots m,\ i=1\dots l$ be the $i$th component of the vector $\omega_j$.
%
%
%
%
%
%
Let us remark that
\be\label{out51}
\machin_1(\omega)=\min_{j=1\dots m}\frac 1{\min\limits_{i=1\dots l}\vert \omega_j^i\vert} \mbox{ and } 
\frac{1}{\machin_1(\omega)}
 =\max\limits_{j=1\dots m}\min\limits_{i=1\dots l}\vert\omega^i_j\vert.
\ee
Let us denote 
\be\label{defb}
B(\omega):=\sum\limits_{r=0}^\infty\frac{\vert\log{\machin_{2^r}}\vert}{2^r}.
\ee
 We have that, by \eqref{out2} and \eqref{pppp}, 
\be\label{ccc}
C_k(\omega)=2B(\omega)-6\log\alpha+6\log{2}+2\log{(Z_kE)}.
\ee
and
\be\label{pppppp}
P(\omega)=B(\omega)+\log{Z_k}+\frac2e+\Delta.
\ee
 \begin{theorem}\label{voila}[Brjuno case]
 Let $\alpha,\rho,\eta, \mbox{ and }C $ be strictly positive constants satisfying
 
 \be\label{condi}\alpha<2\log{2},\ \ \rho>2\alpha,\ \  0<C   <\eta <1.\ 
 \ee
\vskip 0.5cm
Let us define, 
for 
$ \Delta=-\inf\limits_{r\geq 1}\frac 1 {2^r}\log{\frac 12(1-\eta+\frac Cr)} \mbox{ and }
M=\left(\frac{\alpha eC}{8Z_k}\right)^\frac14
 $,
\vskip 0.1cm
\be\label{rayon}
R_k(\omega)=
   \frac{(1-\eta)^4\machin_1(\omega)^2}{({3\alpha+(1+\eta)\bid\machin_1(\omega)})^{2}}
  \frac{\alpha^6e^{-2B(\omega)}}{2^6Z_k^2}\min{\left\{\frac {e^{-B(\omega)-\Delta}}{2^{1/e}Z_k},{M}{}\right\}}.
 \ee
 Let us suppose that, in addition to Assumptions (A1), (A2) Brjuno case and (A3), $\omega$ satisfies 
 \be\label{condomega}
 \frac{3\alpha+(1+\eta)\bid\machin_1(\omega)}{2e^\alpha(1-\eta)^2\machin_1(\omega)^3}
 \leq 
 \frac{\alpha^3 e^{-2B(\omega)}}{2^6Z_k}\min{\left\{\frac {e^{-B(\omega)-\Delta}}{2^{1/e}Z_k},{M}{}\right\}}.
\ee 
and the perturbation $V$ satisfies
 \be\label{condv}
 \Vert V\Vert_{\rho,\bid,k}
 < R_k(\omega),\ \Vert\nabla \overline {\V'}\Vert_{\rho,\bid,k}
 <\frac{\eta-C}{Z_k}.
 \ee
 \vskip 0.5cm
 Then the BNF as constructed in section \ref{start} converges in the space $\J^\dagger_k(\rho-2\alpha,\bid)$ to $\mathcal B^\hbar_\infty$ and 
\be\label{bnf}
 \Vert \mathcal B^\hbar_\infty-\mathcal B^\hbar_0\Vert_{\rho-2\alpha,\bid,k}
 =O(\Vert V\Vert_{\rho,\bid,k}^2) \mbox{ as } \Vert V\Vert_{\rho,\bid,k}\to 0.
 \ee
 That is to say that there exists a (scalar) unitary operator $U_\infty
 $ such that the family of operators $H=(H_i)_{i=1\dots m},  H_i=L_{\omega_i}+V_i,$ satisfies, $\forall\hbar\in (0,1]$, 
\be\label{unit0}
U_\infty^{-1}HU_\infty=\mathcal B^\hbar_\infty(L_\omega).
\ee
$U_\infty$ is  the limit as $r\to\infty$ of the sequence of operators $U_r=e^{i\frac{W_r}\hbar}\dots e^{i\frac{W_0}\hbar}$ constructed in Section \ref{start} and
\[
\Vert U_\infty-U_r\Vert_{\B(L^2(\T^l))}\leq  \frac{A_r}\hbar=O\left(\frac{E^{2r}}\hbar\right)\mbox{ as }r\to\infty\mbox{ for some }E<1,
\]
here $A_r$ is defined by \eqref{thiscn}.

Moreover, $U_\infty-I\in J_0^\hbar(\rho-2\alpha,\bid)$ and
\be\label{age}
\Vert U_\infty-I\Vert_{\rho-2\alpha,\bid,0}^\hbar
=O\left(\frac{\Vert V\Vert_{\rho,\bid,0}}\hbar\right)\mbox{ as } \Vert V\Vert_{\rho,\bid,0}\to 0,
\ee
and, for any operator $X$ 
for which there exists $\overline X_{k,\rho}$ such that for all $ W\in J_k(\rho,\bid)$,
\be\label{fini?}
\Vert [X,W]/i\hbar\Vert_{\rho-\delta,\bid,k}\leq\frac{Z_k}{\delta^2}
 {\overline X_{k,\rho}}\Vert W\Vert_{\rho,\bid,k},
 \ee
 $U_\infty^{-1} XU_\infty-X\in J_k(\rho-2\alpha-\delta,\bid)$ and
 \be\label{obs}
 \Vert U_\infty^{-1} XU_\infty-X\Vert_{\rho-2\alpha-\delta,\bid,k}
 \leq \frac{D}{\delta^2}
\sup\limits_{\rho-2\alpha\leq\rho'\leq\rho}\overline X_{k,\rho'}
=O\left(\frac{\Vert V\Vert_{\rho,\bid,k}}{\delta^2}\sup\limits_{\rho-2\alpha\leq\rho'\leq\rho}\overline X_{k,\rho'}\right)
 \ee
 where $D$ is given by \eqref{defDDquant}.
 
\end{theorem}
\vskip 1cm
Note that the second condition in \eqref{condv} ``touches" only the average $\overline V$  and not the full perturbation $V$. It can also be replaced for any $\rho'>\rho$ by $\Vert \overline V\Vert_{\rho'}\leq 
\Vert  V\Vert_{\rho'}\leq e\frac{\rho'-\rho}{Z_k}$ since $\Vert\nabla\overline {\V'}\Vert_{\rho-\delta}\leq \frac{\Vert \overline {V'}\Vert_{\rho}}{e\delta}$ for any $\delta>0$.


\subsection{Convergence of the KAM iteration II: general $\omega$}
Before we start the proof of theorem \ref{voila}, let us show the way of overcoming the condition \eqref{condomega}.

\noindent We first notice that multiplying the family $L_\omega+V$ by $\lambda >0$ preserves of course integrability. Moreover $\lambda(L_\omega+V)=L_{\lambda\omega}+\lambda V$.

\noindent On the other side we see easily that: 
\be\label{scal}
B(\lambda\omega)=B(\omega)-2\log\lambda,\ \machin_1(\lambda\omega)=\lambda^{-1}\machin_1(\omega)\mbox{ and therefore }
\bid\machin_1\mbox{ is invariant by scaling}.
\ee
Let us show that, for $\lambda$ large enough, \eqref{condomega} will be satisfied for $\omega_\lambda:=\lambda\omega$.
More precisely, let us define
\bea
\mu&=&\frac{\alpha+2[(1-\eta)\alpha+(1+\eta)\bid\machin_1(\omega)]}{2e^\alpha(1-\eta)^2\machin_1(\omega)^3}\frac{2^6Z_k}{\alpha^3 e^{-2B(\omega)}}\nonumber\\
\nu&=&\frac {e^{-B(\omega)-\Delta}}{2^{1/e}Z_k}\nonumber
\eea
we easily see that the following number $\lambda_0$ is uniquely defined:
\be\label{lambda0}
\lambda_0=\lambda_0(\omega):=\inf{\{\lambda>0\mbox{ such that }M\lambda-\mu\geq0\mbox { and }\nu\lambda^3-\mu\geq 0\}}=\sup{\left\{\frac\mu M,\left(\frac\mu\nu\right)^{\frac13}\right\}}.
\ee
\indent Elementary algebra leads to

\noindent{\bf Lemma.}
$\forall \omega$, $\forall \lambda\geq\lambda_o(\omega)$, \eqref{condomega} is satisfied for $\omega_\lambda:=\lambda\omega$.
 
Since the BNF of $\lambda H$ is the  BNF of $H$ multiplied by $\lambda$ we get that the latter will exist  and be convergent 
if $\lambda\Vert V\Vert_{\rho,\lambda\bid,k}\leq R_k(\lambda\omega)$ and $\lambda\Vert\nabla\overline {\V'}\Vert_{\rho,\lambda\bid,k}\leq \lambda\frac{\eta-C}{Z_k}$. we get the
\begin{theorem}\label{easygoing}
 Let $\alpha,\rho,\eta, \mbox{ and }C $ be strictly positive constants satisfying
 
 \be\label{condi}\alpha<2\log{2},\ \ \rho>2\alpha,\ \ 0<C <\eta <1. 
 \ee
\vskip 0.5cm
Let us define
$ \Delta=-\inf\limits_{r\geq 1}\frac 1 {2^r}\log{\frac 12(1-\eta+\frac Cr)},\
M=\left(\frac{\alpha eC}{8Z_k}\right)^\frac14=\inf\limits_{r\geq 1}
 \left(\frac{\alpha2^{-r}eC}{2Z_kr(r+1)}\right)^{2^{-(r+1)}}$\ and, for 
 \ \ $\lambda\geq~\lambda_0(\omega)$ given by \eqref{lambda0},
 
\vskip 0.1cm
 \be\label{rayonl}
R_{\lambda,k}(\omega)=\frac{R_k(\lambda\omega)}\lambda=\lambda
   \frac{(1-\eta)^4\machin_1(\omega)^2}{({3\alpha+(1+\eta)\bid\machin_1(\omega)})^{2}}
  \frac{\alpha^6e^{-2B(\omega)}}{2^6Z_k^2}\min{\left\{\lambda^2\frac {e^{-B(\omega)-\Delta}}{2^{1/e}Z_k},{M}{}\right\}}.
 \ee

 Let us suppose that the general assumption (A1), (A2) Brjuno case and (A3) hold and  
 \be\label{condvl}
 \Vert V\Vert_{\rho,\lambda\bid,k}< R_{\lambda,k}(\omega),\ \Vert\nabla \overline {\V'}\Vert_{\rho,\lambda\bid,k}
  \leq 
  \frac{\eta-C}{Z_k}.
 \ee
 \vskip 0.5cm
 Then the BNF as constructed in section \ref{start} converges in the space $\J^\dagger_k(\rho-2\alpha,\lambda\bid)$ to $\mathcal B^\hbar_\infty$ and 
 \be\label{bnfl}
 \Vert \mathcal B^\hbar_\infty-\mathcal B^\hbar_0\Vert_{\rho-2\alpha,\lambda\bid,k}
 =O(\Vert V\Vert_{\rho,\lambda\bid,k}^2).
 \ee
 That is to say that there exists a (scalar) unitary operator $U_\infty,\ U_\infty-I\in J_k(\rho-2\alpha,\lambda\bid)$ such that the family of operators $H=(H_i)_{i=1\dots m},  H_i=L_{\omega_i}+V_i,$ satisfies, $\forall\hbar\in (0,1]$, 
\be\label{unit0l}
U_\infty^{-1}HU_\infty=\mathcal B^\hbar_\infty(L_\omega).
\ee
$U_\infty$ is  the limit as $r\to\infty$ of the sequence of operators $U_r=e^{i\frac{W_r}\hbar}\dots e^{i\frac{W_0}\hbar}$ constructed in Section \ref{start} and
\[
\Vert U_\infty-U_r\Vert_{\B(L^2(\T^l))}\leq  \frac{A_r}\hbar=O\left(\frac{E^{2r}}\hbar\right)\mbox{ as }r\to\infty\mbox{ for some }E<1,
\]
here $A_r$ is defined by \eqref{thiscn}.

Moreover, $U_\infty-I\in J_0^\hbar(\rho-2\alpha,\lambda\bid)$ and
\be\label{age}
\Vert U_\infty-I\Vert_{\rho-2\alpha,\lambda\bid,0}^\hbar
=O\left(\frac{\Vert V\Vert_{\rho,\lambda\bid,0}}\hbar\right)\mbox{ as } \Vert V\Vert_{\rho,\lambda\bid,0}\to 0,
\ee
and, for any operator $X$ 
for which there exists $\overline X_{k,\rho,\lambda}$ such that for all $ W\in J_k(\rho,\lambda\bid)$,
\be\label{fini?}
\Vert [X,W]/i\hbar\Vert_{\rho-\delta,\lambda\bid,k}\leq\frac{Z_k}{\delta^2}
 {\overline X_{k,\rho,\lambda}}\Vert W\Vert_{\rho,\lambda\bid,k},
 \ee
$U_\infty^{-1} XU_\infty-X\in J_k(\rho-2\alpha-\delta,\lambda\bid)$ and
 \be\label{obs}
 \Vert U_\infty^{-1} XU_\infty-X\Vert_{\rho-2\alpha-\delta,\lambda\bid,k}
 \leq \frac{D}{\delta^2}
\sup\limits_{\rho-2\alpha\leq\rho'\leq\rho}\overline X_{k,\rho',\lambda}
=O\left(\frac{\Vert V\Vert_{\rho,\lambda\bid,k}}{\delta^2}\sup\limits_{\rho-2\alpha\leq\rho'\leq\rho}\overline X_{k,\rho',\lambda}\right)
 \ee
 where $D$ is given by \eqref{defDDquant}.

%
\end{theorem}

\begin{remark}\label{radiusgrand}{
Note that, for $\lambda$ large enough, $R_{\lambda,k}(\omega)=\lambda R_k(\omega)$. Therefore the radius of convergence increases by dilating $\omega$. But this fact is compensated by the fact that the norm in  the condition of convergence \eqref{condvl} (we take here $\overline\V=0$),
$
\Vert\V'\Vert_{\rho,\lambda\bid,0}<R_{\lambda,k}(\omega),
$ increases at least as $\lambda$ (an actually highly non sharp estimate as the Gaussian case shows clearly) when $\lambda$ is large,
 as shown by the following lemma. Therefore the optimization on $\lambda$ of \eqref{condvl} remains between bounded values of $\lambda$.
\begin{lemma}\label{complambda}
For $\bid'\geq\bid$,
$
\Vert \F\Vert_{\rho,\bid',k}-\Vert \F\Vert_{\rho,\bid.k}\geq(\bid'-\bid)\rho
\Vert\nabla\F\Vert_{\rho,\bid,k}.
$
\end{lemma}
The proof is an immediate consequence of 
\[
e^{\rho'X}-e^{\rho X}=e^{\rho X}(e^{(\rho'-\rho)X}-1)\geq
e^{\rho X}(\rho'-\rho)X, \ X\geq 0.
\]}
\end{remark}
\vskip 0.5cm
\subsection{Proof of Theorem \ref{voila}}
 First   notice that $\mathcal B^\hbar_r-\mathcal B^\hbar_{r-1}=\overline V_r=\widetilde V_r(\cdot ,0,\hbar)$ so $\Vert \mathcal B^\hbar_r-\mathcal B^\hbar_0\Vert_{\rho_r}\leq \sum\limits_1^r \Vert V_l\Vert_l$ which is convergent under \eqref{condomega} and \eqref{condv}. What is left is to show that the sequence of unitary operators $U_r:=e^{i\frac{W_r}\hbar}\dots e^{i\frac{W_1}\hbar}$ converges to a unitary operator on $L^2(\T^l)$. This is done by proving that the sequence $U_r$ is Cauchy ($\hbar\in (0,1]$). 
 For $p>n$ let us denote 
 \be\label{enp}
 E_{np}=e^{i\frac{W_{n+p}}\hbar}e^{i\frac{W_{n+p-1}}\hbar}\dots e^{i\frac{W_{n+1}}\hbar}-I,
 \ee
 so that $U_{n+p}-U_n=E_{np}U_n$. We have for all $r$, 
 \be\label{finira}
 e^{i\frac{W_r}\hbar}=I+ T_r \mbox{ with } T_r=i\frac{W_r}\hbar\int_0^1
 e^{it\frac{W_r}\hbar}dt.
 \ee
 Therefore
 \be\label{948}
 \hbar\Vert  T_r\Vert_{\B(L^2(\T^l))}\leq\Vert W_r\Vert_{\B(L^2(\T^l))}\leq\Vert W_r\Vert_{\rho_r,k}.
 \ee
%
 
 By \eqref{estiw} we have also that 
\bea\label{aaa}
\Vert W_r\Vert_{\rho_r}=
 \Vert W_r\Vert_{\rho_r,k}
 &\leq&\frac{\machin_{M_r}}{1-\eta+C/r}\Vert V_r\Vert_{\rho_r,k}\leq\frac{\machin_{M_r}}{1-\eta+C/r}D_k^{2^r}\ \ \ l>0\nonumber\\
 \Vert W_0\Vert_{\rho}=
 \Vert W_0\Vert_{\rho,k}&\leq& \machin_1\Vert V\Vert_{\rho,k}\label{wwww} 
 \eea
 
Note that, by the Brjuno condition, we have for all $r$ $\machin_{M_r}\leq e^{B(\omega)2^r}$ and, by the condition on $D_k$ insuring the convergence of the BNF, $D_k<e^{-B(\omega)}$, so that: 
\be\label{thisc}
\chos:=\sum_{r=1}^\infty\frac{\machin_{M_r}}{1-\eta+C/r}D_k^{2^r}\leq
\frac{(e^{B(\omega)}D_k)^{2^r}}{1-\eta+C/r} <\infty.
\ee
We also define, for $n\geq 1$,
\be\label{thiscn}
\chos_n=\sum_{r=n}^\infty\frac{\machin_{M_r}}{1-\eta+C/r}D_k^{2^r}.
\ee
Note that $\chos_n=O(E^{2^n})$ as $n\to\infty$ for $E=e^BD_k<1$ by \eqref{impli}.

By \eqref{finira} we get that
\bea
E_{np}&=&e^{i\frac{ W_{n+p}}\hbar}E_{np-1}-I+e^{i\frac{ W_{n+p}}\hbar}\nonumber\\
&=&e^{i\frac{ W_{n+p}}\hbar}E_{np-1}+T_{n+p}\nonumber\\
&=&e^{i\frac{ W_{n+p}}\hbar}e^{i\frac{ W_{n+p-2}}\hbar}E_{np-1}+e^{i\frac{ W_{n+p}}\hbar}T_{n+p-1}+T_{n+p}.\label{eeee}
\eea

By iteration we find easily  that
\[
E_{np}=\sum_{k=2}^pe^{i\frac{W_{n+p}}\hbar}\dots e^{i\frac{W_{n+p-k+1}}\hbar}T_{n+p-k}+ e^{i\frac{W_{n+p}}\hbar}T_{n+p-1}+T_{n+p}
\]

and, by unitarity of $e^{i\frac{W_r}\hbar}$ and \eqref{948},
\bea
\Vert E_{np}\Vert_{\B(L^2(\T^l))}
&\leq&\sum_{k=0}^p\Vert T_{n+k}\Vert_{\B(L^2(\T^l))}
\leq\sum_{k=0}^p\Vert T_{n+k}\Vert_{\rho_{n+k}}
\leq\sum_{k=0}^p\frac{\Vert W_{n+k}\Vert_{\rho_{n+k}}}\hbar
\nonumber\\
&\leq&
\sum_{k=0}^\infty\frac{\Vert W_{n+k}\Vert_{\rho_{n+k}}}\hbar
\leq\frac{A_n}\hbar \to 0\mbox{ as }n\to\infty\mbox{ since }A<\infty.\nonumber
\eea
So $\Vert E_{np}\Vert_{\B(L^2(\T^l))}\to 0$ as $n\to\infty$ and so does $U_{n+p}-U_n=E_{np}U_n$ by unitarity of $U_n$, and $U_n$ converges to $U_\infty$ in the operator topology. Moreover we get as a by-product of the preceding estimate that
\[
\Vert U_\infty-U_r\Vert_{\B(L^2(\T^l))}
\leq\frac{A_r}\hbar.
\]
\vskip 1cm
Since $U_\infty$ is  a perturbation of the identity which doesn't belong to any $J(\rho),\ \rho>0$, there is no hope to estimate $\Vert U_\infty\Vert_{\rho,\bid,k}$. Nevertheless, and somehow more interesting, we will estimate $U_\infty-I$ in the $\Vert\cdot\Vert_{\rho-2\alpha,\bid,0}$ topology.
In the sequel of this proof we will denote $\Vert\cdot\Vert_\rho:=\Vert\cdot\Vert_{\rho,\bid,0}$ and use the fact that $\Vert\cdot\Vert_{\rho_r}\geq\Vert\cdot\Vert_{\rho-2\alpha},\ \forall r\in\mathbb N$.

We first remark that,  for $r\geq 0$, since $\Vert\cdot\Vert^\hbar_{\rho}\leq \Vert\cdot\Vert_{\rho}$,
\[
\Vert T_r\Vert_{\rho-2\alpha}^\hbar=\Vert e^{i\frac{W_r}\hbar}-I\Vert^\hbar_{{\rho-2\alpha}}
\leq\Vert e^{i\frac{W_r}\hbar}-I\Vert^\hbar_{\rho_r}
\leq\sum_{j=1}^\infty\frac{(Z_0)^{j-1}\Vert \frac{W_r}\hbar\Vert_{\rho_r}^j}{j!}=\frac{e^{\frac{Z_0\Vert W_r\Vert_{\rho_r}}\hbar}-1}{Z_0 }
\]

We first remark also that 
\[
(I+T_{r+1})U_r=U_{r+1}.
\]
Therefore, denoting $P_r=U_r-I$,
\[
P_{r+1}=(I+T_{r+1})P_r+T_{r+1}
\]
so
\[
\Vert P_{r+1}\Vert_\rho\leq \Vert P_{r}\Vert(Z_0\Vert T_{r+1}\Vert_\rho+1)+\Vert T_{r+1}\Vert_\rho=(\Vert P_{r}\Vert_\rho+\frac 1{Z_0})\Vert(Z_0\Vert T_{r+1}\Vert_\rho+1)-\frac1{Z_0}
\]
so $\Vert P_{r+1}\Vert_\rho+\frac1{Z_0}\leq (\Vert P_{r}\Vert_\rho+\frac 1{Z_0})\Vert({Z_0}\Vert T_{r+1}\Vert_\rho+1)$
and
\bea
\Vert P_{r+1}\Vert^\hbar_{\rho-2\alpha}+\frac1{Z_0}\leq (\Vert P_{0}\Vert^\hbar_{\rho}+\frac 1{Z_0})\prod_{j=1}^{r+1}({Z_0}\Vert T_{j}\Vert^\hbar_{\rho_j}+1)
&\leq& (\Vert P_{0}\Vert^\hbar_\rho+\frac 1{Z_0})\prod_{j=1}^{r+1}
e^{\frac{\Vert W_j\Vert_{\rho_j}}\hbar}\nonumber\\
&=&e^{\sum\limits_{j=1}^{r+1}\frac{\Vert W_j\Vert_{\rho_j}}\hbar}
(\Vert P_{0}\Vert^\hbar_\rho+\frac 1{Z_0})\nonumber\\
&\leq&
e^{\frac \chos\hbar}(\Vert P_{0}\Vert^\hbar_\rho+\frac 1{Z_0})\nonumber
\eea
Therefore
\[
\Vert U_\infty-I\Vert^\hbar_{\rho-2\alpha}=
\Vert P_{\infty}\Vert^\hbar_{\rho-2\alpha}\leq e^{\frac {\chos}\hbar}\left(\frac{\machin_1\Vert V\Vert_{\rho}}\hbar+\frac1Z_0\right)-\frac1Z_0
\]
Let us note that, by construction, $A=O(\frac{D_k^2}{1-\eta})$ and that $D_k$ depends on $\eta$ through \eqref{out22}.
\begin{lemma}\label{bienplace}
$ \exists \eta=\eta(\Vert V\Vert_\rho)$ such that
\[
\frac{D_k^2}{1-\eta}=O(\Vert V\Vert_\rho)\mbox{ as }\Vert V\Vert_\rho\to 0.\]
\end{lemma}
\begin{proof}
By looking at the expression of the radius of convergence which tends to $0$ as $\eta\to 1$ we see that as $V\to 0$ one can take values of $\eta\to 1$ which makes the second term in the definition of $D_k$ of order $\Vert V\Vert_\rho$ and the ratio $\frac {D_k}{1-\eta}$ of order $\Vert V\Vert_\rho$.
\end{proof}
By application of the Lemma we find that
\[
\Vert U_\infty-I\Vert^\hbar_{\rho-2\alpha}=
\Vert P_{\infty}\Vert\leq e^{\frac {\chos}\hbar}\left(\frac{\machin_1\Vert V\Vert_{\rho}}\hbar+\frac1Z_0\right)-\frac1Z_0
=O\left(\frac{\Vert V\Vert_{\rho}}\hbar\right).
\]
which gives \eqref{age}.
\vskip 1cm

In order to prove \eqref{obs} 
we first denote $V_r=e^{i\frac{W_r}\hbar}$.

We have, actually for any operator $X$, that $V_rXV_r^{-1}- X=\sum\limits_{j=1}^\infty
 \frac1{j!} \text{ad}_{W_r}^j(X)$. 
 
 Let us  suppose now that $\exists\overline X_{\rho,k}$ such that for all $ W\in J_k(\rho)$
 \be\label{fck}
 \Vert [X,W]/i\hbar\Vert_{\rho-\delta}\leq\frac{Z_k}{\delta^2}
 {\overline X_{k,\rho}}\Vert W\Vert_{\rho}.
 \ee
 (e.g. $\overline X_{k,\rho}=\Vert X\Vert_{k,\rho}$). 

Using \eqref{emboites} we get (since $2\pi l>1)$ we get
\be\label{emboitesgen}
\Vert V_rFV_r^{-1}\Vert_{\rho-\delta}\leq \frac{\Vert F\Vert_\rho}{1-\frac{Z_k}{\delta^2}\Vert W_r\Vert_\rho},
\ee
and also 
(let us recall again  that $\rho_{r+1}=\rho_r-\delta_r,\ \rho_0=\rho$)
 \bea
 \Vert \frac1{j!\hbar^j} \text{ad}_{W_r}^j(X)\Vert_{\rho_r-\delta_r-\delta=\rho_{r+1}-\delta}
 &\leq&
 \left(\frac{Z_k}{\delta_r^2}\right)^{j-1}\Vert[ X,W_r]/i\hbar\Vert_{\rho_r-\delta}\Vert W_l\Vert^{j-1}_{\rho_r-\delta}\nonumber\\
&\leq&\left(\frac{Z_0}{\delta_r^2}\right)^{j-1}\Vert[X,{W}_r]/i\hbar\Vert_{\rho_r-\delta}\Vert W_l\Vert^{j-1}_{\rho_r} ,\nonumber
 \eea
 out of which we get 
 \bea
 \Vert V_0XV_0^{-1}-X\Vert_{\rho_{1}-\delta}&\leq&
 \frac{\Vert[X,{W}_0]/i\hbar\Vert_{\rho_0-\delta}}{1-\frac{Z_k}{\delta_0^2}\Vert{W}_0\Vert_{\rho}}.\nonumber\\
 \mbox{by }V_1V_0XV_0^{-1}V_1^{-1}-V_1XV_1^{-1}=
 V_1(V_0XV_0^{-1}-X)V_1^{-1}&\mbox{and}&\eqref{emboitesgen}\nonumber\\
 \Vert V_1V_0XV_0^{-1}V_1^{-1}-V_1XV_1^{-1}\Vert_{\rho_{2}-\delta}&\leq&
 \frac{\Vert V_0XV_0^{-1}-X\Vert_{\rho_{1}-\delta}}{1-\frac{Z_k}{\delta_1^2}\Vert{W}_1\Vert_{\rho_{1}}}\nonumber\\
 &\leq&
 \frac{\Vert[X,{W}_0]/i\hbar\Vert_{\rho_0-\delta}}{(1-\frac{Z_k}{\delta_0^2}\Vert{W}_0\Vert_{\rho})(1-\frac{Z_k}{\delta_1^2}\Vert{W}_1\Vert_{\rho_{1}})}
 \nonumber\\
 &\mbox{and}&\mbox{by iteration}\nonumber\\
 \Vert U_rXU_r^{-1}-U_rV_0^{-1}XV_0U_r^{-1}\Vert_{\rho_{r+1}-\delta}
 &\leq&
\frac{\Vert[X,{W}_0]/i\hbar\Vert_{\rho_0-\delta}}{(1-\frac{Z_k}{\delta_0^2}\Vert{W}_0\Vert_{\rho})\dots(1-\frac{Z_k}{\delta_r^2}\Vert{W}_r\Vert_{\rho_{r}})}
\nonumber\\
&\mbox{and}&\mbox{by 
}X\to V_0^{-1}XV_0\nonumber\\
\Vert U_rV_0^{-1}XV_0U_r^{-1}-U_rV_1^{-1}V_0^{-1}XV_0V_1U_r^{-1}\Vert_{\rho_{r+1}-\delta}
 &\leq&
\frac{\Vert[X,{W}_1]/i\hbar\Vert_{\rho_1-\delta}}{(1-\frac{Z_k}{\delta_1^2}\Vert{W}_1\Vert_{\rho})\dots(1-\frac{Z_k}{\delta_r^2}\Vert{W}_r\Vert_{\rho_{r}})} 
 \nonumber\\
 &\cdot&\nonumber\\
 &\cdot&\nonumber\\
 &\cdot&\nonumber\\
 \Vert V_rXV_r^{-1}-X\Vert_{\rho_{r+1}-\delta}&\leq&
 \frac{\Vert[X,{W}_r]/i\hbar\Vert_{\rho_r-\delta}}{1-\frac{Z_k}{\delta_r^2}\Vert{W}_r\Vert_{\rho_r}}.\nonumber
\eea
So that by summing the telescopic sequence we get
 \bea
\Vert U_rXU_r^{-1}- X\Vert_{\rho_{r+1}-\delta}&\leq&
 \sum_{s=0}^r\Vert[X,{W}_s]/i\hbar\Vert_{\rho_s-\delta}
\prod_{j=s}^{r}\frac{1}
{1-\frac{Z_k}{\delta_j^2}\Vert{W}_j\Vert_{\rho_j}}
 \nonumber\\
&\leq&
\sum_{s=0}^r\frac{\overline X_{k,\rho_s}}{\delta^2}Z_k\Vert W_s\Vert_{\rho_s}
e^{-\sum\limits_{j=s}^r\log{(1-\frac{Z_k}{\delta_j^2}\Vert{W}_j\Vert_{\rho_j})}}\nonumber\\
 \mbox{and since } (1-a)(1+2a)\geq 1 &\mbox{  if  }& 0<a\leq 1/2\nonumber
 \eea
 \bea
 \Vert U_rXU_r^{-1}- X\Vert_{\rho_{r+1}-\delta}
&\leq&
\sum_{s=0}^r\frac{\overline X_{k,\rho_s}}{\delta^2}Z_k\Vert W_s\Vert_{\rho_s}
e^{\sum\limits_{j=s}^r\log{(1+2\frac{Z_k}{\delta_j^2}\Vert{W}_j\Vert_{\rho_j})}}\nonumber\\
&\leq&
\sum_{s=0}^r\frac{\overline X_{k,\rho_s}}{\delta^2}Z_k\Vert W_s\Vert_{\rho_s}
e^{2\sum\limits_{j=s}^r\frac{Z_k}{\delta_j^2}\Vert{W}_j\Vert_{\rho_j}}\nonumber\\
&\leq&
\sum_{s=0}^r\frac{\overline X_{k,\rho_s}}{\delta^2}Z_k\Vert W_s\Vert_{\rho_s}
e^{2\sum\limits_{j=0}^\infty\frac{Z_k}{\delta_j^2}\Vert{W}_j\Vert_{\rho_j}}\nonumber\\
&\leq&\frac{\sup\limits_{\rho-2\alpha\leq\rho'\leq\rho}\overline X_{k,\rho'}}{\delta^2}Z_k
\sum_{s=0}^r\Vert W_s\Vert_{\rho_s}
e^{2B}\nonumber
  \label{noteX}\\
  &\leq&
  \frac{\sup\limits_{\rho-2\alpha\leq\rho'\leq\rho}\overline X_{k,\rho'}}{\delta^2}D\nonumber
\eea
with $B=\frac{Z_k}{\alpha^2}\machin_1\Vert V\Vert_{\rho}+
\sum\limits_{j=1}^\infty
\frac{Z_k2^j\machin_{M_j}}{\alpha^2(1-\eta+C/j)}D_k^{2^j}<\infty$ and
\be\label{defDDquant}
D=Z_k
(\machin_1\Vert V\Vert_\rho+A)e^{2B}=O(\Vert V\Vert_\rho)
 \ee
 by Lemma \ref{bienplace}.
 
Therefore we get, by letting $r\to\infty$ so that $\rho_r\to\rho-2\alpha$,,
\be
\Vert U_\infty^{-1} XU_\infty- X\Vert_{\rho-2\alpha-\delta}\leq \frac{D}{\delta^2}
\sup\limits_{\rho-2\alpha\leq\rho'\leq\rho}\overline X_{k,\rho'}.
\ee
\be
\infty>\Vert U_\infty^{-1} XU_\infty- X\Vert_{\rho-2\alpha-\delta}=O\left(\frac{\Vert V\Vert_{\rho}}{\delta^2}\sup\limits_{\rho-2\alpha\leq\rho'\leq\rho}\overline X_{k,\rho'}\right)
\ee

The theorem is proved.
 \vskip 1cm
 \begin{remark}\label{convdio}[Diophantine case]
 In the Diophantine case one immediately sees that
 \be
 B(\omega)\leq 
 2\log{\left[\gamma2^\tau\right]}
 \ee
 Moreover one easily sees that $R_k(\omega)$ and $R_{\lambda,k}(\omega)$, together with $\lambda_0(\omega)$, are decreasing functions of $B(\omega)$. Therefore $R_k(\omega)\geq R_k^{Dio}(\omega)$ and $R_{\lambda,k}(\omega)\geq R_{\lambda,k}^{Dio}(\omega)$ where $R_k^{Dio}(\omega)$ and $R_{\lambda,k}^{Dio}(\omega)$ are obtain by replacing $B(\omega)$ by $2\log{\left[\gamma2^\tau\right]}$ in the r.h.s. of \eqref{rayon} and \eqref{rayonl}.  It follows that  Theorem \ref{voila} (resp. Theorem \ref{easygoing}) is valid with $R_k^{{dio}}(\omega)$ in place of $R_k(\omega)$ (resp. $R_{\lambda,k}^{{dio}}(\omega)$ in place of $R_{\lambda,k}^{{}}(\omega)$).
 \end{remark}
 \vskip 1cm

\subsection{Convergence of the KAM iteration III: the classical limit} Since all the estimates are uniform in $\hbar$, 
 the methods of the present paper 
 allow to prove the following result
 \begin{corollary}\label{coresygoing}
 Let $\mathcal H
 $ a family of 
 $m\leq l$ classical 
 Hamiltonians $(\H_i)_{i=1\dots m}$ on $T^*(\T^l)$ of the form $\mathcal H(x,\xi)=~\omega\cdot\xi+\mathcal V(x,\xi)=\H^0(\omega\cdot\xi)+\V'(\omega\cdot\xi,x)$. Then, under the hypothesis on $\omega$ of Theorem \ref{easygoing} (resp. Theorem \ref{voila}) 
 and the conditions
 \bea
 \{\H_i,\H_j\}&=&0\ \ \  1\leq i,j\leq m\nonumber\\
 \Vert \mathcal V\Vert_\rho&<& \overline R_{\lambda,0}(\omega
 )\ \ \  (resp. \ <R_0(\omega
 ))\nonumber\\
 \Vert\nabla\overline{ \mathcal V'}\Vert_\rho&<&
\frac{\eta-C}{Z_0}
 \nonumber
 \eea 
 $\mathcal H$ is (globally) symplectomorphically and holomorphically conjugated to $\mathcal B^0_\infty(\omega .\xi)$: for all $\delta>0$,
 
 \noindent there exist a symplectomorphism $\Phi^{-1}_\infty$
 such that
 \[
 \H\circ\Phi^{-1}_\infty=\B^0_\infty(\H_0).
 \]
 Moreover, $\Phi^{-1}_\infty-I\in J(\rho-2\alpha-\delta)$ (in particular $\Phi^{-1}_\infty$ is holomorphic) and for any positive $\delta<\rho$, 
 \be\label{esticlass}
 \Vert\Phi^{-1}_\infty-I\Vert_{\rho-2\alpha-\delta}
 \leq\frac{\mathcal D}\delta\Vert \V\Vert_{\rho}
 \ee
 where $\mathcal D$ is given by 
 \eqref{defDDclass} below.
 
 Finally, for any function $\X$satisfying \eqref{fckc}, we have
 \be\label{obsclass}
\Vert \X\circ\Phi_\infty^{-1}- \X\Vert_{\rho-2\alpha-\delta}\leq \frac{\mathcal D}{\delta^2}
\sup\limits_{\rho-2\alpha\leq\rho'\leq\rho}\overline \X_{0,\rho'}=O\left(\frac{\Vert V\Vert_\rho}{\delta^2}\sup\limits_{\rho-2\alpha\leq\rho'\leq\rho}\overline \X_{0,\rho'}\right)
\ee
where $\overline \X_{0,\rho'}$ is defined in \eqref{fckc}.
 \end{corollary}
 \begin{proof}
 Once again the function $\B^\hbar_\infty$ is by construction uniform in $\hbar\in[0,1]$ so it has a limit $\B^0_\infty$ as $\hbar\to 0$. It is easy to get convinced that the construction of $\B^0_\infty$ is the same as the one of $\B^\hbar_\infty$ after the substitution (we use capital letters 
  for operators and calligraphic ones for their symbols at $\hbar=0$):
 \bea\nonumber
 AB&\longrightarrow &\A\times\B\\
 \frac{[A,B]}{i\hbar}&\longrightarrow & \{\A,\B\}\nonumber\\
 e^{i\frac W\hbar}&\longrightarrow &e^{\L_\W}\nonumber\\
 e^{i\frac {W_1}\hbar}e^{i\frac {W_2}\hbar}
 &\longrightarrow &e^{\L_{\W_1}}\circ e^{\L_{\W_2}}\nonumber\\
 e^{i\frac {W}\hbar}Ae^{-i\frac {W}\hbar}
 &\longrightarrow &\A\circ e^{\L_{\W}}.\nonumber
 \eea
 Here $\times$ is the usual function multiplication, $\{.,.\}$ denotes the Poisson bracket and $e^{\L_\W}$ the Hamiltonian flow at time $1$ of Hamiltonian $\W$ (Lie exponential).
 
  What is left is to prove the convergence of the sequence of flows $e^{\L_{\W_r}}\dots e^{\L_{\W_1}}:=~\Phi^r$ as $r\to\infty$. This is done by the same Cauchy argument than in the proof of Theorem \ref{voila}. 
  
  For $\Phi: T^*\T^l\to T^*\T^l$ we denote $\Vert \Phi\Vert_\rho=
  \sum\limits_{i=1}^{2l}
  \Vert \Phi_i\Vert_\rho$ where $\Phi_i$ are the components of $\Phi$ and we define $\mathcal E_{np}$ by 
\[
 \mathcal E_{np}=e^{\L_{\W_{n+p}}}\circ e^{\L_{\W_{n+p-1}}}\circ \dots \circ e^{\L_{\W{n+1}}}-I_{T^*\T^l\to T^*\T^l},
 \]
 so that 
 $
 \Phi^{n+p}-\Phi^n=\mathcal E_{np}\circ\Phi^n$. 

We will need the following

\begin{lemma}\label{passeoucasse} 
Let $\F(z,\theta)$ be analytic in $\{\vert\Im{z}\vert, \vert\Im{\theta}\vert\leq\rho\}$. Then, 
\be\label{compnormes}
\Vert \F\Vert_{L^\infty(\vert\Im{z}\vert,\vert\Im{\theta}\vert\leq \rho)}
\leq\Vert\F\Vert_{\rho}.
\ee
\end{lemma}
\begin{proof}
As in section \ref{sectionweyl}  write
\[
\vert\F(z,\theta)\vert=\vert\sum_q\int\widehat{\widetilde{\F}}(p,q)
e^{i<p,\xi>+i<q,x>}dp\vert\leq\sum_q\int
\vert\widehat{\widetilde{\F}}(p,q)\vert e^{\rho(\vert p\vert+\vert q\vert)}dp=\Vert F\Vert_\rho.
\]
\end{proof}
We will denote
\[
\Vert\cdot\Vert^\infty_\rho=\Vert \cdot\Vert_{L^\infty(\vert\Im{z}\vert,\vert\Im{\theta}\vert\leq \rho)}.
\]
\begin{proposition}\label{bamdes}
Let $H_\rho=\{(z,\theta).\ \vert\Im{z}\vert\leq\rho\mbox{ and }
\vert\Im{\theta}\vert\leq\rho\}$.

Under the hypothesis of Theorems \ref{voila} and \ref{easygoing},\ \ \ 
 $\Phi^r$ is 
 analytic $H_\rho\to H_{\rho_r}$.
 \end{proposition}
 Remember that $\rho_r=\rho-\sum\limits_{j=0}^{r-1}\delta_r, \ \delta_r=\alpha2^{-r}$.
 \begin{proof}
 Let us first remark that the ``rule" $\frac{[A,B]}{i\hbar}\longrightarrow  \{\A,\B\}$ is in fact (and of course) a Lemma.
 \begin{lemma}
 Let $F\in J^m(\rho),\ \G\in J^1(\rho)$. Then $\frac{[F,G]}{i\hbar}$ is the Weyl quantization of a function $\sigma_\hbar$ on $T^*\T^l$ and
 \[
 \lim_{\hbar\to 0}\sigma_\hbar
 =\sigma_0=
 \{\F,\G\}.
 \]
 \end{lemma}
 The proof is an easy exercise which consists (again) in computing the symbol of $\frac{[F,G]}{i\hbar}$ through its matrix elements  using Proposition \ref{corA}, after expressing these matrix elements out of the ones of $F,G$, themselves expressed through the symbols $\F,\G$ of $F,G$ thanks of formula \eqref{stimem}. The limit $\hbar\to 0$ leads naturally to the Poisson bracket. Since these techniques have been extensively used through the present article, we omit the details.
 
 Let us, by a slight abuse of notation, define again
 $\text{ad}_{\W}$ the operator $\F\mapsto \{\W,\F\}$.
 
 \noindent Being uniform in $\hbar$, the formula \eqref{emboites} taken with $k=0$ leads, for any $\F\in\J^m(\rho_r)$, to
 \[
 \frac1{j!}\Vert \text{ad}_{\W_r}^j(\F)\Vert_{\rho_r-\delta_r}\leq
 \left(\frac{Z_0}{\delta_r^2}\right)^j\Vert\F\Vert_{\rho_r}\Vert\W_r\Vert_{\rho_r}^j
 \]
 Let us denote $\varphi_r=e^{\L_{\W_r}}$ and $\text{ad}_{\W_r}(\F):=\{\W_r,\F\}$. Since $\F\circ\varphi_r^{-1}=\sum\limits_{j=0}^\infty
 \frac1{j!} \text{ad}_{\W_r}^j(\F)$ we get
 \[
 \Vert\F\circ\varphi_r^{-1}\Vert_{\rho_{r+1}}\leq
 \frac{\Vert\F\Vert_{\rho_r}}{1-\frac{Z_0}{\delta_r^2}\Vert\W_r\Vert_{\rho_r}}.
 \]
 Therefore, under the hypothesis  of Theorem \ref{easygoing} (resp. Theorem \ref{voila}), $\F\circ\varphi_r^{-1}$ is analytic in $H_{\rho_{r+1}}$ for all $\F$ analytic in $H_{\rho_r}$ and so $\varphi_r^{-1}$ maps analytically $H_{\rho_{r+1}}$ to $H_{\rho_r}$ and so $\varphi_r$ maps analytically $H_{\rho_r}$ to $H_{\rho_{r+1}}$. Writing $\Phi^r=\varphi_r\circ\varphi_{r-1}\circ\dots\circ\varphi_1$ gives the result.
 \end{proof}
 \vskip 1cm
 Let us write now for all $r$ in $\mathbb N$,
 \[
 \varphi_r
 =I+ \mathcal T_r, 
 \]
 with, as for \eqref{mfoism},
 \[
 \Vert  \mathcal T_r\Vert^\infty_{\rho-2\alpha}\leq
 \Vert  \mathcal T_r\Vert^\infty_{\rho_r-\delta_r}=
 \sum_{i=1}^{2l}
 \Vert(\mathcal T_r)_i\Vert^\infty
 _{\rho_r-\delta_r}
  \leq\Vert \nabla\W_r\Vert^\infty_{\rho_r-\delta_r}:=
  \max_i\sum_j\Vert(\nabla_j\W_r)_i\Vert^\infty_{\rho_r-\delta_r}
 \leq\frac{\Vert \W_r\Vert_{\rho_r}^{\infty}}{e\delta_r},
 \] 
 and so
 \[
 \Vert \nabla \mathcal T_r\Vert^\infty_{\rho_r-\delta_r} \leq
 \frac{\Vert\nabla \W_r\Vert_{\rho_r-\delta_r/2}^{\infty}}{e\delta_r/2}
 \leq 4\frac{\Vert \W_r\Vert_{\rho_r}^{\infty}}{e^2\delta_r^2}
 \leq \frac{\Vert \W_r\Vert_{\rho_r}^{\infty}}{\delta_r^2}.
 \]
 In analogy with \eqref{eeee} we write
 \[
 \mathcal E_{np}=\varphi_{n+p}\circ(\mathcal E_{np-1}+I)-I=\varphi_{n+p}\circ(\mathcal E_{np-1}+I)-\varphi_{n+p}+(\varphi_{n+p}-I)
 \]
 so
 \[
 \Vert\mathcal E_{np}\Vert^\infty_{\rho-2\alpha}\leq
 \Vert\nabla\varphi_{n+p}\Vert^\infty_{\rho-2\alpha}\Vert\mathcal E_{np-1}\Vert^\infty_{\rho-2\alpha}+\Vert\mathcal T_{n+p}\Vert^\infty_{\rho-2\alpha}
 \]
 and, by induction,
 \bea
\Vert\mathcal E_{np}\Vert^\infty_{\rho-2\alpha}&\leq&
 \sum_{k=0}^{p-1}\Vert\mathcal T_{n+k}\Vert^\infty_{\rho-2\alpha}\prod_{s=k}^{p-1}(\Vert\nabla\varphi_{n+s}\Vert^\infty_{\rho-2\alpha})+\Vert\mathcal T_{n+p}\Vert^\infty_{\rho-2\alpha}\nonumber\\ 
&\leq&
 \sum_{k=0}^{p-1}\Vert\mathcal T_{n+k}\Vert^\infty_{\rho-2\alpha}\prod_{s=k}^{p-1}(1+\Vert\nabla\mathcal T_{n+s}\Vert^\infty_{\rho-2\alpha})+\Vert\mathcal T_{n+p}\Vert^\infty_{\rho-2\alpha}\nonumber\\ 
 &\leq&
\sum_{k=0}^{p}\Vert\mathcal T_{n+k}\Vert^\infty_{\rho-2\alpha}\prod_{s=k}^{\infty}(1+\Vert\nabla\mathcal T_{n+s}\Vert^\infty_{\rho-2\alpha}) \nonumber\\
&\leq&
\sum_{k=0}^{p}
\frac{\Vert W_{n+k}\Vert_{\rho_{n+k}}}{\delta_{n+k}}
e^{\sum\limits_{s=0}^\infty\frac{\Vert W_{s}\Vert_{\rho_{s}}}{\delta_{s}^2}}\nonumber\\
&\leq&
\mathcal \chos_ne^{\mathcal \chos}\to 0 \mbox{ as }n\to\infty.\nonumber
\eea
Here we defined
 \be\label{thisccal}
\mathcal\chos:=\sum_{j=1}^\infty\frac{\machin_{M_l}}{\delta_l^2(1-\eta+C/l)}D_k^{2^l}<\infty,
\ee
\be\label{thiscncal}
\mathcal\chos_n=\sum_{j=n}^\infty\frac{\machin_{M_l}}{\delta_l(1-\eta+C/l)}D_k^{2^l}\to 0\mbox{ as }n\to\infty\mbox{ since }\mathcal\chos<\infty
\ee
and used \eqref{wwww}.

So $\Phi^r$ converges to $\Phi^\infty$ in the $L^\infty(H_\rho)$ topology.

The proof of \eqref{obsclass} is exactly the one of \eqref{obs} by using the dictionary expressed earlier. Since it is a by-product of the proof of \eqref{esticlass} we repeat it here.
We denote $\Vert\cdot\Vert_\rho:=\Vert\cdot\Vert_{\rho,\bid}$ and use the fact that $\Vert\cdot\Vert_{\rho_r}\geq\Vert\cdot\Vert_{\rho-2\alpha-\delta},\ \forall r\in\mathbb N$.

We have, actually for any operator $\X$, that $\mathcal X\circ\varphi_r^{-1}-\mathcal X=\sum\limits_{j=1}^\infty
 \frac1{j!} \text{ad}_{\W_r}^j(\mathcal X)$. 

Taking \eqref{fck} at $k=0$ we have
\be\label{fckc}
\Vert\{\X,\W\}\Vert_{\rho-\delta}\leq\frac{Z_0}{\delta^2}\overline\X_{0,\rho}\Vert \W\Vert_{\rho}
\ee
We have
\be\label{emboitesgenclass}
\Vert \F\circ\varphi_r^{-1}\Vert_{\rho-\delta}\leq \frac{\Vert \F\Vert_\rho}{1-\frac{Z_0}{\delta^2}\Vert \W_r\Vert_\rho},
\ee
and also (let us recall again  that $\rho_{r+1}=\rho_r-\delta_r,\ \rho_0=\rho$)
 \bea
 \Vert \frac1{j!\hbar^j} \text{ad}_{\W_r}^j(\X)\Vert_{\rho_r-\delta_r-\delta=\rho_{r+1}-\delta}
 &\leq&
 \left(\frac{Z_0}{\delta_r^2}\right)^{j-1}\Vert\{ \X,\W_r\}/i\hbar\Vert_{\rho_r-\delta}\Vert \W_r\Vert^{j-1}_{\rho_r-\delta}\nonumber\\
&\leq&\left(\frac{Z_0}{\delta_r^2}\right)^{j-1}\Vert\{\X,{\W}_r\}/i\hbar\Vert_{\rho_r-\delta}\Vert \W_r\Vert^{j-1}_{\rho_r} ,\nonumber
 \eea
 out of which we get 
 \bea
 \Vert \X\circ\varphi_0^{-1}-\X\Vert_{\rho_{1}-\delta}&\leq&
 \frac{\Vert\{\X,{\W}_0\}\Vert_{\rho_0-\delta}}{1-\frac{Z_0}{\delta_0^2}\Vert{\W}_0\Vert_{\rho}}.\nonumber\\
 \mbox{by }\X\circ\varphi_0^{-1}\circ\varphi_1^{-1}-\X\circ\varphi_1^{-1}=
 (\X\circ\varphi_0^{-1}-\X)\circ\varphi_1^{-1}&\mbox{and}&\eqref{emboitesgenclass}\nonumber\\
 \Vert \X\circ\varphi_0^{-1}\circ\varphi_1^{-1}-\X\circ\varphi_1^{-1}\Vert_{\rho_{2}-\delta}&\leq&
 \frac{\Vert \X\circ\varphi_0^{-1}-\X\Vert_{\rho_{1}-\delta}}{1-\frac{Z_0}{\delta_1^2}\Vert{\W}_1\Vert_{\rho_{1}}}\nonumber\\
 &\leq&
 \frac{\Vert\{\X,{\W}_0\}\Vert_{\rho_0-\delta}}{(1-\frac{Z_0}{\delta_0^2}\Vert{\W}_0\Vert_{\rho})(1-\frac{Z_0}{\delta_1^2}\Vert{\W}_1\Vert_{\rho_{1}})}
 \nonumber
 \\
 &\mbox{and}&\mbox{by iteration}\nonumber\\
 \Vert \X\circ\Phi_s^{-1}-\X\circ\varphi_0\circ\Phi_s^{-1}\Vert_{\rho_{s+1}-\delta}
 &\leq&
\frac{\Vert\{\X,{\W}_0\}\Vert_{\rho_0-\delta}}{(1-\frac{Z_0}{\delta_0^2}\Vert{\W}_0\Vert_{\rho})\dots(1-\frac{Z_0}{\delta_s^2}\Vert{\W}_s\Vert_{\rho_{s}})}
\nonumber
\eea
By the same argument we get
 \bea
\Vert \X\circ\varphi_0\circ\Phi_s^{-1}-\X\circ\varphi_0\circ\varphi_1\circ\Phi_s^{-1}\Vert_{\rho_{s+1}-\delta}
 &\leq&
\frac{\Vert\{\X,{\W}_1\}\Vert_{\rho_1-\delta}}{(1-\frac{Z_0}{\delta_1^2}\Vert{\W}_1\Vert_{\rho})\dots(1-\frac{Z_0}{\delta_s^2}\Vert{\W}_s\Vert_{\rho_{s}})} 
 \nonumber\\
 &\cdot&\nonumber\\
 &\cdot&\nonumber\\
 &\cdot&\nonumber\\
 \Vert \X\circ\varphi_r^{-1}-\X\Vert_{\rho_{r+1}-\delta}&\leq&
 \frac{\Vert\{\X,{\W}_r\}\Vert_{\rho_r-\delta}}{1-\frac{Z_0}{\delta_r^2}\Vert{\W}_r\Vert_{\rho_r}}.\nonumber\\
 \mbox{so that}&\mbox{by}&\mbox{summing the telescopic sequence}\nonumber\\
 &\leq&
\sum_{s=0}^r\frac{\overline X_{k,\rho_s}}{\delta^2}Z_0\Vert W_s\Vert_{\rho_s}
e^{-\sum\limits_{j=s}^r\log\left({1-\frac{Z_0}{\delta_j^2}\Vert{W}_j\Vert_{\rho_j}}\right)}\nonumber\\
 \mbox{and since } (1-a)(1+2a)\geq 1 &\mbox{if}& 0<a\leq 1/2\nonumber\\
\leq
\sum_{s=0}^r\frac{\overline X_{k,\rho_s}}{\delta^2}Z_0\Vert W_s\Vert_{\rho_s}
e^{\sum\limits_{j=s}^r\log\left({1+\frac{2Z_0}{\delta_j^2}\Vert{W}_j\Vert_{\rho_j}}\right)}
&\leq&
\sum_{s=0}^r\frac{\overline X_{k,\rho_s}}{\delta^2}Z_0\Vert W_s\Vert_{\rho_s}
e^{2\sum\limits_{j=s}^r\frac{Z_0}{\delta_j^2}\Vert{W}_j\Vert_{\rho_j}}\nonumber\\
\leq
\sum_{s=0}^r\frac{\overline X_{k,\rho_s}}{\delta^2}Z_0\Vert W_s\Vert_{\rho_s}
e^{2\sum\limits_{j=0}^\infty\frac{Z_0}{\delta_j^2}\Vert{W}_j\Vert_{\rho_j}}
&\leq&\frac{\sup\limits_{\rho-2\alpha\leq\rho'\leq\rho}\overline X_{k,\rho'}}{\delta^2}Z_0
\sum_{s=0}^r\Vert W_s\Vert_{\rho_s}
e^{2\mathcal B}\nonumber
  \label{noteX}\\
  &\leq&
  \frac{\sup\limits_{\rho-2\alpha\leq\rho'\leq\rho}\overline X_{k,\rho'}}{\delta^2}\mathcal D\nonumber
%
\eea
with $\mathcal B=\frac{Z_0}{\alpha^2}\machin_1\Vert V\Vert_{\rho}+
\sum\limits_{j=1}^\infty
\frac{Z_02^j\machin_{M_j}}{\alpha^2(1-\eta+C/j)}D_k^{2^j}<\infty$ and
\be\label{defDDclass}
\mathcal D=Z_0
(\machin_1\Vert V\Vert_\rho+A)e^{2\mathcal B}=O(\Vert V\Vert_\rho)
 \ee
 by Lemma \ref{bienplace}.
 
Therefore we get, by letting $r\to\infty$ so that $\rho_r\to\rho-2\alpha$,
\be
\Vert \X\circ\Phi_\infty^{-1}- \X\Vert_{\rho-2\alpha-\delta}\leq \frac{\mathcal D}{\delta^2}
\sup\limits_{\rho-2\alpha\leq\rho'\leq\rho}\overline \X_{k,\rho'}=O\left(\frac{\Vert V\Vert_\rho}{\delta^2}\sup\limits_{\rho-2\alpha\leq\rho'\leq\rho}\overline \X_{k,\rho'}\right).
\ee

\vskip 1cm

\noindent Let now $\mathcal X\in\{\xi_1,\dots,\xi_l,x_1,\dots,x_l\}$, $\{\mathcal X,\W_1\}=\pm\partial_\Xi\W_1$ where $\Xi$ is the conjugate quantity to $\mathcal X$. 
Therefore $
\Vert\{\mathcal X,\W_0\}\Vert_{\rho-\delta}
\leq\Vert\nabla\W_0\Vert_{\rho-\delta}\leq \frac1\delta\Vert\W_0\Vert_{\rho}
$. 
Therefore $\overline\X_{k,\rho}=1$ and $\Vert\mathcal X\circ\Phi_r^{-1}-\mathcal X\Vert_{\rho_{r+1}-\delta}\leq \frac{ D}{\delta^2},\ \forall \mathcal X\in\{\xi_1,\dots,\xi_l,x_1,\dots,x_l\}$ which means that
\be
\Vert\Phi_r^{-1}-I\Vert_{\rho_{r+1}-\delta}\leq \frac{ \mathcal D}{\delta^2}
\ee

In fact we just proved  that  $\Phi^{-1}_\infty=I+\widetilde\Phi$ with $\Vert\widetilde\Phi\Vert_{\rho_\infty-\delta=\rho-2\alpha-\delta}\leq \frac {\mathcal D}{\delta^2}$. Corollary \ref{coresygoing} is proved.
 \end{proof}

\vskip 1cm

\subsection{Convergence of the KAM iteration IV: the Diophantine case}
Though the Diophantime case is covered by the  Theorem \ref{voila} (see Remark \ref{convdio}), we can also use directly Proposition \ref{oufoufdio} in order to low down the hypothesis of the Theorem.

In fact Proposition \ref{oufoufdio} shows that the same proof will be possible by only replacing $\machin_{M_r}(\omega)$ by $\gamma(\frac\tau{e\delta_r})^\tau$  and \eqref{eee} by
\be\label{eeedio}
E\geq 
 \frac{2^{2+\tau}\alpha+2[\alpha+(1+\eta)\bid\gamma(\frac\tau{e\alpha})^\tau]}{(1-\eta)^2\gamma(\frac\tau{e\alpha})^\tau}=E_1
\ee
Indeed \eqref{poufd} is verbatim the same as  \eqref{pouf} after replacing $\machin_{M_r}(\omega)$ by $\gamma(\frac\tau{e\delta_r})^\tau$ and the first term in the parenthesis, namely $1$, by $2^{2+\tau}$. Therefore the proof will be the same by replacing
$B(\omega)$ by $B_\alpha(\gamma,\tau)$ 
\[B_\alpha(\gamma,\tau)=\sum_{r=0}^\infty \log{(\gamma(\frac\tau{e\delta_r})^\tau)} 2^{-r}=2
\log{\left[\gamma(\frac\tau{e\alpha})^\tau\right]}+2\tau\log{2}
=2\log{\left[2^\tau\gamma(\frac\tau{e\alpha})^\tau\right]
}\]
and of course
 $C_k$ and $P$ by the corresponding expressions $C_k',P'$. 
 
 The very last change will concern $D_k$ which now will be $D_k=e^{C_k}\Vert V\Vert_\rho$ because the estimate of Proposition \ref{oufoufdio} reads now directly$\Vert V^{r+1}\Vert_{\rho_l-\delta_l}\leq F'_r\Vert V^{r}\Vert_{\rho_l}^2$: this will imply that  in the Diophantine case there is no condition for $\omega$ similar to \eqref{condomega}, and no condition $\alpha<2\log{2}$.
 We get:
 
\begin{theorem}\label{voiladio}[Diophantine case]
 Let $\alpha,\rho,\eta, C \mbox{ and }E$ be strictly positive constants satisfying
 
\be\label{condidio}
\rho>2\alpha,\ 0<C <\eta <1.\ 
\ee
\vskip 0.5cm
Let us define, 
for 
$ \Delta=-\inf\limits_{r\geq 1}\frac 1 {2^r}\log{\frac 12(1-\eta+\frac Cr)} \mbox{ and }
M=\left(\frac{\alpha eC}{8Z_k}\right)^\frac14
 $,
\vskip 0.1cm
 \be\label{rayondio}
 R_k(\omega)=
%
\left(\frac{(1-\eta)^2\gamma(\frac\tau{e\alpha})^\tau}{2^{2+\tau}\alpha+2[\alpha+(1+\eta)\bid\gamma(\frac\tau{e\alpha})^\tau]}\right)^2\frac{\alpha^6}{2^6Z_k^2(2^{\tau}\gamma(\frac\tau{e\alpha})^\tau)^{4}}
\min
{\left\{\frac {(2^\tau\gamma(\frac\tau{e\alpha})^\tau)^{-2}e^{-\Delta}}{2^{1/e}Z_k},M\right\}}.
 \ee
 Then if
 \be\label{condvdio}
 \Vert V\Vert_\rok< R_k(\omega),\ \Vert\nabla \overline{ \V'}\Vert_\rok   <\frac{\eta-C}{Z_k},
 \ee
 \vskip 0.5cm
  the same conclusions as in Theorem \ref{voila} and Corollary \ref{coresygoing}  hold.
 \end{theorem}
 \vskip 1cm
 \subsection{Bound on the Brjuno constant insuring integrability}
 As mentioned in the introduction we can use Theorem \ref{easygoing} to estimate the rate of divergence of the Brjuno constant as the system remains  integrable while the perturbation is vanishing.

 Let us suppose that we let $\omega$ vary in a way such that $\bido$   remain
in a bounded set $[\bid_-,\bid_+]$ of $(0,+\infty)$.
 \eqref{rayondio} tells us that, in order that Theorem \ref{voiladio} holds, $\omega$ can be taken as we want as soon as as $\gamma$ and $ \tau$ satisfy $\Vert V\Vert_\rok<\mbox{r.h.s. of }\eqref{rayondio}$ and $\Vert\nabla\overline \V\Vert_\rok<\frac{\eta-C}{Z_k}$. It is easy to check that, since $B_\alpha(\gamma,\tau):=2\log{\left[2^\tau\gamma(\frac\tau{e\alpha})^\tau\right]}\to\infty$ as $\gamma$, $ \tau$ or both of them diverge, we have, for $B_\alpha(\gamma,\tau)$ large enough (in order that the $\min$ in \eqref{rayondio} is reached by the first term and that $2^{2+\tau}<B_\alpha(\gamma,\tau)$),
\[
R_k(\omega)\geq 2Ke^{-3B_\alpha(\gamma,\tau)}
\]
with $K=\frac{(1-\eta)^4\alpha^6}{(\alpha+2(1+\eta)\bid_-)2^62^{1/e}Z_k^3}$. 
Therefore for $\Vert\nabla\overline {\V'}\Vert_\rok<\frac{\eta-C}{Z_k}$ and $\Vert V\Vert_\rok$ small enough (namely $\Vert V\Vert_\rok\leq 2Ke^{-3B_\alpha^-}$ where $B_\alpha^-$ is the smallest value of $B_\alpha(\gamma,\tau)$ which makes the min in \eqref{rayondio} reached by the first term and which is larger than $2^{2+\tau}$), we have

\begin{corollary}\label{corder}
The conclusions of Theorem \ref{voiladio} hold as soon as
\[
B_\alpha(\gamma,\tau)<\frac 13 \log{\left(\frac1 {\Vert V\Vert_{\rho,\bid_+,k}}\right)}+\frac12\log{2K}.
\]
\end{corollary}
\vskip 1cm
\noindent{\bf Remark}. In the case of the Brjuno condition, Theorem \ref{voila}, it happens that $\lambda_0(\omega)\sim C' e^ {2B(\omega)}$ and 
$R_{\lambda_0(\omega)}\sim C$ as $B(\omega)\to\infty$ for some bounded constants $C,C'$. Therefore our condition of convergence takes the form $\Vert V\Vert_{\rho,C'e^{2B(\omega)}\bid,k}<C$. This leads to a sufficient condition on $B(\omega)$ depending on the way $V\to 0$. For example it is easy to check that, if $V\to 0$ as $V=\epsilon V_0$, $\epsilon\to 0$ and $V_0$ with a symbol $\V'_0$ whose  Fourier transform in $\xi$ is compactly supported, one gets a condition of the form $B(\omega)<D \log{\log{\frac1
\epsilon
}}+D'$ for some constants $D,\ D'$.
 \section{The case $m=l$}\label{secextrem}
 \begin{lemma}\label{extrem}
 Let the vectors $\omega_j,\  j=1\dots m=l$, be independent over $\mathbb R$ and let $\Omega$ the  matrix of matrix elements 
$\Omega=(\Omega_{ij})_{i,j=1\dots l}$ with $\Omega_{ij}:=\omega_i^j$. Then
 \begin{enumerate}
 \item any $V$ satisfies \eqref{Aiweyl}
 \item 
$\forall q \in\Z^l,\ q \neq 0,\ \min\limits_{1\leq i\leq m}|\la\om_i,q\ra|^{-1}\leq l/|\Omega| $ 
 (there is no small denominator).
 \end{enumerate}
 \end{lemma}
 \begin{proof}
$\Omega$ is invertible by the independence of the $\omega_j$s. This proves (1). Moreover one has immediately that
 $1\leq\vert q\vert\leq l\vert\Omega^{-1}\vert\max\limits_{j=1\dots l}\vert
 \la\om_j,q\ra
 \vert$.
 \end{proof}
 
 Therefore the main assumption reduces to:

\centerline{\textbf{Main assumptions (extreme case)}}
\be\label{com1}
\omega_j\in\mathbb R^l,\  j=1\dots 
l,\mbox{ are  independent over }\mathbb R \mbox{ and }[H^{}_i,H^{}_j]=0,\ \forall 1\leq i, j\leq 
l.
\ee


\vskip 1.0cm\noindent
\end{document}